\PassOptionsToPackage{nameinlink}{cleveref}
\documentclass[a4paper,UKenglish,cleveref,autoref,thm-restate]{lipics-v2021}

\pdfoutput=1
\hideLIPIcs

\title{Just Verification of Mutual Exclusion Algorithms}

\author{Rob {van Glabbeek}}{School of Informatics, University of Edinburgh, UK \and School of Computer Science and Engineering, University of New South Wales, Sydney, Australia \and \url{https://theory.stanford.edu/~rvg/}}{rvg@cs.stanford.edu}{https://orcid.org/0000-0003-4712-7423}{Supported by Royal Society Wolfson Fellowship RSWF\textbackslash R1\textbackslash 221008}
\author{Bas Luttik}{Eindhoven University of Technology, The Netherlands \and \url{https://www.win.tue.nl/luttik/}}{s.p.luttik@tue.nl}{https://orcid.org/0000-0001-6710-8436}{}
\author{Myrthe S.C. Spronck\def\thefootnote{\,$*$\!}\thanks{Corresponding author}}{Eindhoven University of Technology, The Netherlands}{m.s.c.spronck@tue.nl}{https://orcid.org/0000-0003-2909-7515}{}

\authorrunning{R.J. van Glabbeek, B. Luttik and M.S.C. Spronck}
\Copyright{Rob van Glabbeek, Bas Luttik and Myrthe S.C. Spronck}



\relatedversion{An abbreviated version of this paper will appear in Proc.~CONCUR'25, Leibniz International Proceedings in Informatics (LIPIcs), Schloss Dagstuhl---Leibniz-Zentrum f\"ur Informatik.}

\supplement{\href{https://github.com/mCRL2org/mCRL2/tree/master/examples/academic/non-atomic_registers}{The mCRL2 models and modal $\mu$-calculus formulae can be found in the mCRL2 GitHub repository (commit ff6122b).}}

\nolinenumbers

\usepackage{silence}   
\usepackage{centernot}
\usepackage{tikz}
\usetikzlibrary{arrows,automata,shapes,decorations,calc,patterns,positioning,math}
\tikzmath{
    \distX=40;
    \distY=28;
    \lineGap=(0.25*\distX) - \distX;
    \lineCut=(0.5*\distX);
    \barX=6;
    \barY=3;
    \textY=8;
    \shiftX=15;
    \shiftXbig=2*\shiftX;
}

\usepackage{algcompatible}
\usepackage[noend]{algpseudocode}
\usepackage{algorithm}
\usepackage{longtable, booktabs}
\newcommand{\lFor}[2]{%
    \State\algorithmicfor\ {#1}\ \algorithmicdo\ {#2}%
  }
\newcommand{\lWhile}[2]{%
    \State\algorithmicwhile\ {#1}\ \algorithmicdo\ {#2}%
  }
\newcommand{\varname}[1]{\ensuremath{\mathit{#1}}}
\newcommand{\varidx}[2]{\ensuremath{\varname{#1}[#2]}}
\newcommand{\writeop}{\ensuremath{\gets}}
\newcommand{\defeq}{\stackrel{\text{def}}{\equiv}}
\newcommand{\VBar}[2]{($ (0,#1) + #2 $) -- ($(0,-#1)+ #2 $)}
\newcommand{\Above}[2]{($ (0,#1) + #2 $)}
\newcommand{\Below}[2]{($ (0,-#1) + #2 $)}
\newcommand{\concsym}{\smile^{\hspace{-.5ex}\raisebox{-.2ex}{\tiny$\bullet$}}}
\newcommand{\nconcsym}{{\centernot\smile}^{\hspace{-.5ex}\raisebox{-.4ex}{\tiny$\bullet$}}}
\newcommand{\conc}{\ensuremath{\mathbin{\concsym}}}
\newcommand{\nconc}{\ensuremath{\mathbin{\nconcsym}}}
\newcommand{\comp}[1]{\ensuremath{\overline{\mathit{#1}}}}
\newcommand{\co}{\ensuremath{\cdot}}
\newcommand{\states}{\ensuremath{\mathcal{S}}}
\newcommand{\initstate}{\ensuremath{\mathit{init}}}
\newcommand{\actionset}{\ensuremath{\mathit{Act}}}
\newcommand{\transrel}{\ensuremath{\mathit{Trans}}}

\newcommand{\thrlocacts}{\ensuremath{\mathit{TLoc}}}
\newcommand{\block}{\ensuremath{\mathcal{B}}}
\newcommand{\nonblock}{\ensuremath{\comp{\block}}}

\newcommand{\elimf}[2]{\ensuremath{\#_{#1}\{#2\}}}
\newcommand{\thrsym}{\ensuremath{\mathit{thr}}}
\newcommand{\thrmap}[1]{\ensuremath{\thrsym(#1)}}
\newcommand{\regsym}{\ensuremath{\mathit{reg}}}
\newcommand{\regmap}[1]{\ensuremath{\regsym(#1)}}
\newcommand{\justact}[2]{\ensuremath{\mathit{JA}^{#1}_{#2}}}
\newcommand{\truevalsym}{\ensuremath{\mathit{stor}}}
\newcommand{\trueval}[1]{\ensuremath{\truevalsym(#1)}}
\newcommand{\safRep}{\ensuremath{\mathit{saf}}}
\newcommand{\regRep}{\ensuremath{\mathit{reg}}}
\newcommand{\atoRep}{\ensuremath{\mathit{ato}}}
\newcommand{\TID}{\ensuremath{\mathbb{T}}}
\newcommand{\RID}{\ensuremath{\mathbb{R}}}
\newcommand{\Data}[1][\rid]{\ensuremath{\mathbb{D}_{#1}}}
\newcommand{\startread}[1][i,x]{\ensuremath{\mathit{sr_{#1}}}}
\newcommand{\finishread}[2][i,x]{\ensuremath{\mathit{fr_{#1}(#2)}}}
\newcommand{\startwrite}[2][i,x]{\ensuremath{\mathit{sw_{#1}(#2)}}}
\newcommand{\finishwrite}[1][i,x]{\ensuremath{\mathit{fw_{#1}}}}
\newcommand{\orderwrite}[1][i,x]{\ensuremath{\mathit{ow_{#1}}}}
\newcommand{\orderread}[1][i,x]{\ensuremath{\mathit{or_{#1}}}}
\newcommand{\overlapsym}{\ensuremath{\mathit{ovrl}}}
\newcommand{\overlap}[1]{\ensuremath{\overlapsym(#1)}}
\newcommand{\posvalsym}{\ensuremath{\mathit{posv}}}
\newcommand{\posval}[1]{\ensuremath{\posvalsym(#1)}}

\newcommand{\readerssym}{\ensuremath{\mathit{rds}}}
\newcommand{\readers}[1]{\ensuremath{\readerssym(#1)}}
\newcommand{\writerssym}{\ensuremath{\mathit{wrts}}}
\newcommand{\writers}[1]{\ensuremath{\writerssym(#1)}}
\newcommand{\pendingsym}{\ensuremath{\mathit{pend}}}
\newcommand{\pending}[1]{\ensuremath{\pendingsym(#1)}}
\newcommand{\usrsym}{\ensuremath{\mathit{usr}}}
\newcommand{\usr}[1]{\ensuremath{\usrsym(#1)}}
\newcommand{\ufrsym}{\ensuremath{\mathit{ufr}}}
\newcommand{\ufr}[1]{\ensuremath{\ufrsym(#1)}}
\newcommand{\uswsym}{\ensuremath{\mathit{usw}}}
\newcommand{\usw}[1]{\ensuremath{\uswsym(#1)}}
\newcommand{\ufwsym}{\ensuremath{\mathit{ufw}}}
\newcommand{\ufw}[1]{\ensuremath{\ufwsym(#1)}}
\newcommand{\uowsym}{\ensuremath{\mathit{uow}}}
\newcommand{\uow}[1]{\ensuremath{\uowsym(#1)}}
\newcommand{\uorsym}{\ensuremath{\mathit{uor}}}
\newcommand{\uor}[1]{\ensuremath{\uorsym(#1)}}

\newcommand{\valssym}{\ensuremath{\mathit{rec}}}
\newcommand{\vals}[1]{\ensuremath{\valssym(#1)}}
\newcommand{\StatusAll}{\ensuremath{\mathbb{S}}}

\newcommand{\Reg}[1][m]{\ensuremath{\mathit{Reg}_{#1}}}
\newcommand{\SReg}{\ensuremath{\Reg[\safRep]}}
\newcommand{\RReg}{\ensuremath{\Reg[\regRep]}}
\newcommand{\AReg}{\ensuremath{\Reg[\atoRep]}}
\newcommand{\undefsymb}{\ensuremath{\bot}}
\newcommand{\false}{\ensuremath{\mathit{false}}}
\newcommand{\true}{\ensuremath{\mathit{true}}}
\newcommand{\critsym}{\ensuremath{\mathit{c}}}
\newcommand{\crit}[1][i]{\ensuremath{\critsym_{#1}}}
\newcommand{\noncritsym}{\ensuremath{\mathit{nc}}}
\newcommand{\noncrit}[1][i]{\ensuremath{\noncritsym_{#1}}}
\newcommand{\issrsym}{\ensuremath{\mathit{sr}?}}
\newcommand{\issr}[1]{\ensuremath{\issrsym(#1)}}
\newcommand{\isswsym}{\ensuremath{\mathit{sw}?}}
\newcommand{\issw}[1]{\ensuremath{\isswsym(#1)}}
\newcommand{\idx}[1]{\ensuremath{\#({#1})}}
\newcommand{\tid}{\ensuremath{t}}
\newcommand{\tidtwo}{\ensuremath{t'}}
\newcommand{\rid}{\ensuremath{r}}
\newcommand{\data}{\ensuremath{d}}
\newcommand{\regtype}{\ensuremath{\gamma}}
\newcommand{\parcomp}{\ensuremath{\parallel}}
\newcommand{\then}{\ensuremath{\rightarrow}}
\newcommand{\satnone}{\ensuremath{\mathrm{X}}}
\newcommand{\satmutex}{\ensuremath{\mathrm{M}}}
\newcommand{\satdf}{\ensuremath{\mathrm{D}}}
\newcommand{\satsf}{\ensuremath{\mathrm{S}}}
\newcommand{\tp}{\ensuremath{\mathit{tt}}}

\newcommand{\allact}{\ensuremath{\actionset}}
\newcommand{\diam}[1]{\ensuremath{\langle \mathit{#1} \rangle}}
\newcommand{\boxm}[1]{\ensuremath{[ \mathit{#1} ]}}
\newcommand{\clos}[1]{\ensuremath{\mathit{#1}^\star}}
\newcommand{\imps}{\ensuremath{\Rightarrow}}
\newcommand{\trace}{execution}
\newcommand{\Atrace}{An execution}
\newcommand{\WCT}{\textit{WCT}}
\newcommand{\norm}[1]{\ensuremath{\mathit{norm}(#1)}}
\DeclareMathAlphabet{\mathbbm}{U}{bbm}{m}{n}            
\definecolor{highlightColour}{named}{orange}
\newcommand{\colourname}{orange}
\begin{document}

\renewcommand{\subsectionautorefname}{Section}

\maketitle
\setcounter{footnote}{0}

\begin{abstract}
  We verify the correctness of a variety of mutual exclusion algorithms through model checking.
    We look at algorithms where communication is via shared read/write registers, where those registers can be atomic or non-atomic.
    For the verification of liveness properties, it is necessary to assume a completeness criterion to eliminate spurious counterexamples. 
    We use justness as completeness criterion.
    Justness depends on a concurrency relation; we consider several such relations, modelling different assumptions on the working of the shared registers.
    We present {\trace}s demonstrating the violation of correctness properties by several algorithms, and in some cases suggest improvements.
\end{abstract}

\section{Introduction}\label{sec:introduction}
\vspace{-3pt}
The mutual exclusion problem is a fundamental problem in concurrent programming.
Given $N\geq 2$ \emph{threads},\footnote{What we call threads are in the literature frequently referred to as \emph{processes} or \emph{computers}. We use \emph{threads} to distinguish between the real systems and our models of them, expressed in a process algebra.} each of which may occasionally wish to access a \emph{critical section}, a \emph{mutual exclusion algorithm} seeks to ensure that at most one thread accesses its critical section at any given time. Ideally, this is done in such a way that whenever a thread wishes to access its critical section, it eventually succeeds in doing so.
Many mutual exclusion algorithms have been proposed in the literature, and in general their correctness depends on assumptions one can make on the environment in which these algorithms will be running. The present paper aims to make these assumptions explicit, and to verify the correctness of some of the most popular mutual exclusion algorithms as a function of these assumptions.

\vspace{-5pt}
\subparagraph{Correctness properties of mutual exclusion algorithms.}

A thread that does not seek to execute its critical section is said to be executing its \emph{non-critical section}. We regard \emph{leaving the non-critical section} as getting the desire to enter the critical section. After this happens, the thread is executing its \emph{entry protocol},\pagebreak[3] the part of the mutual exclusion algorithm in which it negotiates with other threads who gets to enter the critical section first. The critical section occurs right after the entry protocol, and is followed by an \emph{exit protocol}, after which the thread returns to its non-critical section.
When in its non-critical section, a thread is not expected to communicate with the other threads in any way. Moreover, a thread may choose to remain in its non-critical section forever after. However, once a thread gains access to its critical section, it must leave it within a finite time, so as to make space for other threads.

The most crucial correctness property of a mutual exclusion algorithm is \emph{mutual exclusion}: at any given time, at most one thread will be in its critical section.
This is a safety property. In addition, a hierarchy of liveness properties have been considered. The weakest one is \emph{deadlock freedom}: 
Whenever at least one thread is running its entry protocol, eventually some thread will enter its critical section. This need not be one of the threads that was observed to be in its entry protocol.
A stronger property is \emph{starvation freedom}: whenever a thread leaves its non-critical section, it will eventually enter its critical section.
A yet stronger property, called \emph{bounded bypass}, augments starvation freedom with a bound on the number of times other threads can gain access to the critical section before any given thread in its entry protocol.

In this paper we check for over a dozen mutual exclusion protocols, and for six possible assumptions on the environment in which they are running, whether they satisfy mutual exclusion, deadlock freedom and starvation freedom.
We will not investigate bounded bypass, nor other desirable properties of mutual exclusion protocols, such as \emph{first-come-first-served}, \emph{shutdown safety},  \emph{abortion safety}, \emph{fail safety} and \emph{self-stabilisation} \cite{Lamport86Mutex2}.

\vspace{-3pt}
\subparagraph{Memory models.}\hspace{-10pt}\footnote{A \href{https://en.wikipedia.org/wiki/Memory_model_(programming)}{\emph{memory model}} describes the interactions of threads through memory and their shared use of the data. 
The models reviewed here differ in the degree in which different register accesses exclude each other, and in what values a register may return in case of overlapping reads and writes. In this paper, we do not consider \emph{weak memory models}, that allow for compiler optimisations, and for reads to sometimes fetch values that were already changed by another thread. In \cite{AttiyaGHKMV11} is has been shown that mutual exclusion cannot be realised in weak memory models, unless those models come with \emph{memory fences} or \emph{barriers} that can be used to undermine their weak nature.
}\hspace{6pt}
In the mutual exclusion algorithms considered here, the threads communicate with each other solely by reading from and writing to shared registers. The main assumptions on the environment in which mutual exclusion algorithms will be running concern these registers.
It is frequently assumed that (read and write) operations on registers are ``undividable'', meaning that they cannot overlap or interleave each other: if two threads attempt to perform an operation on the same register at the same time, one operation will be performed before the other. This assumption, sometimes referred to as \emph{atomicity}, is explicitly made in Dijkstra's first paper on mutual exclusion \cite{dijkstra65}.
Atomicity is sometimes conceptualised as operations occurring at a single moment in time. We instead acknowledge that operations have duration. 
Consequently, if operations cannot overlap in time, then, when multiple operations are attempted simultaneously, the one performed first must postpone the occurrence of the others by at least its own duration.
One operation postponing another is called \emph{blocking} \cite{CDV09}.

Deviating from Dijkstra's original presentation, several authors have considered a variation of the mutual exclusion problem where the atomicity assumption is dropped \cite{Lamport74,peterson1983new,Lamport86Mutex1,Lamport86Mutex2,Szy88,Szy90,anderson1993fine,aravind2010yet}.
Attempted operations can then occur immediately, without blocking each other. We say these operations are \emph{non-blocking}. In this context, read and write operations may be \emph{concurrent}, i.e. overlap in time. 
We must then consider the consequences of operations overlapping each other.

In \cite{Lamport86IPCbasic,Lamport86IPCalg},\pagebreak[3] Lamport proposes a hierarchy of three memory models in this context, specifically for single-writer multi-reader (SWMR) registers; such registers are owned by one thread, and only that thread is capable of writing to it.
Crucial for these definitions is the assumption that every register has a domain, and a read of that register always yields a value from that domain. It is also important to note that threads can only perform a single operation at a time, meaning that a thread's operations can never overlap each other.
\begin{itemize}
\item A \textbf{safe} register guarantees merely that a read that is not concurrent with any write returns the most recently written value.
    \item A \textbf{regular} register guarantees that any read returns either the last value written before it started, or the value of any overlapping write, if there is one. 
    \item An \textbf{atomic} register guarantees that reads and writes behave as though they occur in some total order. This total order must comply with the real-time ordering of the operations: if operation $a$ ends before operation $b$ begins, then $a$ must be ordered before $b$.
\end{itemize}
In \autoref{sec:non-atomic-communication} we illustrate the differences between these memory models with an example. They form a hierarchy, in the sense that any atomic register is regular, and any regular one is safe. 
When we merely know that a register is safe,
a read that overlaps with any write might return any value in the domain of the register. In \autoref{sec:registers summary} we discuss the generalisation of these memory models to multi-writer multi-reader (MWMR) registers, ones that can be written and read by all threads.

Besides blocking and non-blocking registers, as explained above, we consider two intermediate memory models. The \emph{blocking model with concurrent reads} requires (1) any scheduled read or write to await the completion of any write that is in progress, and (2) any scheduled write to await the completion of any unfinished read. However, reads from different threads need not wait for each other and may overlap in time without ill effects. In the model of \emph{non-blocking reads},\footnote{In this terminology, from \cite{CDV09}, a \emph{blocking read} blocks a write; it does not refer to a read that is blocked.} we have (1) but not (2).
This model, where writes block reads but reads do not block writes, may apply when writes can abort in-progress reads, superseding them.

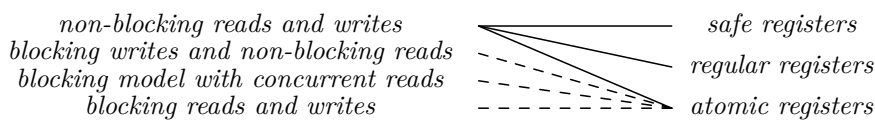
\expandafter\ifx\csname graph\endcsname\relax
   \csname newbox\expandafter\endcsname\csname graph\endcsname
\fi
\ifx\graphtemp\undefined
  \csname newdimen\endcsname\graphtemp
\fi
\expandafter\setbox\csname graph\endcsname
 =\vtop{\vskip 0pt\hbox{%
    \graphtemp=.5ex
    \advance\graphtemp by 0.429in
    \rlap{\kern 0.000in\lower\graphtemp\hbox to 0pt{\hss \emph{blocking reads and writes}\hss}}%
    \graphtemp=.5ex
    \advance\graphtemp by 0.286in
    \rlap{\kern 0.000in\lower\graphtemp\hbox to 0pt{\hss \emph{blocking model with concurrent reads}\hss}}%
    \graphtemp=.5ex
    \advance\graphtemp by 0.143in
    \rlap{\kern 0.000in\lower\graphtemp\hbox to 0pt{\hss \emph{blocking writes and non-blocking reads}\hss}}%
    \graphtemp=.5ex
    \advance\graphtemp by 0.000in
    \rlap{\kern 0.000in\lower\graphtemp\hbox to 0pt{\hss \emph{non-blocking reads and writes}\hss}}%
    \graphtemp=.5ex
    \advance\graphtemp by 0.429in
    \rlap{\kern 2.857in\lower\graphtemp\hbox to 0pt{\hss \emph{atomic registers}\hss}}%
    \graphtemp=.5ex
    \advance\graphtemp by 0.214in
    \rlap{\kern 2.857in\lower\graphtemp\hbox to 0pt{\hss \emph{regular registers}\hss}}%
    \graphtemp=.5ex
    \advance\graphtemp by 0.000in
    \rlap{\kern 2.857in\lower\graphtemp\hbox to 0pt{\hss \emph{safe registers}\hss}}%
\pdfliteral{
q [] 0 d 1 J 1 j
0.576 w
0.576 w
q [3.6 4.] 0 d
92.592 -30.888 m
164.592 -30.888 l
S Q
q [3.6 4.081382] 0 d
92.592 -20.592 m
164.592 -30.888 l
S Q
q [3.6 3.528678] 0 d
92.592 -10.296 m
164.592 -30.888 l
S Q
92.592 0 m
164.592 -30.888 l
S
92.592 0 m
164.592 -15.408 l
S
92.592 0 m
164.592 0 l
S
Q
}%
    \hbox{\vrule depth0.429in width0pt height 0pt}%
    \kern 2.857in
  }%
}%

\centerline{\box\graph}
\vspace{2ex}

In this paper, we model six different memory models, which are illustrated above. The blocking aspect of our memory models is captured via different concurrency relations (\autoref{sec:justness-thread-register}). The distinction between safe, regular and atomic registers is captured via three different process algebraic models (\autoref{sec:registers summary}). \hypertarget{justification}{Since the safe/regular/atomic distinction is only relevant in models that allow writes to overlap reads and writes, we only make it for the non-blocking model; for the other three memory models we reuse our atomic register models}.

\vspace{-3pt}
\subparagraph{Completeness criteria.}

In previous work \cite{spronck2023process}, we checked the mutual exclusion property of several algorithms, with safe, regular and atomic MWMR registers, through model checking with the mCRL2 toolset~\cite{mCRL2toolset}.
We did not check the liveness properties at that time; the presence of certain infinite loops in our models introduced spurious counterexamples to such properties, which hindered our verification efforts.
As an example of what we call a ``spurious counterexample'', we frequently found violations to starvation freedom where one thread, $i$, never obtained access to its critical section because a different thread, $j$, was endlessly repeating a busy wait, or some other infinite cycle which should reasonably not prevent $i$ from progressing to its critical section.
Yet, the model checker does not know this, and can therefore only conclude that the property is not satisfied.\pagebreak[3]
In this paper, we extend our previous work by addressing this problem and checking liveness properties as well.

One method for discarding spurious counterexamples from verification results is applying completeness criteria: rules for determining which paths in the model represent real executions of the modelled system.
By ensuring that all spurious paths are classified as incomplete and only taking complete paths into consideration when verifying liveness properties, we can circumvent the spurious counterexamples.
Of course, one must take care not to discard true system executions by classifying those as incomplete.
The completeness criterion must therefore be chosen with care.
Examples of well-known completeness criteria are weak fairness and strong fairness.
Weak fairness assumes that every task\footnote{What constitutes a \emph{task} differs from paper to paper; hence there are multiple flavours of strong and weak fairness; here a task could be a read or write action of a certain thread on a certain register.} that eventually is perpetually enabled must occur infinitely often; strong fairness assumes that if a task is infinitely often enabled it must occur infinitely often \cite{lehmann1981impartiality,apt1983proof,glabbeek2019progress}.
In effect, making a fairness assumptions amounts to assuming that if something is tried often enough, it will always eventually succeed \cite{glabbeek2019progress}.
In that sense, these assumptions, even weak fairness, are rather strong, and may well result in true system executions being classified as incomplete.
In this paper, we therefore use the weaker completeness criterion \emph{justness} \cite{glabbeek2019progress,glabbeek2019justness,bouwman2020off}.

Unlike weak and strong fairness, justness takes into account how different actions in the model relate to each other. 
Informally, it says that if an action $a$ can occur, then eventually $a$ occurs itself, or a different action occurs that interferes with the occurrence of $a$.
The underlying idea
of justness 
is that the different components that make up a system must all be capable of making progress: if thread $i$ wants to perform an action entirely independent of the actions performed by $j$, then there can be no interference.
However, if both threads are interacting with a shared register, then we may decide that one thread writing to the register can prevent the other from reading it at the same time, or vice versa.
Which actions interfere with each other is a modelling decision, dependent on our understanding of the real underlying system.
It is formalised through a \emph{concurrency relation}, which must adhere to some restrictions. In this paper we propose four concurrency relations, each modelling one of the four major memory models reviewed above: 
non-blocking reads and writes, blocking writes and non-blocking reads, the blocking model with concurrent reads, and blocking reads and writes.

\subparagraph{Model checking.}
Traditionally, mutual exclusion algorithms have been verified by pen-and-paper proofs using behavioural reasoning. As remarked by Lamport \cite{Lamport86Mutex2},
``the behavioral reasoning used in our correctness proofs, and in most other published correctness proofs of concurrent algorithms, is inherently unreliable''.  This is especially the case when dealing with the intricacies of non-atomic registers.\footnote{A good illustration of unreliable behavioural reasoning is given in \cite[Section 21]{glabbeek2023modelling}, through a short but fallacious argument that the mutual exclusion property of Peterson's mutual exclusion protocol, which is known to hold for atomic registers, would also hold for safe registers. We challenge the reader to find the fallacy in this argument before looking at the solution.} 
This problem can be alleviated by automated formal verification; here we employ model checking.

While the precise modelling of the algorithms, the registers and the employed completeness criterion requires great care, the subsequent verification requires a mere button-push and some patience.
Since our model checker traverses the entire state-space of a protocol, the verified protocols and all their registers need to be finite. This prevented us from checking the bakery algorithm \cite{Lamport74}, as it is one of the few mutual exclusion protocols that employs an unbounded state space.\pagebreak[3] Moreover, those algorithms that work for $N$ threads, for any $N\in\mathbbm{N}$, could be checked for small values of $N$ only; in this paper we take $N=3$. Consequently, any failure of a correctness property that shows up only for $>3$ threads will not be caught here.

As stated, we employed these methods in previous work to check mutual exclusion algorithms.
Although there we checked only safety properties, and did not consider the blocking aspects of memory, this already gave interesting results.
For instance, we showed that Szymanski's flag algorithm from \cite{Szy88}, even when adapted to use Booleans, violates mutual exclusion with non-atomic registers.
Here, we expand this previous work by checking deadlock freedom and starvation freedom in addition to mutual exclusion, and by including blocking into our memory models.
In total, we check the three correctness properties of over a dozen mutual exclusion algorithms, for six different memory models.
Among others, we cover Aravind's BLRU algorithm \cite{aravind2010yet}, Dekker's algorithm \cite{dijkstra1962over,alagarsamy2003some} and its RW-safe variant \cite{buhr2016dekker}, and Szymanski's 3-bit linear wait algorithm \cite{Szy90}.
In some cases where we find property violations, we suggest fixes to the algorithms so that the properties are satisfied.

\section{Preliminaries}\label{sec:preliminaries}

A \emph{labelled transition system} (LTS) is a tuple $(\states, \actionset, \initstate, \transrel)$ in which $\states$ is a finite set of states, $\actionset$ is a finite set of actions, $\initstate \in \states$ is the initial state, and $\transrel \subseteq \states \times \actionset \times \states$ is a transition relation.
We write $s \xrightarrow{a} s'$ for $(s, a, s') \in \transrel$.
We say an action $a$ is \emph{enabled} in a state $s$ if there exists a state $s'$ such that $(s, a, s') \in \transrel$.

A \emph{path} $\pi$ is a non-empty, potentially infinite alternating sequence of states and actions $s_0 a_1 s_{1} a_2 \ldots$, with $s_0, s_1, \ldots \in \states$ and $a_1, a_2, \ldots \in \actionset$, such that if $\pi$ is finite, then its last element is a state, and for all $i \in \mathbbm{N}$, $s_i \xrightarrow{a_{i+1}} s_{i+1}$. 
The first state of $\pi$ is its \emph{initial state}.
The \emph{length} of $\pi$ is the number of transitions in it.

We use a notion of parallel composition that is taken from Hoare's CSP \cite{Ho85}, where synchronisation between components is enforced on all shared actions. 
It is defined as follows:
    For some $k\geq 1$, let $P_1, \ldots, P_k$ be LTSs, where $P_i = (\states_i, \actionset_i, \initstate_i, \transrel_i)$ for all $1 \leq i \leq k$.
    The \emph{parallel composition} $P_1 \parcomp \ldots \parcomp P_k$ of $P_1, \ldots, P_k$ is the LTS $P = (\states, \actionset, \initstate, \transrel)$ in which $\states = \states_1 \times \ldots \times \states_k$, $\actionset = \bigcup_{1 \leq i \leq k}\actionset_i$, $\initstate = (\initstate_1, \ldots, \initstate_k)$, and a transition $((s_1, \ldots, s_k), a, (s_1', \ldots, s_k'))$ is in $\transrel$ if, and only if, $a \in \actionset$ and the following are true for all $1 \leq i \leq k$:
                   if $a \notin \actionset_i$, then $s_i = s_i'$, and
                   if $a \in \actionset_i$, then $(s_i, a, s_i') \in \transrel_i$.
\\
Note that by this definition of $\transrel$, if an action is in the action set of a component but not enabled by that component in a particular state of the parallel composition, then the composition cannot perform a transition labelled with that action.

As mentioned in the introduction, the completeness criterion we use for our liveness verification is a variant of \emph{justness} \cite{glabbeek2019justness,glabbeek2019progress}.
Specifically, while justness is originally defined on transitions, we here define it on action labels, an adaption we take from \cite{bouwman2020off}. 
As stated earlier, the definition of justness relies on the notion of a concurrency relation.

\hypertarget{second}{
\begin{definition}\rm\label{def:conc}
    Given an LTS $(\states, \actionset, \initstate, \transrel)$, a relation $\conc \subseteq \actionset \times \actionset$ is a \emph{concurrency relation} if, and only if:
    \begin{itemize}
        \item $\conc$ is irreflexive.
        \item For all $a \in \actionset$, if $\pi$ is a path from a state $s \in \states$ to a state $s' \in \states$ such that $a$ is enabled in $s$ and $a \conc b$ for all $b \in \actionset$ occurring on $\pi$, then $a$ is enabled in $s'$.
    \end{itemize}
\end{definition}}\vspace{1ex}
A concurrency relation may be asymmetric. We often reason about the complement of $\conc$, $\nconc$. Read $a \conc b$ as ``$a$ is independent from $b$'' and $a \nconc b$ as ``$b$ interferes with/postpones $a$''. 

\begin{observation}\rm\label{obs:subset}
Concurrency relations can be refined by removing pairs; a subset of a concurrency relation is still a concurrency relation.
\end{observation}

Informally, justness says that a path is complete if whenever an action $a$ is enabled along the path, there is eventually an occurrence of an action (possibly $a$ itself) that interferes with it. 
This can be weakened by defining a set of \emph{blockable actions}, for which this restriction does not hold; a blockable action may be enabled on a complete path without there being a subsequent occurrence of an interfering action.
In this paper, the action of a thread to leave its non-critical section will be blockable. This way we model that a thread may choose to never take that option.
We give the formal definition of justness, incorporating the blockable actions. We represent the set of blockable actions as $\block$. Its complement, $\nonblock$, is defined as $\actionset \setminus \block$, given a set of actions $\actionset$.

\hypertarget{just}{
\begin{definition}\rm\label{def:justness}
    A path $\pi$ in an LTS $(\states, \actionset, \initstate, \transrel)$ satisfies $\block$-$\conc$-\emph{justness of actions} ($\justact{\conc}{\block}$) if, and only if, for each suffix $\pi'$ of $\pi$, if an action $a \in \nonblock$ is enabled in the initial state of $\pi'$, then an action $b \in \actionset$ occurs in $\pi'$ such that $a \nconc b$.
\end{definition}}
We say that a property is satisfied on a model under $\justact{\conc}{\block}$ if it is satisfied on every path of that model, starting from the model's initial state, that satisfies $\justact{\conc}{\block}$. If $\block$ and $\conc$ are clear from the context, we simply say that a path that satisfies $\justact{\conc}{\block}$ is \emph{just}.

\section{Register models}\label{sec:registers summary}

In \cite{spronck2023process}, we presented process-algebraic models of MWMR safe, regular and atomic registers. Through the semantics of the process algebra, this determines an LTS for each register of a given kind. In this paper, we use the same definitions for the three register types, but we have altered the process-algebraic models to be more compact and better facilitate the definition of the concurrency relations. The process-algebraic models can be found in \autoref{sec:registers}; here we merely summarise the key design decisions.

A register model represents a multi-reader multi-writer read-write register that allows every thread to read from and write to it.
However, every thread may only perform a single operation on the register at a time.
The register model specifies the behaviour of the register in response to operations performed by threads.
Here we presuppose two disjoint finite sets: $\TID$ of thread identifiers (thread id's) and $\RID$ of register identifiers (register id's).
Additionally, for every $r\mathbin\in\RID$ we reference the set $\Data$ of all data values that the register $r$ can hold. 

Recall that read and write operations take time, and may hence be concurrent. Therefore we represent a single operation with two actions: an \emph{invocation} to indicate the start of the operation, and a \emph{response} to indicate its end. Two operations are concurrent if the interval between their respective invocations and responses overlaps.
The interface of a register is represented by the following actions: a read by thread $\tid \in \TID$ of register $\rid \in \RID$ that returns value $\data \in \Data$ is a start read action $\startread[\tid,\rid]$ followed by a finish read action $\finishread[\tid,\rid]{\data}$; a write of value $\data \in \Data$ by a thread $\tid \in \TID$ to register $\rid \in \RID$ is a start write action $\startwrite[\tid,\rid]{\data}$ followed by a finish write action $\finishwrite[\tid,\rid]$.
In addition to this interface, which the threads can use to perform operations on the register, 
the regular and atomic models also use \emph{register local actions}; these are internal actions by the register that are used to model the correct behaviour.

A register model requires some finite amount of memory to store a representation of relevant past events. We store this in what we call the \emph{status object}, which features a finite set $\StatusAll$ of possible states. We abstract away from the exact implementation; for this presentation, all that is relevant is which information can be retrieved from it. Amongst others, we use the following \emph{access functions}, which are local to any given register $r$:
\begin{itemize}
    \item $\truevalsym: \StatusAll \rightarrow \Data$, the value that is currently stored in the register.
    \item $\writerssym: \StatusAll \rightarrow 2^{\TID}$, the set of thread id's of threads that have invoked a write operation on this register that has not yet had its response.
\end{itemize}
Any occurrence of a register action $a$ induces a state change $s \xrightarrow{a} s'$, resulting in an \emph{update} to these access functions. For instance, the actions $\startwrite[\tid,\rid]{\data}$ and $\finishwrite[\tid,\rid]$ cause the updates $\writers{s'} = \writers{s} \cup \{\tid\}$ and $\writers{s'} = \writers{s} \setminus \{\tid\}$, respectively.

\vspace{-3pt}
\subsection{Safe MWMR registers}\label{sec:safe}
To extend the single-writer definition of safe registers to a multi-writer one, we follow Lamport in assuming that a concurrent read cannot affect the behaviour of a read or write. Lamport's SWMR definitions consider how concurrent writes affect reads, but not how concurrent writes affect writes.
Here we follow Raynal's approach for safe registers: a write that is concurrent with another write sets the value of the register to some arbitrary value in the domain of the register \cite{Raynal13}.
We can summarise the behaviour of MWMR safe registers with four rules:
\begin{enumerate}
    \item A read not concurrent with any writes on the same register returns the value most recently written into the register.
    \item A read concurrent with one or more writes on the same register returns an arbitrary value in the domain of the register.
    \item A write not concurrent with any other write on the same register results in the intended value being set.
    \item A write concurrent with one or more other writes on the same register results in an arbitrary value in the domain of the register being set.\vspace{-3pt}
\end{enumerate}

In our model of a safe register $r$, its status object maintains a Boolean variable $\overlapsym$ for each thread id, telling whether an ongoing read or write action of this thread overlapped with a write by another thread. The value of $\overlapsym$ is updated in a straightforward way each time $r$ experiences a register interface action $\startread[\tid,\rid]$, $\finishread[\tid,\rid]{\data}$, $\startwrite[\tid,\rid]{\data}$ or $\finishwrite[\tid,\rid]$, aided by the access function $\writerssym$.
Using this function, our model can determine which of the above four rules applies when a read or write finishes, and behave accordingly.

\vspace{-3pt}
\subsection{Regular MWMR registers}

We wish to define regular MWMR registers as an extension of Lamport's definition of SWMR regular registers: a read returns either the last written value before the read began, or the value of any concurrent write, if there is one. This is non-trivial; in \cite{spronck2023process} we present one extension and compare it to four different suggestions from \cite{Shao11}. The complexity comes from determining what the last written value is, given that writes may be concurrent with each other. Here, 
following \cite{spronck2023process}, we require that all threads see the same global ordering on writes once those writes have completed. Hence, if two writes $w_1$ and $w_2$ occur concurrently, and after their completion, but before the invocation of any other write, there are two reads $r_1$ and $r_2$, then either both $r_1$ and $r_2$ see $w_1$'s value as the last written value, or they both see $w_2$'s value as the last written value.
We generate the global ordering through the register local order write action $\orderwrite[\tid,\rid]$, which is scheduled between the start write action $\startwrite[\tid,\rid]{\data}$ and the finish write $\finishwrite[\tid,\rid]$. This action does not represent any true internal behaviour by the register; the interleaving of order write actions from various threads merely determines the global ordering. 
Given a read, we say the ``last written'' value this read sees, and hence the value this read may return in addition to those of overlapping writes, is the intended value of the write whose ordering was the most recent before the read's invocation.

In our process algebraic model, the value $\truevalsym$ of the register is set when an action $\orderwrite[\tid,\rid]$ occurs. During any read action of a thread $\tid$, that is, between $\startread[\tid,\rid]$ and $\finishread[\tid,\rid]{\data}$, the register model builds a set of the possible return values on the fly.
When the read starts, this set is initialised to $\truevalsym$ and the intended value of every active write.
Subsequently, whenever a write occurs, its intended value is added to the set.
This way, at the finish read, the set will contain exactly those values that the read could return.

\subsection{Atomic MWMR registers}
\vspace{-3pt}
Lamport's definition of SWMR atomic registers, namely that the register must behave as though reads and writes occur in some strict order, is directly applicable to the MWMR case.
We reuse the register local order write action from the regular register model, and add the similar order read action $\orderread[\tid,\rid]$ for read operations. This way, we generate an ordering on all operations. In our process-algebraic model, $\truevalsym$ is updated when the $\orderwrite[\tid,\rid]$ occurs, similar to regular registers. For read operations, the value of $\truevalsym$ when $\orderread[\tid,\rid]$ occurs is remembered, and returned at the matching response.

\vspace{-2pt}
\section{Thread-register models}\label{sec:thread-register}
\vspace{-3pt}
In our models, we combine processes representing both threads and registers. 
Similar to how register processes may contain register local actions, threads may contain \emph{thread local actions}.
Crucially, the local actions of both registers and threads are not involved in any communication, meaning that the only way for two threads to communicate is by writing to and reading from registers using the register interface actions.

The register models are mostly independent of the algorithm that we analyse with them; the algorithm merely dictates a register's identifier, domain, and initial value.
However, the thread models are fully dependent on the modelled algorithm, which dictates their behaviour.
Therefore, we cannot present an algorithm-independent thread process.
Instead, we presuppose the existence of an LTS $T_{\tid} = (\states_{\tid}, \actionset_{\tid}, \initstate_{\tid}, \transrel_{\tid})$ and a set of thread local actions $\thrlocacts_{\tid}$ for every $\tid \mathbin\in \TID$ such that $\actionset_{\tid} = \{\startread[\tid, \rid], \finishread[\tid, \rid]{\data}, \startwrite[\tid, \rid]{\data}, \finishwrite[\tid, \rid] \mid \rid \mathbin\in \RID, \data \mathbin\in \Data\} \cup \thrlocacts_{\tid}$. We assume that all sets of thread local actions are pairwise disjoint and that all thread LTSs $T_{\tid}$ satisfy two properties, representing the reasonable implementation of read and write operations. Firstly, on all paths from $\initstate_{\tid}$, each transition labelled $\startread[\tid,\rid]$ for some $\rid \in \RID$ must go to a state where exactly the actions $\finishread[\tid,\rid]{\data}$ for all $\data \in \Data$ are enabled. Similarly, all transitions labelled $\startwrite[\tid,\rid]{\data}$ for some $\rid\in\RID, \data \in \Data$ must go to states where only $\finishwrite[\tid,\rid]$ is enabled. Secondly, transitions labelled $\finishread[\tid,\rid]{\data}$ or $\finishwrite[\tid,\rid]$ are only enabled in these states.

We combine thread and register LTSs into \emph{thread-register models}.
For the following definition, we let $R_{\rid} = (\states_{\rid}, \actionset_{\rid}, \initstate_{\rid}, \transrel_{\rid})$ be the LTS associated with each $\rid \in \RID$.
\begin{definition}\rm\label{def:tr-model}
    A \emph{thread-register model} is a six-tuple $(\states, \actionset, \initstate, \transrel, \thrsym, \regsym)$, such that $(\states, \actionset, \initstate, \transrel)$ is a parallel composition of thread and register LTSs and
    \begin{itemize}
        \item $\thrsym \!: \actionset \rightarrow \TID$ is a mapping from actions to thread id's; and
        \item $\regsym \!: \actionset \rightarrow \RID \cup \{\undefsymb\}$ is a mapping from actions to register id's and the special value $\undefsymb \mathbin{\notin} \RID$.
    \end{itemize}
\noindent
$\thrsym$ and $\regsym$ are defined in the obvious way, e.g., $\thrmap{\startread[\tid,\rid]} = \tid$ and $\regmap{\startread[\tid,\rid]} = \rid$.
Crucially, for a thread local action $a$, $\regmap{a} = \undefsymb$.
\end{definition}
Note that by our construction of the thread and register LTSs, every action in $\actionset$ appears in at most two components of the parallel composition. Specifically, for all $\tid \in \TID$ and $a_{\tid} \in \thrlocacts_{\tid}$, $a_{\tid}$ appears only in $T_{\tid}$; for all $\tid \in \TID, \rid \in \RID$, $\orderread[\tid,\rid]$ and $\orderwrite[\tid,\rid]$ appear only in $R_{\rid}$; and for all $\tid \in \TID, \rid \in \RID$ and $\data \in \Data$, $\startread[\tid,\rid], \finishread[\tid,\rid]{\data}, \startwrite[\tid,\rid]{\data}$ and $\finishwrite[\tid,\rid]$ appear exactly in $T_{\tid}$ and $R_{\rid}$.

\vspace{-2pt}
\section{Justness for thread-register models}\label{sec:justness-thread-register}

\vspace{-2pt}
In order to obtain a suitable notion of justness for our thread-register models, we need to choose both $\block$ and $\conc$. 
Only thread local actions will be blockable; we define $\block$ in \autoref{sec:verification}.

The concurrency relation, on the other hand, should relate the register interface actions.
This is how we represent whether it is reasonable for one thread's operations on a register to interfere with (and thereby postpone) another thread's operations on that register.
We use four different concurrency relations in our verifications, representing the four different models of blocking described in \autoref{sec:introduction}. 
These concurrency relations do not reference the thread local actions outside of the $\thrsym$ mapping, so we can already present these relations before giving more details on the precise models.
To establish that the relations we present are indeed concurrency relations, we first establish a property of our models. We call this property \emph{thread consistency}.

\begin{definition}\rm\label{def:thr-const}
    An LTS $(\states, \actionset, \initstate, \transrel)$ is \emph{thread consistent} with respect to a mapping $\thrsym: \actionset \rightarrow \TID$ if, and only if, for all states $s \in \states$, if an action $a \in \actionset$ is enabled in $s$ and there exists a transition $s \xrightarrow{b} s'$ for some $s' \in \states, b \in \actionset$ such that $\thrmap{a} \neq \thrmap{b}$, then $a$ is also enabled in $s'$.
\end{definition}

\noindent
The correctness of our concurrency relations (cf.\ \autoref{def:conc}) relies on our thread-register models being thread-consistent. 
The proof of this fact is given in \autoref{app:thr-consist-proof}.
\begin{restatable}{lemma}{thrconsist}\rm\label{lem:thr-consist}
    Let $M = (\states, \actionset, \initstate, \transrel, \thrsym, \regsym)$ be a thread-register model. 
    Then the LTS $(\states, \actionset, \initstate, \transrel)$ is thread consistent with respect to the mapping $\thrsym$.
\end{restatable}

For the definitions of our four concurrency relations, we fix a thread-register model $M = (\states, \actionset, \initstate, \transrel, \thrsym, \regsym)$.
We also introduce two predicates on $\actionset$: $\issrsym$  and $ \isswsym$. For an action $a \in \actionset$, these are defined as:
\[
    \issr{a} = \exists_{\tid\in\TID,\rid\in\RID}. (a=\startread[\tid,\rid])
\qquad\qquad
    \issw{a} = \exists_{\tid\in\TID,\tid\in\RID,\data\in\Data}. (a=\startwrite[\tid,\rid]{\data}).
\]

The \emph{thread interference relation}, $\conc_T$, expresses that every action is independent from every other action \emph{unless} the two actions belong to the same thread; every two actions by the same thread interfere with each other. It captures the memory model with non-blocking reads and writes. This is the coarsest concurrency relation we will use.

\begin{definition}\rm\label{def:concT}
    $\conc_T = \{(a, b) \mid a,b \in \actionset, \thrmap{a} \neq \thrmap{b}\}$
\vspace{-1ex}
\end{definition}

\begin{restatable}{lemma}{concTval}\rm\label{lem:concT-val}
   $\conc_T$ is a concurrency relation for $M$.
\end{restatable}
\noindent
This follows by a straightforward application of \autoref{lem:thr-consist}.
The details are in \autoref{app:concT-val}.

The model with blocking writes and non-blocking reads is captured by the \emph{signalling reads relation}, $\conc_S$.
\begin{definition}\rm
    $\conc_S = \conc_T \setminus \{(a, b) \mid a,b \in \actionset, \issr{a} \lor \issw{a},  \issw{b}, \regmap{a} = \regmap{b} \}$
\end{definition}
Intuitively, this is the same as $\conc_T$ except that one thread starting a write to a register can interfere with a write to or a read from that same register by another thread. However, a read cannot interfere with another thread's read or write.
This concurrency relation has a precedent in \cite{glabbeeksignals,bouwman2020off}. There, reads are modelled as \emph{signals}, which differ from standard actions in that they do not block any other actions. Hence the name of this relation.

The blocking model with concurrent reads is captured by the \emph{interfering reads relation}, $\conc_I$.
This is a further refinement from $\conc_S$, where a start read can interfere with a start write on the same register, but cannot interfere with a start read.

\begin{definition}\rm\label{def:concI}
    $\conc_I = \conc_S \setminus \{(a, b) \mid a,b \in \actionset, \issw{a},  \issr{b}, \regmap{a} = \regmap{b} \}$\pagebreak[3]
\end{definition}
This goes with the idea that performing a write on a memory location can only be done when the memory is reserved: repeated reads can prevent the memory from being reserved for a write, but as long as there is no write all the reads can take place concurrently.\pagebreak[4]

Finally, the model of blocking reads and writes is captured by the \emph{all interfering relation}, $\conc_{\!A}$, a refinement of $\conc_I$ where a start read can also interfere with another start read.
\begin{definition}\rm\label{def:concA}
    $\conc_{\!A}= \conc_I \setminus \{(a, b) \mid a,b \in \actionset, \issr{a},  \issr{b}, \regmap{a} = \regmap{b}\}$
\end{definition}

\noindent
In this model, every operation on a register fully reserves that register for only that operation, and hence can prevent any other operation from taking place at the same time.

\begin{lemma}\rm\label{lem:ai-val}
    $\conc_S$, $\conc_I$ and $\conc_{\!A}$ are concurrency relations for $M$.
\end{lemma}
\begin{proof}
    This follows from \autoref{lem:concT-val} and \autoref{obs:subset}.
\end{proof}

\noindent
As stated \hyperlink{justification}{in the introduction}, we capture six memory models. We obtain three variants of the non-blocking model by combining the $\conc_T$ relation with the safe, regular and atomic register models. The remaining three memory models are represented by combining $\conc_S$, $\conc_I$ and $\conc_{\!A}$ with the atomic register model. In \autoref{app:complete paths}, we formally characterise just paths in our thread-register models for all six variants. 
In \autoref{app:blocking registers}, we prove that using the atomic register model, which allows overlapping writes, for the three memory models with blocking writes is sound for verification purposes.

\section{Verification}\label{sec:verification}

Below, and in \autoref{app:algorithms}, we collect over a dozen mutual exclusion algorithms from the literature.
We also present altered versions of several algorithms, incorporating fixes we propose.
All these algorithms, and the registers themselves, have been translated to the process algebra mCRL2 \cite{mCRL2language}.
\href{https://github.com/mCRL2org/mCRL2/tree/master/examples/academic/non-atomic_registers}{The mCRL2 files are available as supplementary material.}
We give the most important design decisions regarding this translation here; further details can be found in \autoref{app:mCRL2}.
The only operations on registers we allow are reading and writing.
Hence, more complicated statements that may have been present in the original presentation of an algorithm, such as compare-and-swap instructions, are converted into these primitive operations. A statement like ``\textbf{await} $\forall_{i \in \TID}: x(i)$'', where $x$ is some condition on a thread id $i$, is modelled as a recursive process that checks each thread id from smallest to largest, waiting for each until $x(i)$ is satisfied.
Where an algorithm does not specify the initial value of a register, we take the lowest value from the given domain.

We use $\crit[\tid]$ and $\noncrit[\tid]$, with $\tid\in\TID$, as thread local actions.
These actions represent a thread entering its critical section and leaving its non-critical section, respectively. We define $\block = \{\noncrit[\tid] \mid \tid \in \TID\}$ to capture that a thread may always choose to remain in its non-critical section indefinitely; this is an important assumption underlying the correctness of mutual exclusion protocols.
  
We did our verification with the mCRL2 toolset~\cite{mCRL2toolset}.
To this end, we encoded mutual exclusion, deadlock freedom, and starvation freedom in the modal $\mu$-calculus, the logic used to represent properties in the mCRL2 toolset.
We used the patterns from \cite{spronck2024progress} to incorporate the justness assumption into our formulae for deadlock freedom and starvation freedom.
The full modal $\mu$-calculus formulae appear in \autoref{app:mucalc} (and also as supplementary material).
Besides the correctness properties discussed in \autoref{sec:introduction}, Dijkstra \cite{dijkstra65} requires that the correctness of the algorithms may not depend on the relative speeds of the threads. This requirement is automatically satisfied in our approach, since we allow all possible interleavings of thread actions in our models.

We checked mutual exclusion, deadlock freedom, and starvation freedom. 
If mutual exclusion is not satisfied, we do not care about the other two properties.
Additionally, if deadlock freedom is not satisfied, we know that starvation freedom is not satisfied either.
We can therefore summarise our results in a single table: $\satnone$ if none of the three properties are satisfied, $\satmutex$ if only mutual exclusion is satisfied, $\satdf$ if only mutual exclusion and deadlock freedom hold, and $\satsf$ if all three are satisfied. See \autoref{tab:results}.
As stated previously, we verify liveness properties under justness, where we employ $\conc_T$ for safe and regular registers and all four concurrency relations $\conc_T$, $\conc_S$, $\conc_I$ and $\conc_{\!A}$ for atomic registers.
We checked with 2 threads for algorithms designed for 2 threads, and with 3 for all others. We restrict\linebreak[3] ourselves to at most 3 threads because, due to the state-space explosion problem, even models with only 3 threads frequently take hours or even days to check these properties on. 

\begin{table}[tbp]
\caption{Verification results.}
\label{tab:results}
\resizebox{\textwidth}{!}{%
\begin{tabular}{@{}lc|cccccc@{}}
\toprule
\multirow{2}{*}{\textit{Algorithm}}              & \multirow{2}{*}{\textit{$\#$ threads}} & \textit{Safe} & \textit{Regular} & \multicolumn{4}{c}{\textit{Atomic}} \\
                                                 &                                                   & $T$             & $T$                & $T$       & $S$       & $I$      & $A$      \\ \midrule
    Anderson \cite{anderson1993fine}                 & 2 &  $\satsf$     & $\satsf$ &  $\satsf$ & $\satsf$ & $\satmutex$ &  $\satmutex$ \\* \midrule
    Aravind BLRU \cite{aravind2010yet}             & 3 & $\satsf$ & $\satsf$ & $\satsf$ & $\satmutex$ & $\satmutex$ & $\satmutex$ \\
    Aravind BLRU (alt.)                             & 3  & $\satsf$ & $\satsf$ & $\satsf$ & $\satsf$ & $\satmutex$ & $\satmutex$ \\* \midrule
    Attiya-Welch (orig.) \cite{AttiyaWelch04}        & 2  &  $\satdf$     & $\satsf$    &  $\satsf$ & $\satdf$ & $\satmutex$ & $\satmutex $ \\
    Attiya-Welch (orig., alt.)                       & 2  &  $\satsf$     & $\satsf$    &   $\satsf$ & $\satdf$ & $\satmutex$ & $\satmutex $ \\
    Attiya-Welch (var.) \cite{Shao11}                & 2  &  $\satmutex$  & $\satmutex$ &  $\satsf$ & $\satdf$ & $\satmutex$ & $\satmutex $ \\
    Attiya-Welch (var., alt.)                        & 2  &  $\satsf$     & $\satsf$    & $\satsf$ & $\satdf$ & $\satmutex$ & $\satmutex$  \\* \midrule
    Burns-Lynch \cite{burns1993bounds}               & 3  &  $\satdf$     & $\satdf$    &  $\satdf$ & $\satdf$ & $\satmutex$ & $\satmutex$  \\* \midrule
    Dekker \cite{alagarsamy2003some}                 & 2  &  $\satmutex$  & $\satmutex$ &  $\satsf$ & $\satdf$ & $\satmutex$ & $\satmutex $ \\
    Dekker (alt.)                                    & 2  &  $\satmutex$  & $\satmutex$ &  $\satsf$ & $\satsf$ & $\satmutex$ & $\satmutex $ \\
    Dekker RW-safe \cite{buhr2016dekker}           & 2  &  $\satsf$     & $\satsf$    &  $\satsf$ & $\satdf$ & $\satmutex$ & $\satmutex $ \\
    Dekker RW-safe (DFtoSF) & 2 & $\satsf$ & $\satsf$ & $\satsf$ & $\satsf$ & $\satmutex$ & $\satmutex$  \\* \midrule
    Dijkstra \cite{dijkstra65}                       & 3  &  $\satmutex$  & $\satdf$ &  $\satdf$ & $\satmutex$ & $\satmutex$ & $\satmutex$  \\* \midrule
    Kessels \cite{kessels1982arbitration}            & 2  &  $\satnone$   & $\satnone$  &  $\satsf$ & $\satsf$ & $\satmutex$ & $\satmutex $ \\* \midrule
    Knuth \cite{knuth1966additional}                 & 3  &  $\satmutex$  &  $\satsf$ & $\satsf$ & $\satmutex$ & $\satmutex$ & $\satmutex$ \\* \midrule
    Lamport 1-bit \cite{Lamport86Mutex2}             & 3  &  $\satdf$     &   $\satdf$  &  $\satdf$ & $\satdf$ & $\satmutex$ & $\satmutex$  \\
    Lamport 1-bit (DFtoSF) & 3 & $\satsf$ & $\satsf$ & $\satsf$ & $\satsf$ & $\satmutex$ & $\satmutex$ \\
    Lamport 3-bit \cite{Lamport86Mutex2}             & 3  &  $\satsf$ & $\satsf$ & $\satsf$ & $\satsf$ & $\satmutex$ & $\satmutex$ \\* \midrule
    Peterson \cite{Peterson81}                       & 2  &  $\satnone$   & $\satnone$  &  $\satsf$ & $\satsf$ & $\satmutex$ & $\satmutex $ \\
\midrule
    Szymanski flag (int.) \cite{Szy88}                    & 3  &  $\satnone$   & $\satnone$  &  $\satsf$ & $\satsf$ & $\satmutex$ & $\satmutex$ \\
    Szymanski flag (bit) \cite{Szy88}                & 3  &  $\satnone$   & $\satnone$ &   $\satnone$ & $\satnone$ & $\satnone$ & $\satnone$    \\
    Szymanski 3-bit lin. wait \cite{Szy90}         & 3  &  $\satnone$   & $\satnone$ &   $\satnone$ & $\satnone$ & $\satnone$ & $\satnone$   \\
    Szymanski 3-bit lin. wait (alt.)                & 2  &  $\satsf$     & $\satsf$    &  $\satsf$ & $\satsf$ & $\satmutex$ & $\satmutex$     \\* 
    \bottomrule
\end{tabular}%
}
\end{table}

We list the origin of each algorithm in the table; the results of verifying our proposed alternate versions are indicated by ``alt.''.
We give pseudocode for the algorithms.
Therein we merely present the entry and exit protocols of an algorithm, separated by the instruction \textbf{critical section}. Implicitly these instructions alternate with the non-critical section, and may be repeated indefinitely.
\hypertarget{threadorder}{We use $N$ for the number of threads.
As identifiers for threads}, we use the integers $0\ldots N{-}1$. So $\TID = \{0,\ldots,N{-}1\}$. When presenting pseudocode, we give the algorithm for an arbitrary thread $i$. When $N = 2$, we use the shorthand notation $j = 1-i$.

In the subsequent sections, we discuss the most interesting results.
Further discussion, and the remaining pseudocode, appears in \autoref{app:algorithms}.

\subsection[Impossibility of liveness with reads interference]{Impossibility of liveness with $\conc_I$}
Perhaps the most notable result in \autoref{tab:results} is that no algorithm satisfies either liveness property under $\justact{\conc_I}{\block}$ or $\justact{\conc_{\!A}}{\block}$.
Since $\conc_{\!A}$ is a refinement of $\conc_I$, we focus on the behaviour for $\conc_I$.
When we take $\conc_I$ as our concurrency relation, then one thread's read of a register can interfere with another thread's write to that same register.
It turns out that when this is the case, starvation freedom is impossible for algorithms that rely on communication via registers. The following argument is adapted from \cite{glabbeek2018speed,glabbeek2023modelling}.
Assume that $\mathit{Alg}$ is an algorithm that satisfies starvation freedom.
Let $i$ and $j$ be different threads, and assume that all other threads, if any, stay in their non-critical section forever.
Since $\mathit{Alg}$ is starvation-free, thread $i$ must be capable of freely entering the critical section if thread $j$ is not competing for access.
Hence, 
thread $j$ must communicate its interest in the critical section to thread $i$ as part of its entry protocol.
Since reading from and writing to registers is the only form of communication we allow, 
thread $j$ must, in its entry protocol, write to some register $\varname{reg}$, which $i$ must read in its own entry protocol.
As long as $i$ does not read $j$'s interest from $\varname{reg}$, thread $i$ can enter the critical section freely.
Therefore, if thread $i$'s read of $\varname{reg}$ can block thread $j$'s write to $\varname{reg}$, thread $i$ can infinitely often access the critical section without ever letting thread $j$ communicate its interest, thus never letting thread $j$ enter.

For this argument it is crucial that right after $i$'s read of $\varname{reg}$, thread $i$ enters and then leaves the critical section and returns to its entry protocol, where it engages in another read of $\varname{reg}$, so quickly that thread $j$ has not yet started its write to $\varname{reg}$ in the meantime. 
This uses the requirement on mutual exclusion protocols that their correctness may not depend on the relative speeds of the
threads.
Without that requirement one can easily achieve starvation freedom even with blocking reads, as demonstrated in \cite{glabbeek2023modelling}.

The argument above explains why starvation freedom is never satisfied for $\conc_I$ or $\conc_{\!A}$.
However, it does not explain why we also never observe deadlock freedom.
After all, in the {\trace} sketched above, while thread $j$ is stuck in its entry protocol, thread $i$ infinitely often accesses the critical section.
While we do not (yet) have an argument that deadlock freedom is impossible to satisfy if reads can block writes for all possible algorithms, we do observe this to be the case for all algorithms we have analysed.

For many algorithms, it is possible for both competing threads to become stuck in their entry protocol. Consider, for example, Peterson's algorithm from \cite{Peterson81}, here given as \cref{alg:peterson}.
If $\varname{turn}$ is initially $0$, and thread 1 manages to set $\varidx{flag}{1}$ to $\true$ before thread 0 starts the competition, then on line 3 thread 0 will get stuck in a busy waiting loop. Thread 1 needs to set $\varname{turn}$ to $1$ to let thread 0 pass line 3, but thread 0's repeated reads of $\varname{turn}$ prevent this write from taking place, resulting in both threads being trapped in the entry protocol.
An alternative way to get a deadlock freedom violation is via the exit protocol.
Once a thread has finished its critical section access, it needs to communicate that it no longer requires access to the other thread. In Peterson's, this is done on line 5 by setting the thread's $\varname{flag}$ to $\false$.
However, if the other thread is repeatedly reading this register, such as is done on line 3, then the completion of the exit protocol can be blocked, once again preventing both threads from accessing their critical sections.

We see similar behaviour in all algorithms we analyse.
Frequently, although not always, the problem lies in busy waiting loops.
Given this behaviour, it would be interesting to modify our models to treat busy waiting reads differently from normal reads, and only allow normal reads to interfere with writes.
This would give us greater insight into whether for some of the algorithms it is truly the busy waiting that is the source of the deadlock freedom violation.
We leave this as future work.

\subsection{Aravind's BLRU algorithm}\label{sec:aravind}

Aravind's BLRU algorithm \cite{aravind2010yet}, here given as \cref{alg:aravind}, is designed for an arbitrary number of threads $N$.
\begin{table}[t]
\vspace{-1em}
\noindent\begin{minipage}[t]{0.38\textwidth}
\begin{algorithm}[H]
\caption{Peterson's algorithm}\label{alg:peterson}
\begin{algorithmic}[1]
    \State{$\varidx{flag}{i} \writeop \true$}
    \State{$\varname{turn} \writeop i$}
    \State{\textbf{await} $\varidx{flag}{j} = \false \lor \varname{turn} = j$}
    \State{\textbf{critical section}}
    \State{$\varidx{flag}{i} \writeop \false$}
\end{algorithmic}
\end{algorithm}
\vspace{-1em}
\begin{algorithm}[H]
\caption{Dekker's algorithm}\label{alg:dekker}
\begin{algorithmic}[1]
    \State{$\varidx{flag}{i} \writeop \true$}
    \While{$\varidx{flag}{j} = \true$}
        \If{$\varname{turn} = j$}
            \State{$\varidx{flag}{i} \writeop \false$}
            \State{\textbf{await} $\varname{turn} = i$}
            \State{$\varidx{flag}{i} \writeop \true$}
        \EndIf
    \EndWhile
    \State{\textbf{critical section}}
    \State{$\varname{turn} \writeop j$}
    \State{$\varidx{flag}{i} \writeop \false$}
\end{algorithmic}
\end{algorithm}
\end{minipage}
\hfill
\begin{minipage}[t]{0.56\textwidth}
\begin{algorithm}[H]
\caption{Aravind's BLRU algorithm}\label{alg:aravind}
\begin{algorithmic}[1]
    \State{$\varidx{flag}{i} \writeop \true$}
    \Repeat
        \State{$\varidx{stage}{i} \writeop \false$}
        \State{\textbf{await} $\forall_{j \neq i}: \varidx{flag}{j} = \false \lor \varidx{date}{i} < \varidx{date}{j}$}
        \State{$\varidx{stage}{i} \writeop \true$}
    \Until{$\forall_{j \neq i}:$ $\varidx{stage}{j} = \false$}
    \State{\textbf{critical section}}
    \State{$\varidx{date}{i} \writeop \max(\varidx{date}{0}, ..., \varidx{date}{N-1}) + 1$}
    \If{$\varidx{date}{i} \geq 2N-1$}
        \State{$\forall_{j \in [0...N-1]}: \varidx{date}{j} \writeop j$}
    \EndIf
    \State{$\varidx{stage}{i} \writeop \false$}
    \State{$\varidx{flag}{i} \writeop \false$}
\end{algorithmic}
\end{algorithm}
\end{minipage}
\end{table}
Every thread has three registers: $\varname{flag}$ and $\varname{stage}$, Booleans that are initialised at $\false$, and a natural number $\varname{date}$, initialised at the thread's own id.
We observe that this algorithm satisfies all three properties with safe and regular registers, as claimed in \cite{aravind2010yet}.
However, with atomic registers, deadlock freedom is violated under $\justact{\conc_S}{\block}$.
The following {\trace} for two threads demonstrates this violation:
\begin{itemize}
    \item Thread 1 moves through lines 1 through 5, setting $\varidx{flag}{1}$ and $\varidx{stage}{1}$ to $\true$. Note that thread 1 can go through line 4 because $\varidx{flag}{0} = \false$.
    \item Thread 0 can similarly move through lines 1 through 5; while $\varidx{flag}{1} = \true$, we do have that $\varidx{date}{0} = 0 < 1 = \varidx{date}{1}$, so it can pass through line 4. However, on line 6, thread 0 observes $\varidx{stage}{1} = \true$, so it has to return to line 2. 
\end{itemize}
At this point, thread 0 can repeat lines 2 through 5 endlessly, as long as thread 1 does not set $\varidx{stage}{1}$ to $\false$.
Note that the resulting infinite {\trace} satisfies $\justact{\conc_S}{\block}$: thread 1's read of $\varidx{stage}{0}$, which it has to perform on line 6, is repeatedly blocked by thread 0's writes to $\varidx{stage}{0}$ on lines 3 and 5.

This violation can easily be fixed by preventing a thread from endlessly repeating the loop in the entry protocol while the other thread's $\varname{stage}$ is $\false$.
This can be done by altering line 4 to instead say $\textbf{await } \forall_{j \neq i}:\varidx{flag}{j} = \false \lor (\varidx{date}{i} < \varidx{date}{j} \land \varidx{stage}{j} = \false)$.
While this makes it more difficult to progress through line 4, it is impossible for all threads to get stuck there: if all threads are on line 4, then $\varname{stage}$ is $\false$ for all of them, and so the one with the lowest $\varname{date}$ can go to line 5.
Indeed, as is shown in \autoref{tab:results}, with this modification Aravind's algorithm now satisfies starvation freedom under $\justact{\conc_S}{\block}$.

\subsection{Dekker's algorithm}\label{sec:dekker}

Dekker's algorithm originally appears in \cite{dijkstra1962over}. There is no clear pseudocode given there, so we use the pseudocode from \cite{alagarsamy2003some}, here given as \cref{alg:dekker}.
The algorithm uses a Boolean $\varname{flag}$ per thread, initially $\false$, and a multi-writer register $\varname{turn}$, initially $0$.
{\Atrace} showing that Dekker's algorithm does not satisfy starvation freedom with safe registers is reported in~\cite{buhr2016dekker}.
This same {\trace} can be found by mCRL2:
\begin{itemize}
    \item Thread 0 goes through the algorithm without competition, and starts setting $\varidx{flag}{0}$ to $\false$ in the exit protocol, when it is currently $\true$.
    \item Thread 1 starts the competition and reads $\varidx{flag}{0} = \false$, the new value, on line 2. It can therefore go to the critical section, and set $\varname{turn}$ to $0$ in the exit protocol.
    \item Thread 1 then starts the competition again, now reading $\varidx{flag}{0} = \true$, the old value, on line 2. Since $\varname{turn} = 0$, it goes to line 5 and starts waiting for $\varname{turn}$ to be $1$.
    \item Thread 0 finishes the exit protocol and never re-attempts to enter the critical section.
\end{itemize}
Since thread 0 will never set $\varname{turn}$ to $1$, thread 1 can never escape line 5.
This {\trace} violates deadlock freedom as well as starvation freedom, and is also applicable to regular registers.
The phenomenon where two reads concurrent with the same write return first the new and then the old value is called \emph{new-old inversion}, and is explicitly allowed by Lamport in his definitions of safe and regular registers \cite{Lamport86IPCalg}.
An interesting quality of this {\trace} is that it relies on thread 0 only finitely often executing the algorithm. If we did not define the actions $\noncrit[\tid]$ for all $\tid \in \TID$ to be blockable, this {\trace} would be missed.

In \cite{buhr2016dekker} the following improvements are suggested to make the algorithm ``RW-safe'', i.e., correct with safe registers: on line 5, \textbf{await} $\varname{turn} = i \lor \varidx{flag}{j} = \false$, and on line 8, only write to $\varname{turn}$ if its value would be changed.
Our model checking confirms that with these alterations, starvation freedom is satisfied with both safe and regular registers.

In \cite{nigro2024modeling}, it is claimed that Dekker's algorithm without alterations is correct with non-atomic registers.
Instead of dealing with the spurious violations of liveness properties via completeness criteria, they use the model checking tool UPPAAL \cite{Behrmann2004} to compute the maximum number of times a thread may be overtaken by another thread. They determine that bound to be finite and conclude starvation freedom is satisfied.
However, the deadlock freedom violation observed here and in \cite{buhr2016dekker} shows a thread never gaining access to the critical section while only being overtaken once.
Hence, finding a finite upper bound to the number of overtakes is insufficient to establish deadlock freedom.

As can be observed in \autoref{tab:results}, both the version of Dekker's algorithm presented in \cite{alagarsamy2003some}\pagebreak[3] and the RW-safe version from \cite{buhr2016dekker} are starvation-free with atomic registers under $\justact{\conc_T}{\block}$, but only deadlock-free under $\justact{\conc_S}{\block}$. In both variants, this is because one thread, say $0$, can remain stuck on line 5 trying to perform a read, while the other thread, in this case $1$, repeatedly executes the full algorithm without having to wait on thread $0$, since $\varidx{flag}{0} = \false$. In the process, thread $1$ writes to the variable that thread $0$ is trying to read, meaning this execution is just under $\justact{\conc_S}{\block}$. 

This starvation freedom violation can be easily fixed for the presentation in \cite{alagarsamy2003some}, since the variable that thread $0$ is trying to read on line 5 is $\varname{turn}$. If we alter the algorithm so that a thread only writes to $\varname{turn}$ on line 8 if this would change the value, then starvation freedom is satisfied.
This change is part of the changes suggested in \cite{buhr2016dekker}, yet that version of the algorithm is not starvation-free under $\justact{\conc_S}{\block}$.
This is due to the other change: on line 5, thread $0$ now also has to read $\varidx{flag}{1}$, the value of which thread $1$ does change every time it executes the algorithm. 
This violation therefore cannot be fixed so easily.

\subsection{Szymanski's 3-bit linear wait algorithm}\label{sec:Szymanski}
\vspace{-5pt}

 \begin{algorithm}[hb]
    \caption{Szymanski's 3-bit linear wait algorithm}
    \label{alg:szymanski-3bit}
    \begin{algorithmic}[1]
      \State{$\varidx{a}{i} \writeop \true$}\label{szy-3bit-1}
      \lFor{$j$ \textbf{from} $0$\ \textbf{to}\ $N{-}1$} {\textbf{await} $\varidx{s}{j} = \false$}\label{szy-3bit-2}
      \State{$\varidx{w}{i} \writeop \true$}\label{szy-3bit-3}
      \State{$\varidx{a}{i} \writeop \false$}\label{szy-3bit-4}
      \While{$\varidx{s}{i} = \false$}\label{szy-3bit-5}
        \State{$j \writeop 0$}\label{szy-3bit-6}
        \lWhile{$j < N \land \varidx{a}{j} = \false$}{$j\writeop j+1$}\label{szy-3bit-7}
          \If{$j=N$}\label{szy-3bit-8}
            \State{$\varidx{s}{i} \writeop \true$}\label{szy-3bit-9}
            \State{$j\writeop 0$}         \label{szy-3bit-10}
            \lWhile{$j<N\land \varidx{a}{j} = \false$}{$j\writeop j+1$}\label{szy-3bit-11}
            \If{$j< N$} {$\varidx{s}{i} \writeop \false$}\label{szy-3bit-12}
           \Else\label{szy-3bit-13}
             \State{$\varidx{w}{i} \writeop \false$}\label{szy-3bit-14}
             \lFor{$j$ \textbf{from} $0$\ \textbf{to} $N-1$}{\textbf{await} $\varidx{w}{j} = \false$}\label{szy-3bit-15}
           \EndIf
         \EndIf
         \If{$j<N$}\label{szy-3bit-16}
           \State{$j\writeop 0$}\label{szy-3bit-17}
           \lWhile{$j<N \land (\varidx{w}{j} = \true \lor \varidx{s}{j} = \false)$}{$j\writeop j+1$}\label{szy-3bit-18}
         \EndIf
         \If{$j \neq i \land j<N$}\label{szy-3bit-19}
         \State{$\varidx{s}{i} \writeop \true$}\label{szy-3bit-20}
         \State{$\varidx{w}{i} \writeop \false$}\label{szy-3bit-21}
         \EndIf
     \EndWhile
      \lFor{$j$ \textbf{from} $0$ \textbf{to} $i-1$}{\textbf{await} $\varidx{s}{j} = \false$}\label{szy-3bit-22}
      \State{\textbf{critical section}}\label{szy-3bit-23}
      \State{$\varidx{s}{i} \writeop \false$}\label{szy-3bit-24}
     \end{algorithmic}
  \end{algorithm}
In \cite{Szy90}, Szymanski proposes four mutual exclusion algorithms. Here, we discuss the first: the 3-bit linear wait algorithm, which is claimed to be correct with non-atomic registers.
See \autoref{alg:szymanski-3bit}. Each thread has three Booleans, $\varname{a}$, $\varname{w}$, and $\varname{s}$, all initially $\false$.
{\Atrace} showing a mutual exclusion violation for safe, regular, and atomic registers with three threads for the 3-bit linear wait algorithm is given in \cite{spronck2023process}.

With two threads,\pagebreak[4] we do still get a mutual exclusion violation with safe and regular registers, given in detail in \cite{spronck2023process}, but not with atomic registers.
The violation specifically relies on reading $\varidx{w}{j}$ before $\varidx{s}{j}$ on line 18, so that there can be new-old inversion on $\varidx{s}{j}$, while still obtaining the value of $\varidx{w}{j}$ from before the other thread started writing to $\varidx{s}{j}$.
We therefore considered the alternative, where $\varidx{s}{j}$ is read before $\varidx{w}{j}$ on line 18.
With this change and only two threads, the algorithm satisfies all expected properties.

The observation that Szymanski's 3-bit linear wait algorithm violates mutual exclusion with three threads is also made in \cite{nigro2024Verifying}. 
They suggest a different way to make the algorithm correct for two threads: instead of swapping the reads on line 18, they require a thread to also read $\varidx{w}{j}$ on line 22.
We did not investigate this suggested change further, as we prefer the solution that does not require an additional read.

\subsection{From deadlock freedom to starvation freedom}\label{sec:dftosf}
In \cite[Section 2.2.2]{Raynal13}, an algorithm is presented to turn any mutual exclusion algorithm that satisfies mutual exclusion and deadlock freedom into one that satisfies starvation freedom as well. It is due to Yoah Bar-David (1998) and first appears in \cite{taubenfeld2006synchronization}.
We take the pseudocode from \cite{GG}, where it is proven that this algorithm works for safe, regular and atomic registers. 

To confirm the correctness of the algorithm experimentally, we applied it
to Lamport's 1-bit algorithm \cite{Lamport86Mutex2}, which (by design) does not satisfy starvation freedom at all, and to the RW-safe version of Dekker's algorithm \cite{buhr2016dekker}, where starvation freedom only fails due to writes interfering with reads. 
The results are given in the table with ``DFtoSF''. 
We indeed find that starvation freedom is now satisfied where previously only deadlock freedom was.

In view of space, the pseudocode is given in \autoref{app:algorithms}.

\section{Related work}\label{sec:relwork}
To the best of our knowledge, we are the first to do automatic verification of mutual exclusion algorithms while incorporating both which operations can block each other and the effects of overlapping write operations. The two elements have been considered separately previously.

In \cite{bouwman2020off}, starvation-freedom of Peterson's algorithm with atomic registers is checked using mCRL2 under the justness assumption. Two different concurrency relations are considered, which are similar to our $\conc_S$ and $\conc_{\!A}$.

There have been many formal verifications of mutual exclusion algorithms with atomic registers \cite{benari2008principles,mateescu2013model,groote2021tutorial}.
Non-atomic registers have been covered less frequently, and their verification is often restricted to single-writer registers \cite{lamportHyperook,buhr2016dekker}.
The safety properties of mutual exclusion algorithms with MWMR non-atomic registers were verified using mCRL2 in~\cite{spronck2023process}. In several papers by Nigro, including \cite{nigro2024Verifying,nigro2024modeling}, safety and liveness verification with MWMR non-atomic registers has been done using UPPAAL.

A major drawback of our approach is that we only consider a small number of parallel threads.
There has been much work in the literature on parametrised verification, where the correctness of an algorithm is established for an arbitrary number of threads \cite{bouajjani2000regular,fisman2001beyond,zuck2004model,bloem2016decidability}.
Where these techniques handle liveness, it is under forms of weak or strong fairness, not justness.
To our knowledge, these techniques have not yet been applied in the context of non-atomic registers. 
While several papers on parametrised verification, such as \cite{abdulla2008handling,abdulla2016parameterized}, mention dropping the atomicity assumption, this refers to evaluating existential and universal conditions on all threads in a single atomic operation, rather than atomicity of the registers.

In \cite{AravindH11,Hesselink13a,hesselink2015mutual} and other work by Hesselink, interactive theorem proving with the proof assistant PVS is used to check safety and liveness properties (under fairness) of mutual exclusion algorithms for an arbitrary number of threads with non-atomic registers. To our knowledge, the case of writes overlapping each other has not been covered with this technique.

\section{Conclusion}\label{sec:conclusion}
When it comes to analysing the correctness of algorithms, particularly when considering non-atomic registers, behavioural reasoning is often insufficient.
Mistakes can be subtle, and may depend on edge-cases that are easily overlooked.
Model checking is a solution here; we formally model the threads executing the algorithm, as well as the registers through which they communicate, and the entire state-space is searched for possibly violations of correctness properties.
In this work, we verify a large number of mutual exclusion algorithms using the model checking toolset mCRL2.
We expand on previous work by checking liveness properties -- deadlock freedom and starvation freedom -- in addition to the main safety property.
To circumvent spurious violating paths in our models, we incorporate the completeness criterion justness into our verifications.
We checked algorithms under six different memory models, where a memory model is a combination of a register model (safe, regular, or atomic) and a model of which register access operations can block each other.
The former dimension we capture in the models themselves, by modelling the behaviour of the three types of register.
The latter, we capture in the concurrency relations employed as part of justness.
We found a number of interesting violations of correctness properties, and in some cases could suggest improvements to algorithms to fix these violations.
We find that there are several algorithms that satisfy all three properties for four out of six memory models.
For three threads, this is accomplished by the fixed version of Aravind's BLRU algorithm and Lamport's 3-bit algorithm. If we also consider algorithms for just two threads, then Anderson's algorithm and the fixed version of Szymanski's 3-bit linear wait algorithm also meet this bar. We also considered an algorithm to turn deadlock-free algorithms into starvation-free ones, and experimentally confirmed that it indeed works for Dekker's algorithm made RW-safe and Lamport's 1-bit algorithm.

\bibliography{bib2doi}

\newpage
\appendix

\section{Safe, regular and atomic SWMR registers (an example)}\label{sec:non-atomic-communication}

We illustrate the differences between Lamport's safe, regular and atomic SWMR registers with an example.     
Consider the execution illustrated in \autoref{fig:SWMRexample}, which has three threads, with id's $0$, $1$ and $2$, operating on a register $x$ with domain $\{0,1,2\}$.
    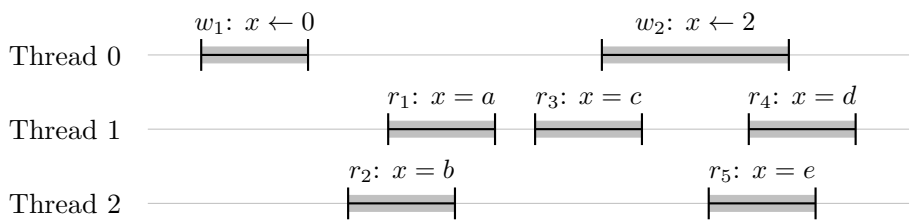
\begin{figure}[ht]
    \centering
    \begin{tikzpicture}[initial text=,inner sep=0pt, outer sep=0pt, minimum size=0pt, node distance=\distY pt and \distX pt]
        \node [draw=none] (0-name) {\large Thread $0$};

        \node [draw=none, right=of 0-name, xshift=\lineGap pt] (0-line-start) {};
        \node [draw=none, below=of 0-line-start] (1-line-start) {};
        \node [draw=none, below=of 1-line-start] (2-line-start) {};

        \node [draw=none, left=of 1-line-start, xshift=-\lineGap pt] (1-name) {\large Thread $1$};
        \node [draw=none, left=of 2-line-start, xshift=-\lineGap pt] (2-name) {\large Thread $2$};

        \node [draw=none, right=of 0-line-start, xshift=-\lineCut pt] (0-op1-s) {};
        \node [draw=none, right=of 0-op1-s] (0-op1-e) {};
        \node (0-op1-m) at ($(0-op1-s)!0.5!(0-op1-e)$) {};
        \node [draw=none, below=of 0-op1-e, xshift=\shiftXbig pt] (1-op2-s) {};
        \node [draw=none, right=of 1-op2-s] (1-op2-e) {};
        \node (1-op2-m) at ($(1-op2-s)!0.5!(1-op2-e)$) {};
        \node [draw=none, below=of 1-op2-s, xshift=-\shiftX pt] (2-op3-s) {};
        \node [draw=none, right=of 2-op3-s] (2-op3-e) {};
        \node (2-op3-m) at ($(2-op3-s)!0.5!(2-op3-e)$) {};
        \node [draw=none, above=of 2-op3-e, xshift= \shiftXbig pt] (1-op4-s) {};
        \node [draw=none, right=of 1-op4-s] (1-op4-e) {};
        \node (1-op4-m) at ($(1-op4-s)!0.5!(1-op4-e)$) {};
        \node [draw=none, above=of 1-op4-e, xshift=-\shiftX pt] (0-op5-s) {};
        \node [draw=none, right=of 0-op5-s, xshift=\shiftXbig pt] (0-op5-e) {}; 
        \node (0-op5-m) at ($(0-op5-s)!0.5!(0-op5-e)$) {};
        \node [draw=none, below=of 0-op5-e, xshift=-\shiftX pt] (1-op6-s) {};
        \node [draw=none, right=of 1-op6-s] (1-op6-e) {};
        \node (1-op6-m) at ($(1-op6-s)!0.5!(1-op6-e)$) {};
        \node [draw=none, below=of 1-op6-s, xshift=-\shiftX pt] (2-op7-s) {};
        \node [draw=none, right=of 2-op7-s] (2-op7-e) {};
        \node (2-op7-m) at ($(2-op7-s)!0.5!(2-op7-e)$) {};

        \node [draw=none, right=of 1-op6-e, xshift=-\lineCut pt] (1-line-end) {};
        \node [draw=none, above=of 1-line-end] (0-line-end) {};
        \node [draw=none, below=of 1-line-end] (2-line-end) {};

        \path[-, opacity=0.25]
            (0-line-start) edge (0-line-end)
            (1-line-start) edge (1-line-end)
            (2-line-start) edge (2-line-end);

        \draw[thick]
            \VBar{\barX pt}{(0-op1-s)}
            \VBar{\barX pt}{(0-op1-e)}
            \VBar{\barX pt}{(1-op2-s)}
            \VBar{\barX pt}{(1-op2-e)}
            \VBar{\barX pt}{(2-op3-s)}
            \VBar{\barX pt}{(2-op3-e)}
            \VBar{\barX pt}{(1-op4-s)}
            \VBar{\barX pt}{(1-op4-e)}
            \VBar{\barX pt}{(0-op5-s)}
            \VBar{\barX pt}{(0-op5-e)}
            \VBar{\barX pt}{(1-op6-s)}
            \VBar{\barX pt}{(1-op6-e)}
            \VBar{\barX pt}{(2-op7-s)}
            \VBar{\barX pt}{(2-op7-e)};

        \filldraw[opacity=0.25]
            \Above{\barY pt}{(0-op1-s)} rectangle \Below{\barY pt}{(0-op1-e)}
            \Above{\barY pt}{(1-op2-s)} rectangle \Below{\barY pt}{(1-op2-e)}
            \Above{\barY pt}{(2-op3-s)} rectangle \Below{\barY pt}{(2-op3-e)}
            \Above{\barY pt}{(1-op4-s)} rectangle \Below{\barY pt}{(1-op4-e)}
            \Above{\barY pt}{(0-op5-s)} rectangle \Below{\barY pt}{(0-op5-e)}
            \Above{\barY pt}{(1-op6-s)} rectangle \Below{\barY pt}{(1-op6-e)}
            \Above{\barY pt}{(2-op7-s)} rectangle \Below{\barY pt}{(2-op7-e)};

        \path[-,thick]
            (0-op1-s) edge (0-op1-e)        
            (1-op2-s) edge (1-op2-e)  
            (2-op3-s) edge (2-op3-e) 
            (1-op4-s) edge (1-op4-e)     
            (0-op5-s) edge (0-op5-e)    
            (1-op6-s) edge (1-op6-e)    
            (2-op7-s) edge (2-op7-e);      

        \draw[]
            (0-op1-m) node[above,yshift=\textY pt] {$w_1$: $x \writeop 0$}
            (1-op2-m) node[above,yshift=\textY pt] {$r_1$: $x = a$}
            (2-op3-m) node[above,yshift=\textY pt] {$r_2$: $x = b$}
            (1-op4-m) node[above,yshift=\textY pt] {$r_3$: $x = c$}
            (0-op5-m) node[above,yshift=\textY pt] {$w_2$: $x \writeop 2$}
            (1-op6-m) node[above,yshift=\textY pt] {$r_4$: $x = d$}
            (2-op7-m) node[above,yshift=\textY pt] {$r_5$: $x = e$};
    \end{tikzpicture}
    \caption{Example behaviour of SWMR safe, regular and atomic registers.}
    \label{fig:SWMRexample}
    \end{figure}

    Let us consider which values are possible for $a, b, c, d$ and $e$, depending on which type of register $x$ is.

    Both $r_1$ and $r_2$ are not concurrent with any write, hence both reads will return the last written value of the register, regardless of the type of $x$.
    Therefore, we have $a=0$ and $b=0$ for safe, regular and atomic registers. 

    If $x$ is a safe register, then $c, d$ and $e$ may all be any value of $0, 1$ or $2$. These three assignments are independent of each other, so all 27 different possible combinations of choices for $c, d$ and $e$ are allowed.

    If $x$ is a regular register, then $c, d$ and $e$ may all be any value of $0$ or $2$: the old value or the value that is being written concurrently. Here too, the three returned values are independent, so there are 8 different possibilities. In particular, it can be that $c{=}2$ but $d{=}0$; this possibility is called new-old inversion.

    If $x$ is an atomic register, then the values returned by the reads must reflect some order on all operations in the execution, and this order must respect the real-time ordering of the operations. We interpret the figure as showing real time, meaning that, for example, $r_1$ must be ordered before $w_2$ since $r_1$ ends before $w_2$ begins. Consequently, there are only five combinations of values of $c, d$ and $e$ allowed. If $c$ is $0$, then any combination of $d$ and $e$ being $0$ or $2$ is accepted, since $r_3$ can be ordered before $w_2$ without affecting whether $r_4$ and $r_5$ are before or after $w_2$. If instead we have $c = 2$, then $r_3$ must be ordered after $w_2$ and hence $r_4$ and $r_5$, which in real time come after $r_3$, must also be ordered after $r_2$, giving us $d=e=2$.

\section{Register models}\label{sec:registers}
In this appendix, we expand on the material of \autoref{sec:registers summary}, by giving the formal process-algebraic models of the registers.
First, we present the general structure of such a model and explain how to translate it to an LTS.
We then give the three models, as well as the definitions of all the access and update functions.
We leave the status object abstract: it is not necessary to define the data structure itself, as long as it is clear what information can be retrieved from it.

\subsection{Structure and functions}\label{sec:structure}
Recall that we use $\TID$ for the thread identifiers, $\RID$ for register identifiers, and that for every $\rid \in \RID$, the set $\Data$ contains all values that $\rid$ may hold.
Additionally, we use a status object as the finite memory of a register, the set of possible statuses being $\StatusAll$.

We define the following structure, shared by all three register models.
Let $\regtype \mathbin\in \{\safRep, \regRep, \atoRep\}$, then each model looks as follows, for some number $n$:
\begin{equation}\label{eq:structure}
    \Reg[\regtype](\rid : \RID, s: \StatusAll) = \sum_{\tid \in \TID}\sum_{\data \in \Data}\sum_{0 \leq j < n} (c_j(s, \tid, \data) \rightarrow a_j(\tid,\data) \cdot \Reg[\regtype](\rid, u_j(s,\tid,\data)))
\end{equation}
This represents a register with id $\rid$, that tracks its status with $s$.
The process first sums over $\tid \in \TID$ and $\data \in \Data$, allowing interaction by all threads and with all possible data parameters.
Furthermore, it has $n$ summands, each of the form $c_j(s,\tid,\data) \rightarrow a_j(\tid,\data) \cdot \Reg[\regtype](\rid, u_j(s,\tid,\data))$ where  $c_j(s,\tid,\data)$ is a Boolean condition, $a_j(\tid,\data)$ is an action, and $u_j$ is an \emph{update function} that takes $s$, $\tid$ and $\data$ and returns the updated status object $s' \in \StatusAll$.
Such a process equation gives rise to an LTS. 
Given a predefined initial state $\initstate \in \StatusAll$, the LTS of $\Reg[\regtype](\rid, \initstate)$ is:
\begin{equation}\label{eq:LTS}
    (\StatusAll, \!\!\bigcup_{0 \leq j < n}\!\!\{a_j(\tid,\data) \mid \tid \mathbin\in \TID, \data \mathbin\in \Data\}, \initstate, \!\!\bigcup_{0 \leq j < n}\!\!\{(s, a_j(\tid,\data), u_j(s,\tid,\data)) \mid \tid \mathbin\in \TID, \data \mathbin\in \Data \land c_j(s,\tid,\data)\})
\end{equation}

We now describe the initial state and the update functions.
As stated above, we do not wish the go into the implementation details of the status object.
Instead, we define a collection of \emph{access functions} which retrieve information from the current state.
The initial state, as well as the effects of the update functions, are then defined by how they alter the results of the access functions.
We use the following access functions, which are local to any given register $r$:
\begin{itemize}
    \item $\truevalsym: \StatusAll \rightarrow \Data$, the value that is currently stored in the register.
    \item $\readerssym: \StatusAll \rightarrow 2^{\TID}$, the set of thread id's of threads that have invoked a read operation on this register that has not yet had its response.
    \item $\writerssym: \StatusAll \rightarrow 2^{\TID}$, the set of thread id's of threads that have invoked a write operation on this register that has not yet had its response.
    \item $\pendingsym: \StatusAll \rightarrow 2^{\TID}$, the set of thread id's of threads that have invoked an operation that has not been ordered yet. Only used by the regular and atomic models.
    \item $\valssym: \StatusAll \times \TID \rightarrow \Data$, a mapping that allows us to record a single data value per thread, for instance used to remember what value was passed with a start write action.
    \item The predicate $\overlapsym$ on $\StatusAll \times \TID$, which stores whether an ongoing read or write operation of a thread has encountered an overlapping write. Only used by the safe model.
    \item $\posvalsym: \StatusAll \times \TID \rightarrow 2^{\Data}$, a mapping that stores a set of values per thread, representing the possible return values of an ongoing read by a thread. Only used by the regular model.
\end{itemize}
The initial state $\initstate$ is defined as follows, for all $\tid \in \TID$ and a pre-defined initial value $\data_{\initstate}$:\vspace{-1ex}
\[\begin{array}{r@{~}c@{~}lr@{~}c@{~}lr@{~}c@{~}lr@{~}c@{~}l}
    \trueval{\initstate} &=& \data_{\initstate} & \writers{\initstate} &=& \emptyset & \vals{\initstate, \tid} &=& \data_{\initstate}  & \posval{\initstate, \tid} &=& \emptyset\\
    \readers{\initstate} &=& \emptyset & \pending{\initstate} &=& \emptyset & \overlap{\initstate, \tid} &=& \false & &\\[-1ex]
\end{array}\]
The initial value $\data_{\initstate}$ of a register depends on the modelled algorithm.

We now define the update functions in a similar way to the initial state: by showing how the return values of the access functions are altered by the update function.
Each update function corresponds to an action and is applied when that action occurs; if the action's name is $a$, the update function is called $\mathit{ua}$.
Not every update function uses the data parameter that is passed to it according to (\ref{eq:structure}); in these cases we only give the thread id and status parameters.
If an access function is not mentioned, then its return value after the update is the same as before.
Given an arbitrary state $s \in \StatusAll$, thread id $\tid \in \TID$ and data value $\data \in \Data$:\pagebreak[3]\\
If $s' = \usr{s, \tid}$, then:\vspace{-1.5ex}
\begin{align*}
    \readers{s'} &= \readers{s} \cup \{\tid\}\\
    \pending{s'} &= \pending{s} \cup \{\tid\}\\
    \overlap{s', \tid} &= (\writers{s} > 0)\\
    \posval{s', \tid} &= \{\trueval{s}\} \cup \{\data' \mid \exists_{\tidtwo \in \writers{s}}.\vals{s, \tidtwo} = \data'\}
\end{align*}
If $s' = \ufr{s, \tid}$, then:\vspace{-1.5ex}
\begin{align*}
    \readers{s'} &= \readers{s} \setminus \{\tid\}
\end{align*}
If $s' = \usw{s, \tid, \data}$, then for all $\tidtwo \neq \tid$:\vspace{-1.5ex}
\begin{align*}
    \writers{s'} &= \writers{s} \cup \{\tid\}\\
    \pending{s'} &= \pending{s} \cup \{\tid\}\\
    \vals{s', \tid} &= \data\\
    \overlap{s', \tid} &= (\writers{s} > 0)\\
    \overlap{s', \tidtwo} &= \true\\
    \posval{s', \tidtwo} &= \posval{s, \tidtwo} \cup \{\data\}
\end{align*}
If $s' = \ufw{s, \tid, \data}$, then:\vspace{-1.5ex}
\begin{align*}
    \trueval{s'} &= \data\\
    \writers{s'} &= \writers{s} \setminus \{\tid\}
\end{align*}
If $s' = \uor{s, \tid}$, then:\vspace{-1.5ex}
\begin{align*}
    \pending{s'} &= \pending{s} \setminus \{\tid\}\\
    \vals{s', \tid} &= \trueval{s}
\end{align*}
If $s' = \uow{s, \tid, \data}$, then:\vspace{-1.5ex}
\begin{align*}
    \trueval{s'} &= \data\\
    \pending{s'} &= \pending{s} \setminus \{\tid\}
\end{align*}

\noindent
These formal definitions correspond to the intuitive descriptions given above. Of note is that $\overlap{s', \tidtwo}$ is set to $\true$ whenever a thread $\tid \neq \tidtwo$ starts a write, even if $\tidtwo$ is not actively reading or writing. This is done for simplicity of the definition: when $\tidtwo$ starts reading or writing, it will reset its own $\overlapsym$ to the correct value, depending on whether there is an overlapping write active at that point.
Something similar is done with $\posvalsym$: when a write is started, its value gets added to the $\posvalsym$ sets of every other thread, even if they are not actively reading. When a thread starts reading, it sets its own $\posvalsym$ correctly.

We now give the three register models.

\subsection{Safe MWMR registers}

See \autoref{fig:procsafe} for the process equation representing our safe register model.
\begin{figure}[ht]
    \begin{multline*}
        \SReg(\rid: \RID, s:\StatusAll) = \\
        \sum_{\tid\in\TID}\sum_{\data \in \Data}
        \left(\begin{array}{ll}
            & (\tid\notin(\readers{s} \cup \writers{s})) \then
                 \startread[\tid,\rid]\co\SReg(\rid,\usr{s,\tid}) \\
          + & (\tid\notin(\readers{s} \cup \writers{s})) \then
                 \startwrite[\tid,\rid]{\data}\co\SReg(\rid,\usw{s,\tid,\data})\\
          + & (\tid\in\readers{s}\wedge\neg\overlap{s,\tid}) \then
                 \finishread[\tid,\rid]{\trueval{s}}\co\SReg(\rid,\ufr{s,\tid}) \\
          + & (\tid\in\readers{s}\wedge\overlap{s,\tid}) \then
                 \finishread[\tid,\rid]{\data}\co\SReg(\rid,\ufr{s,\tid})\\
          + & (\tid\in\writers{s}\wedge\neg\overlap{s,\tid}) \then
                 \finishwrite[\tid,\rid]\co\SReg(\rid,\ufw{s,\tid,\vals{s, \tid}})\\
          + & (\tid\in\writers{s}\wedge\overlap{s,\tid}) \then
                 \finishwrite[\tid,\rid]\co\SReg(\rid,\ufw{s,\tid,\data})
        \end{array}
        \right)
    \end{multline*} \caption{Safe register process}\label{fig:procsafe}
\end{figure}

The correspondence between the process and the four rules given for MWMR safe registers in \autoref{sec:safe} is rather direct: the first two summands allow a thread that is not currently reading or writing to begin a read or a write, and the remaining four each represent one of the four rules, in order. 
Note that, in the case of a finish write without overlap, we use $\valssym$ to retrieve which value this thread intended to write so that the register state can be appropriately updated.

\subsection{Regular MWMR registers}
See \autoref{fig:procregular} for the process equation representing our regular register model.
Recall that regular registers use the order write action to generate a global ordering on all write operations on a register on the fly.

\begin{figure}[hb]
    \begin{multline*}
        \RReg(\rid: \RID, s:\StatusAll) = \\
        \sum_{\tid\in\TID}\sum_{\data \in \Data}
        \left(\begin{array}{ll}
            & (\tid\notin(\readers{s} \cup \writers{s})) \then
                 \startread[\tid,\rid]\co\RReg(\rid,\usr{s,\tid}) \\
          + & (\tid\notin(\readers{s} \cup \writers{s})) \then
                 \startwrite[\tid,\rid]{\data}\co\RReg(\rid,\usw{s,\tid,\data})\\
          + & (\tid\in\readers{s} \land \data \in \posval{s,\tid}) \then
                 \finishread[\tid,\rid]{\data}\co\RReg(\rid,\ufr{s,\tid})\\
          + & (\tid \in \writers{s} \land \tid\in\pending{s}) \then \orderwrite[\tid,\rid]\co\RReg(\rid,\uow{s,\tid,\vals{s,\tid}})\\
          + & (\tid\in\writers{s}\wedge \tid\notin\pending{s}) \then
                 \finishwrite[\tid,\rid]\co\RReg(\rid,\ufw{s,\tid,\trueval{s}})
        \end{array}
        \right)
    \end{multline*}
    \caption{Regular register process}\label{fig:procregular}
\end{figure}

Similar to the safe register process, the first two summands are merely allowing an idle thread to begin a read or write operation. 
The third summand corresponds to finishing a read by returning a value that is in the set of possible values to be returned for this read.
Recall that, by the definition of $\posvalsym$, this set is constructed as follows: when a read starts, the set is initialised to the current stored value of the register and the intended write value of every active write. 
Subsequently, whenever a write occurs, its intended value is added to the set.
This way, at the finish read, the set will contain exactly those values that the read could return.
The fourth summand allows the occurrence of the ordering action; at this time the intended value of the write, which was temporary stored in the access function $\valssym$, is logged as the stored value of the register. 
The final summand describes the ending of a write operation. The $\ufwsym$ update function will set the stored value to whatever data value is passed. In this case, the stored value should not change at the finish write, since it was already changed at the order write. Hence, we simply pass $\trueval{s}$.
Note that we use $\pendingsym$ to determine if the order write action still has to occur.

\subsection{Atomic MWMR registers}
See \autoref{fig:procatomic} for our model of MWMR atomic registers.
Recall that, in addition to the order write action, the atomic register model also uses the order read action. This way, it generates an ordering on all operations on a register.

\begin{figure}[ht]
    \begin{multline*}
        \AReg(\rid:\RID, s:\StatusAll) = \\
        \sum_{\tid\in\TID}\sum_{\data \in \Data}
        \left(\begin{array}{ll}
            & (\tid\notin(\readers{s}\cup\writers{s})) \then
                 \startread[\tid,\rid]\co\AReg(\rid,\usr{s,\tid}) \\
          + & (\tid\notin(\readers{s}\cup\writers{s})) \then
                 \startwrite[\tid,\rid]{\data}\co\AReg(\rid,\usw{s,\tid,\data})\\
          + & (\tid\in\readers{s} \wedge \tid \in \pending{s}) \then\orderread[\tid,\rid]\co\AReg(\rid,\uor{s,\tid}) \\
          + & (\tid\in\writers{s}\wedge\tid \in \pending{s}) \then \orderwrite[\tid,\rid]\co\AReg(\rid,\uow{s,\tid,\vals{s,\tid}}) \\
          + & (\tid\in\readers{s}\wedge\tid \notin \pending{s})\then \finishread[\tid,\rid]{\vals{s,\tid}}\co\AReg(\rid,\ufr{s,\tid}) \\
          + & (\tid\in\writers{s}\wedge\tid \notin \pending{s})\then\finishwrite[\tid,\rid]\co\AReg(\rid,\ufw{s,\tid,\trueval{s}})
          \end{array}
        \right)
    \end{multline*}
    \caption{Atomic register process}\label{fig:procatomic}
\end{figure}

Summands one, three and five are the invocation, ordering and response of a read operation respectively.
Similarly, summands two, four and six are the invocation, ordering and response of a write operation.
We use $\pendingsym$ to determine whether an operation's order action still has to occur. When a read is ordered, we save the current stored value of the register in $\valssym$, so that this value can be returned when the read ends.
For writes, we already update $\truevalsym$ when the write is ordered, meaning that when the write ends we do not want to change $\truevalsym$ further and just pass the current value back.

\section{Thread consistency proof}\label{app:thr-consist-proof}
In this appendix, we prove \autoref{lem:thr-consist}. We first provide an alternate perspective on the definitions of the access functions given in \autoref{sec:registers}, which make it clearer how a transition affects the status object. We then prove a supporting lemma, and finally prove \autoref{lem:thr-consist} itself.

When it comes to applying the definitions of the access functions, we find it easier to reason from the perspective of transitions: given a transition $s \xrightarrow{a} s'$ in an LTS that is derived from a register model as described in \autoref{sec:registers}, the relation of the output of an access function on $s'$ to the output on $s$, depending on the action $a$.
This perspective can be derived directly from the definitions of the update functions, and the tight coupling between update functions and actions.
\begin{observation}\rm\label{lem:access}
Consider the LTS $(\states, \actionset, \initstate, \transrel)$ associated with a register $\rid \in \RID$. Let $s, s' \in \states$ and $a \in \actionset$ such that $s \xrightarrow{a} s'$, and let $\tid \in \TID$. Then the following equations hold:
{\allowdisplaybreaks
\begin{align*}
    \trueval{s'} &= \begin{cases}
        \data & \text{if $a =\finishwrite[\tid,\rid]$ and $ s' = \ufw{s,\tid,\data}$ for $\tid \in \TID, \data \in \Data$}\\
        \data & \text{if $a =\orderwrite[\tid,\rid]$ and $s' = \uow{s,\tid,\data}$ for $\tid \in \TID, \data \in \Data$}\\
        \trueval{s} & \text{otherwise}
    \end{cases}\\
\hspace{-8pt}
    \readers{s'} &= \begin{cases}
        \readers{s} \cup \{\tid\} & \text{if $a =\startread[\tid,\rid]$ for $\tid \in \TID$}\\
        \readers{s} \setminus \{\tid\} & \text{if $a =\finishread[\tid,\rid]{\data}$ for $\tid \in \TID, \data \in \Data$}\\
        \readers{s} & \text{otherwise}
    \end{cases}\\
\hspace{-8pt}
    \writers{s'} &= \begin{cases}
        \writers{s} \cup \{\tid\} & \text{if $a = \startwrite[\tid,\rid]{\data}$ for $\tid \in \TID, \data \in \Data$}\\
        \writers{s} \setminus \{\tid\} & \text{if $a = \finishwrite[\tid,\rid]$ for $\tid \in \TID$}\\
        \writers{s} & \text{otherwise}
    \end{cases}\\
\hspace{-8pt}
    \pending{s'} &= \begin{cases}
        \pending{s} \cup \{\tid\} & \text{if $a = \startread[\tid,\rid]$ or $a = \startwrite[\tid,\rid]{\data}$ for $\tid \in \TID, \data \in \Data$}\\
        \pending{s} \setminus \{\tid\} & \text{if $a = \orderread[\tid,\rid]$ or $a = \orderwrite[\tid,\rid]$ for $\tid \in \TID$}\\
        \pending{s} & \text{otherwise}
    \end{cases}\\
\hspace{-8pt}
    \vals{s', \tid} &= \begin{cases}
        \data & \text{if $a = \startwrite[\tid,\rid]{\data}$ for $\data \in \Data$}\\
        \trueval{s} & \text{if $a = \orderread[\tid,\rid]$}\\
        \vals{s, \tid} & \text{otherwise}        
    \end{cases}\\
\hspace{-8pt}
    \overlap{s', \tid} &= \begin{cases}
        (\writers{s} > 0) & \text{if $a = \startread[\tid,\rid]$ or $a = \startwrite[\tid,\rid]{\data}$ for $\data \in \Data$}\\
        \true & \text{if $a = \startwrite[\tidtwo,\rid]{\data}$ for $\tidtwo\in\TID, \tidtwo \neq \tid$ and $\data \in \Data$}\\
        \overlap{s, \tid} & \text{otherwise}
    \end{cases}\\
\hspace{-8pt}
    \posval{s', \tid} &= \begin{cases}
        \{\trueval{s}\} \cup \{\data \mid \exists_{\tidtwo \in \writers{s}}:\vals{s, \tidtwo} = d\} & \text{if $a = \startread[\tid,\rid]$}\\
        \posval{s, \tid} \cup \{\data\} & \hspace{-35pt}\text{if $a = \startwrite[\tidtwo,\rid]{\data}$ for $\tidtwo\in\TID, \tid\neq\tidtwo, \data \in \Data$}\\
        \posval{s, \tid} & \hspace{-35pt}\text{otherwise}
    \end{cases}
\end{align*}
}
\end{observation}

We use the following lemma in our proof of \autoref{lem:thr-consist}.
\begin{lemma}\rm\label{lem:overlap}
    Let $s$ be the state of a safe register that is reachable from its initial state, and let $\tid \in \TID$ be a thread identifier. If $\tid \in \readers{s}$ and $\writers{s} \neq \emptyset$, then $\overlap{s,\tid}$.
\end{lemma}
\begin{proof}
    Let $\rid$ be a safe register process with LTS $(\states, \actionset, \initstate, \transrel)$ and let $s \in \states$ be a state reachable from $\initstate$. Let $\tid$ be an element of $\TID$ and assume that $\tid \in \readers{s}$ and $\writers{s} \neq \emptyset$. Then there exists $\tidtwo \in \TID$ such that $\tidtwo \in \writers{s}$. Since a thread cannot read and write simultaneously, we know that $\tid \neq \tidtwo$. We prove $\overlap{s,\tid}$. Recall that $\readers{\initstate} = \writers{\initstate} = \emptyset$. Since $\tid \in \readers{s}$ and $\tidtwo \in \writers{s}$, it follows from the definitions of $\readerssym$ and $\writerssym$ that $s \neq \initstate$ and that every path from $\initstate$ to $s$ must contain occurrences of $\startread[\tid,\rid]$ and $\startwrite[\tidtwo,\rid]{\data}$ for some $\data \in \Data$. Let $\pi$ be an arbitrary path from $\initstate$ to $s$, and consider the last occurrence of $\startread[\tid,\rid]$ on $\pi$, and the last occurrence of $\startwrite[\tidtwo,\rid]{\data}$ for any $\data \in \Data$ on $\pi$. If after $\startread[\tid,\rid]$ there were an occurrence of $\finishread[\tid,\rid]{\data'}$ for some $\data' \in \Data$ along $\pi$, then $\tid \notin \readers{s}$, so there is no such occurrence. Similarly, if there were an occurrence of $\finishwrite[\tidtwo,\rid]$ after $\startwrite[\tidtwo,\rid]{\data}$ on $\pi$, then $\tidtwo \notin \writers{s}$, so there is no such $\finishwrite[\tidtwo,\rid]$.
    We do a case distinction on whether $\startread[\tid,\rid]$ occurs before $\startwrite[\tidtwo,\rid]{\data}$ or vice versa.
    \begin{itemize}
        \item Suppose $\startread[\tid,\rid]$ is before $\startwrite[\tidtwo,\rid]{\data}$. Let $s_b \startwrite[\tidtwo,\rid]{\data} s_a$ be the fragment $\pi$ with the last occurrence of $\startwrite[\tidtwo,\rid]{\data}$. Then $\overlap{s_a,\tid} = \true$ according to the definition of $\overlapsym$ as given in \autoref{lem:access}. Assume towards a contradiction that $\neg\overlap{s,\tid}$ is true. Then on the suffix $\pi'$ of $\pi$ between $s_a$ and $s$, there must be an occurrence of $\startread[\tid,\rid]$ or $\startwrite[\tid,\rid]{\data'}$ for some $\data' \in \Data$ from a state $s_c$ such that $\writers{s_c} = \emptyset$. However, since there is no $\finishwrite[\tidtwo,\rid]$  on $\pi'$, $\tidtwo \in \writers{s_c}$ for all states $s_c$ on $\pi'$. Hence, we have reached a contradiction and $\overlap{s,\tid}$ must be true.
        
        \item Suppose $\startwrite[\tidtwo,\rid]{\data}$ is before $\startread[\tid,\rid]$. Let $s_b \startread[\tid,\rid] s_a$ be the fragment of $\pi$ with the last occurrence of $\startread[\tid,\rid]$. Since there is no $\finishwrite[\tidtwo,\rid]$ after $\startwrite[\tidtwo,\rid]{\data}$ on $\pi$, $\tidtwo \in \writers{s_b}$. It follows from \autoref{lem:access} that $\overlap{s_a,\tid}$. Assume towards a contradiction that $\neg\overlap{s,\tid}$ is true. Then on the suffix $\pi'$ of $\pi$ between $s_a$ and $s$, $\startread[\tid,\rid]$ or $\startwrite[\tid,\rid]{\data'}$ occurs for some $\data'\in\Data$ from a state $s_c$ such that $\writers{s_c} = \emptyset$. But since $\tidtwo \in \writers{s_c}$ for all $s_c$ along $\pi'$, we have a contradiction. Therefore, $\overlap{s,\tid}$ must be true.
    \end{itemize}
    In both cases, we have shown that $\overlap{s,\tid}$. 
\end{proof}

\pagebreak[2]
We now provide the proof of \autoref{lem:thr-consist}.
\thrconsist*
\begin{proof}
    \hypertarget{thrconsist}{Let $M = (\states, \actionset, \initstate, \transrel, \thrsym, \regsym)$ be a thread-register model constructed as the parallel composition $P_1 \parcomp \ldots \parcomp P_k$ with $k = |\TID| + |\RID|$, where $P_1$ to $P_{|\TID|}$ are thread LTSs and $P_{|\TID| + 1}$ to $P_k$ are register LTSs. Let $P_x = (\states_x, \actionset_x, \initstate_x, \transrel_x)$ for all $1 \leq x \leq k$. Let $\#$ be a mapping from thread and register id's to natural numbers, which yields the index of that thread or register in the parallel composition.}

    Let $s = (s_1, \ldots, s_k)$ be a state in $\states$ and $a$ an action in $\actionset$ such that $a$ is enabled in $s$ and there exists an action $b \in \actionset$ and a state $s' = (s_1', \ldots, s_k') \in \states$ such that $\thrmap{a} \neq \thrmap{b}$ and $(s, b, s') \in \transrel$. We prove that $a$ is enabled in $s'$.
    
    Let $\thrmap{a} = \tid$ and $\regmap{a} = \rid$, with $\tid \in \TID$ and $\rid \in \RID \cup \{\undefsymb\}$.
    From the definition of parallel composition, specifically the construction of $\transrel$, it follows that the target state of $(s, b, s')$ can only differ from $s$ for the parallel components that are involved in $b$: $P_{\idx{\thrmap{b}}}$ and, if $\regmap{b} \neq \undefsymb$, $P_{\idx{\regmap{b}}}$. Since $\thrmap{b} \neq \tid$, $s_{\idx{\tid}} = s'_{\idx{\tid}}$.
    We do a case distinction on whether $\regmap{a} = \undefsymb$. 
    \begin{itemize}
        \item If $\regmap{a} = \undefsymb$, then $a$ is in $\thrlocacts_{\tid}$. Since these actions occur only in $\actionset_{\idx{\tid}}$, we merely need to show that $a$ is enabled in $s'_{\idx{\tid}}$ to establish that $a$ is enabled in $s'$. Since $a$ is enabled\linebreak[3] in $s$, it is enabled in $s_{\idx{\tid}}$. We have established that $s_{\idx{\tid}} \mathbin= s'_{\idx{\tid}}$, so $a$ is enabled in $s'$. 

        \item If $\regmap{a} \neq \undefsymb$, then $a$ is not a thread local action, and must be a register interface or register local action. In both cases, it is an action in $\actionset_{\idx{\rid}}$. If it is a register local action, it solely remains to prove that $a$ is enabled in $s'_{\idx{\rid}}$; if it is a register interface action, we must also prove that $a$ is enabled in $s'_{\idx{\tid}}$. The latter follows immediately, in the case that $a$ is a register interface action, from the observation that $s_{\idx{\tid}} = s'_{\idx{\tid}}$ and $a$ being enabled in $s_{\idx{\tid}}$. In both cases, it therefore suffices to prove that $a$ is enabled in $s'_{\idx{\rid}}$.
        From the construction of $\transrel$ it follows that $a$ is enabled in $s_{\idx{\rid}}$. Note that if $\regmap{b} \neq \rid$, then $s_{\idx{\rid}} = s'_{\idx{\rid}}$ and so $a$ is trivially enabled in $s'_{\idx{\rid}}$. Therefore, we continue under the assumption that $\regmap{b} = \rid$. Since $b$ is the action label of a transition from $s$ to $s'$, and $\regmap{b} = \rid \neq \undefsymb$, we know that $(s_{\idx{\rid}}, b, s'_{\idx{\rid}}) \in \transrel_{\idx{\rid}}$. 
        We do a further case distinction on what action $a$ is:
        \begin{itemize}
            \item If $a = \startread[\tid,\rid]$ or $a=\startwrite[\tid,\rid]{\data}$ for some $\data \in \Data$, then for all three register models, it follows that $\tid \notin (\readers{s_{\idx{\rid}}} \cup \writers{s_{\idx{\rid}}})$. We need to show that $\tid \notin (\readers{s'_{\idx{\rid}}} \cup \writers{s'_{\idx{\rid}}})$ as well. From $\thrmap{b} \neq \tid$ it follows that $b \neq \startread[\tid,\rid]$ and $b \neq \startwrite[\tid,\rid]{\data'}$ for any $\data' \in \Data$. We know by \autoref{lem:access} that if $\tid \notin \readers{s_{\idx{\rid}}}$ then $\tid \notin \readers{s'_{\idx{\rid}}}$ and that if $\tid \notin \writers{s_{\idx{\rid}}}$ then $\tid \notin \writers{s'_{\idx{\rid}}}$. Therefore, $\tid \notin (\readers{s'_{\idx{\rid}}} \cup \writers{s'_{\idx{rid}}})$. 
            Hence, it follows for all three register models that $a$ is enabled in $s'_{\idx{\rid}}$ and thus enabled in $s'$. 

            \item If $a = \finishread[\tid,\rid]{\data}$ for some $\data \in \Data$, then to prove that $a$ is enabled in $s'_{\idx{\rid}}$, we do a case distinction on which type of register $\rid$ is:
            \begin{itemize}
                \item If $\rid$ is a safe register, then it follows from $a$ being enabled in $s_{\idx{\rid}}$ that $\tid \in \readers{s_{\idx{\rid}}}$ and either $\neg \overlap{s_{\idx{\rid}}, \tid}$ and $\data =\trueval{s_{\idx{\rid}}}$, or $\overlap{s_{\idx{\rid}}, \tid}$. To show that $a$ is enabled in $s'_{\idx{\rid}}$, we need to show $\tid \in \readers{s'_{\idx{\rid}}}$, and $\data = \trueval{s'_{\idx{\rid}}}$ or $\overlap{s'_{\idx{\rid}}, \tid}$.\linebreak[3] From $\thrmap{b} \neq \tid$ it follows that $b \neq \finishread[\tid,\rid]{\data'}$ for some $\data' \in \Data$. Hence, it follows from \autoref{lem:access} that if $\tid \in \readers{s_{\idx{\rid}}}$ then $\tid \in \readers{s'_{\idx{\rid}}}$. Additionally, if $\overlap{s_{\idx{\rid}}, \tid}$ then $\neg\overlap{s'_{\idx{\rid}}, \tid}$ would only be possible if $b = \startread[\tid,\rid]$ or $b = \startwrite[\tid,\rid]{\data'}$ for some $\data' \in \Data$. Since $\thrmap{b} \neq \tid$, we can conclude that this is not the case and therefore that if $\overlap{s_{\idx{\rid}}, \tid}$, then $\overlap{s'_{\idx{\rid}}, \tid}$. Hence, if $\overlap{s_{\idx{\rid}}, \tid}$ then $a$ is enabled in $s'_{\idx{\rid}}$ and thus $a$ is enabled in $s'$. We continue under the assumption that $\neg\overlap{s_{\idx{\rid}}, \tid}$, and thus that $\data= \trueval{s_{\idx{\rid}}}$. Assume towards a contradiction that $\neg\overlap{s'_{\idx{\rid}}, \tid}$ and $\data \neq \trueval{s'_{\idx{\rid}}}$. If $\trueval{s_{\idx{\rid}}} \neq \trueval{s'_{\idx{\rid}}}$, then $b$ must be $\finishwrite[\tidtwo,\rid]$ or $\orderwrite[\tidtwo,\rid]$ for some $\tidtwo \neq \tid$. Since there are no order write actions in the safe register model, it must be the case that $b$ is a finish write action. As $b$ is enabled in $s_{\idx{\rid}}$, this means that $\tidtwo \in \writers{s_{\idx{\rid}}} \neq \emptyset$. Since $\tid \in \readers{s_{\idx{\rid}}}$, it follows from \autoref{lem:overlap} that $\overlap{ s_{\idx{\rid}}, \tid}$. This contradicts our assumption that $\neg\overlap{s_{\idx{\rid}}, \tid}$, thus we can conclude that we must have $\overlap{s'_{\idx{\rid}}, \tid}$ or $\data = \trueval{s'_{\idx{\rid}}}$, and consequently $a$ is enabled in $s'_{\idx{\rid}}$. Hence, $a$ is enabled in $s'$.

                \item If $\rid$ is a regular register, then it follows from $a$ being enabled in $s_{\idx{\rid}}$ that $\tid \in \readers{s_{\idx{\rid}}}$ and  $\data \in \posval{s_{\idx{\rid}}, \tid}$. To show $a$ is enabled in $s'_{\idx{\rid}}$, it suffices to show that $\tid \in \readers{s'_{\idx{\rid}}}$ and $\data \in \posval{s'_{\idx{\rid}}. \tid}$. From \autoref{lem:access} and $\thrmap{b} \neq \tid$, it follows that if $\tid \in \readers{s_{\idx{\rid}}}$, then $\tid \in \readers{s'_{\idx{\rid}}}$. Additionally, as can be seen in \autoref{lem:access}, the set $\posvalsym$ for a specific thread id only grows over time until it is reset when that thread starts a read. Therefore, it is only possible that $\data \in \posval{s_{\idx{\rid}}, \tid}$ and $\data \notin \posval{s'_{\idx{\rid}}, \tid}$ are both true if $b = \startread[\tid,\rid]$. Since $\thrmap{b} \neq \tid$, we can conclude that $\data \in \posval{s'_{\idx{\rid}}, \tid}$. Thus, $a$ is enabled in $s'_{\idx{\rid}}$ and therefore $a$ is enabled in $s'$.

                \item If $r$ is an atomic register, then it follows from $a$ being enabled in $s_{\idx{\rid}}$ that $\tid \in \readers{s_{\idx{\rid}}}$, $\tid \notin \pending{s_{\idx{\rid}}}$ and $\data = \vals{s_{\idx{\rid}}, \tid}$. From \autoref{lem:access} it follows that when $\thrmap{b} \neq \tid$, then $\tid \in \readers{s_{\idx{\rid}}}$ implies $\tid \in \readers{s'_{\idx{\rid}}}$ and $\tid \notin \pending{s_{\idx{\rid}}}$ implies $\tid \notin \pending{s'_{\idx{\rid}}}$. Similarly, since $\thrmap{b} \neq \tid$, $\vals{s_{\idx{\rid}}, \tid} = \vals{s'_{\idx{\rid}}, \tid}$. Thus, $a$ is enabled in $s'_{\idx{\rid}}$ and therefore also in $s'$.
            \end{itemize}
            In all three cases, $a$ is enabled in $s'$. 

            \item If $a = \finishwrite[\tid,\rid]$, then what we need to show about $s'_{\idx{\rid}}$ to establish that $a$ is enabled differs depending on which type of register $\rid$ is. If $\rid$ is safe, then it suffices to show $\tid \in \writers{s'_{\idx{\rid}}}$. If $\rid$ is regular or atomic, we need to show $\tid \in \writers{s'_{\idx{\rid}}}$ and $\tid \notin \pending{s'_{\idx{\rid}}}$. From the definitions of these access function as given in \autoref{lem:access}, and $\thrmap{b} \neq \tid$, it follows that $\tid \in \writers{s_{\idx{\rid}}}$ implies $\tid \in \writers{s'_{\idx{\rid}}}$ and $\tid \notin \pending{s_{\idx{\rid}}}$ implies $\tid \notin \pending{s'_{\idx{\rid}}}$. Consequently, if $a$ is enabled in $s_{\idx{\rid}}$ then it is also enabled in $s'_{\idx{\rid}}$. Thus, we can conclude that $a$ is enabled in $s'_{\idx{\rid}}$ and therefore also in $s'$. 

            \item If $a = \orderwrite[\tid,\rid]$, then $\rid$ could be a regular or atomic register. From the two models, we know that since $a$ is enabled in $s_{\idx{\rid}}$, we have $\tid \in \writers{s_{\idx{\rid}}}$ and $\tid \in \pending{s_{\idx{\rid}}}$. For both models, to show $a$ is enabled in $s'_{\idx{\rid}}$, it suffices to show $\tid \in \writers{s'_{\idx{\rid}}}$ and $\tid \in \pending{s'_{\idx{\rid}}}$. From \autoref{lem:access} and $\thrmap{b} \neq \tid$, it follows that if $\tid \in \writers{s_{\idx{\rid}}}$ then $\tid \in \writers{s'_{\idx{\rid}}}$ and if $\tid \in \pending{s_{\idx{\rid}}}$ then  $\tid \in \pending{s'_{\idx{\rid}}}$. Thus, $a$ is enabled in $s'_{\idx{\rid}}$ and therefore $a$ is enabled in $s'$. 

            \item If $a = \orderread[\tid,\rid]$ then $\rid$ must be an atomic register.  From the atomic register model and $a$ being enabled in $s_{\idx{\rid}}$, we know that $\tid \in \readers{s_{\idx{\rid}}}$ and $\tid \in \pending{s_{\idx{\rid}}}$. From the definitions of the access functions as given in \autoref{lem:access} and $\thrmap{b}\neq \tid$, it follows that $\tid \in \readers{s_{\idx{\rid}}}$ implies $\tid \in \readers{s'_{\idx{\rid}}}$ and $\tid \in \pending{s_{\idx{\rid}}}$ implies $\tid \in \pending{s'_{\idx{\rid}}}$. Thus, when $a$ is enabled in $s_{\idx{\rid}}$ it is also enabled in $s'_{\idx{\rid}}$. We conclude that $a$ is enabled in $s'$.
        \end{itemize}
    \end{itemize}
    We have shown in all cases that $a$ is enabled in $s'$. Hence, our thread-register models are thread consistent with respect to the mapping $\thrsym$.
\end{proof}

\section{The tread interference relation is a concurrency relation} \label{app:concT-val}

\concTval*
\begin{proof}
To prove $\conc_T$ is a concurrency relation, we need to prove the two properties of concurrency relations.
    \begin{enumerate}
        \item It follows directly from the definition of $\conc_T$ that it is irreflexive: trivially $\thrmap{a} = \thrmap{a}$ for all $a \in \actionset$.
        \item Let $a$ be an arbitrary action in $\actionset$ and let $s$ be an arbitrary state in $\states$. Assume that $a$ is enabled in $s$ and that there exists a path $\pi$ from $s$ to some $s'$ such that $a \conc_T b$ for all $b$ occurring on $\pi$. We prove that $a$ is enabled in $s'$. 
        We do induction on the length of $\pi$.
        \begin{itemize}
            \item For the base case, $|\pi| = 0$, we observe that if $\pi$ has length $0$, then $s = s'$ and hence trivially $a$ is enabled in $s'$.
            \item Let $i \geq 0$ and assume that the claim holds for all paths from $s$ of length $i$; we prove that it also holds for all paths from $s$ of length $i + 1$. Let $\pi$ be a path from $s$ to some state $s' \in \states$ with $|\pi| = i + 1$ such that $a \conc_T b$ for all $b$ occurring on $\pi$. Then there exists a path $\pi'$ from $s$ to some state $s''$ of length $i$ such that $a \conc_T b$ for all $b$ occurring on $\pi'$, and $\pi = \pi'cs'$ for some $c \in \actionset$. Note that by assumption on $\pi$, $a \conc_T c$. By the induction hypothesis, $a$ is enabled in $s''$. Since $a \conc_T c$, it follows by definition of $\conc_T$ that  $\thrmap{a} \neq \thrmap{c}$ and thus, since $M$ is thread-consistent by \autoref{lem:thr-consist}, it follows that $a$ is enabled in $s'$.
        \end{itemize}
        We conclude that the property holds for all paths $\pi$.\qedhere
    \end{enumerate}
\end{proof}

\section{An explicit characterisation of complete paths}\label{app:complete paths}

\newcommand{\finish}{\textit{end}}

For a given mutual exclusion algorithm, let $M$ be the LTS of its thread-register model from \autoref{sec:thread-register}, using one of the register models employed in this paper---cf.\ \autoref{sec:registers}. Recalling that $M$ is a parallel composition of thread and register processes, each state in $M$ is a tuple $s=(s_1,\dots,s_k)$, where each of the indices $i \in \{1,\dots,k\}$ corresponds to a thread $t\in\TID$ or a register $r\in\RID$; as in \hyperlink{thrconsist}{the proof of} \autoref{lem:thr-consist} in
\autoref{app:thr-consist-proof} we denote the corresponding component $s_i$ as $s_{\idx{\tid}}$ or $s_{\idx{\rid}}$; it is a state in the LTS $T_t$ or $R_r$.

\begin{lemma}\rm\label{lem:1}
If $s$ is a reachable state of $M$ such that $s_{\idx{r}}$ enables an action $\orderwrite[\tid,\rid]$ or $\orderread[\tid,\rid]$ in $R_r$, then $s_{\idx{t}}$ enables an action $\finishwrite[\tid,\rid]$ or $\finishread[\tid,\rid]{d}$, respectively, in $T_t$.
Moreover, if  $s_{\idx{r}}$ enables $c = \finishwrite[\tid,\rid]$ or $c = \finishread[\tid,\rid]{d}$, then also  $s_{\idx{t}}$ enables $c$, and thus $s$ enables $c$ in $M$.
\end{lemma}
\begin{proof}
We will provide the proof for the cases that $s_{\idx{r}}$ enables an action $\orderwrite[\tid,\rid]$ or $\finishwrite[\tid,\rid]$. The case for read actions proceeds along the exact same lines.

Given a register $r\in \RID$ and a thread $t\in \TID$, in any execution path for any of the register models from Figures \ref{fig:procsafe}--\ref{fig:procatomic}, the actions $c$ with $\thrmap{c}=t$ and $\regmap{c}=r$ occur strictly in the order
$\startwrite[\tid,\rid]{d}$ -- $\orderwrite[\tid,\rid]$ -- $\finishwrite[\tid,\rid]$. This is a direct consequence of the conditions $\tid \in \writers{s}$ and $\tid \in \pending{s}$ in Figures \ref{fig:procsafe}--\ref{fig:procatomic}. Moreover, the actions $\startwrite[\tid,\rid]{d}$ for $d \in \Data$ and $\finishwrite[\tid,\rid]$ can occur only in synchronisation between the register $r$ and the thread $t$. Thus, if $s$ is a reachable state of $M$ such that $s_{\idx{r}}$ enables an action $\orderwrite[\tid,\rid]$ or $\finishwrite[\tid,\rid]$, then the last synchronisation between $r$ and $t$ must have been an action $\startwrite[\tid,\rid]{d}$.

In \autoref{sec:thread-register} we postulated that on all paths from the initial state of $T_t$, each transition labelled $\startwrite[\tid,\rid]{d}$ for some $\rid \in \RID$ and $d\in\Data$ must go to a state where only the action $\finishwrite[\tid,\rid]$ is enabled. Since the last synchronisation between $r$ and $t$ was on an action $\startwrite[\tid,\rid]{\data}$, this $\finishwrite[\tid,\rid]$ cannot have occurred yet. Hence $s_{\idx{t}}$ enables the action $\finishwrite[\tid,\rid]$.

In case $s_{\idx{r}}$ enables $\finishwrite[\tid,\rid]$, then (by the above) both $s_{\idx{r}}$ and $s_{\idx{t}}$ enable this action, and therefore it is also enabled by $s$.
\end{proof}

\begin{lemma}\rm\label{lem:2}
If $s$ is a state of $M$ such that $s_{\idx{t}}$ enables an action $\finishwrite[\tid,\rid]$ or $\finishread[\tid,\rid]{d}$, then either $\orderwrite[\tid,\rid]$ or $\orderread[\tid,\rid]$ or 
$\finishwrite[\tid,\rid]$ or $\finishread[\tid,\rid]{d'}$ for some $d'\in\Data$ is enabled by $s_{\idx{r}}$.
\end{lemma}
\begin{proof}
In \autoref{sec:thread-register} we postulated that in $T_t$ transitions labelled $\finishread[\tid,\rid]{\data}$ are only enabled right after performing a transition labelled $\startread[\tid,\rid]$. As in $M$ the  $\startread[\tid,\rid]$-transition must have been a synchronisation between thread $t$ and register $r$, and the register cannot execute an action $\finishread[\tid,\rid]{\data}$ without synchronising with $t$,
in each of our register models we have $\tid \in \readers{s}$. Consequently,  by Figures \ref{fig:procsafe}--\ref{fig:procatomic}, either $\orderread[\tid,\rid]$ or $\finishread[\tid,\rid]{d'}$ for some $d'\in\Data$ is enabled by $s_{\idx{r}}$.

The proof for write actions proceeds along the same lines.
\end{proof}

\begin{definition}\rm
Given any path $\pi$ starting in the initial state of $M$, and given any thread $t\in \TID$, if $\pi$ has a suffix $\pi'$ on which no actions $b$ with $\thrmap{b} = t$ occur, starting from a state $s$ of $\pi$, then $s'_{\idx{t}} = s_{\idx{t}}$ for all states $s'$ in $\pi'$ and we define $\finish_t (\pi) = s_{\idx{t}}$.

Call action $a\in Act$ \emph{thread-enabled} by $\pi$ if, for $t= \thrmap{a}$,
$\pi$ contains only finitely many actions $b$ with $\thrmap{b} = t$ and 
$a$ is enabled in the state $\finish_t(\pi)$ of the LTS $T_t$.
\end{definition}

\noindent
Recall from \autoref{def:concT} that $a \nconc_T b$ iff $\thrmap{a}=\thrmap{b}$. Hence a path $\pi$ is \hyperlink{just}{$\block$-$\conc_T$-unjust} if, and only if, it has a suffix $\pi'$ such that an action $a \in\nonblock$ is enabled in the initial state of $\pi'$, but $\pi$ does not contain any action $b$ with $\thrmap{a}=\thrmap{b}$. Using this, we obtain the following characterisation of the $\block$-$\conc_T$-just paths of $M$.

\begin{proposition}\rm\label{pr:T-just paths}
  A path $\pi$ starting in the initial state of $M$, is \hyperlink{just}{$\block$-$\conc_T$-just} if, and only if, $\pi$ thread-enables no actions $a \in \nonblock$.
\end{proposition}
\begin{proof}
Suppose $\pi$ thread-enables an action $a \in\nonblock$. We have to show that $\pi$ is \hyperlink{just}{not just}. Let $t= \thrmap{a}$ and $r= \regmap{a}$.
Let $\pi'$ be a suffix of $\pi$ on which no actions $b$ with $\thrmap{b}=t$ occur. Let $s$ be the initial state of $\pi'$. So $s_{\idx{t}} = \finish_t(\pi)$ and $s_{\idx{t}}$ enables $a$.

In case $r \mathbin=\undefsymb$, as $s_{\idx{t}}$ enables $a$ and $a$ does not require synchronisation with any register, also $s$ enables $a$. As $\pi'$ does not contain actions $b$ with $a \nconc_T b$, the path $\pi$ is \hyperlink{just}{not just}.\footnote{When rereading this proof as part of the proof of \autoref{pr:A-just paths} we also use that $\forall r\in\RID.~ a \notin \textit{start}(r)$.}

So assume that $r \in \RID$.
Based on Figures \ref{fig:procsafe}--\ref{fig:procatomic},
any state $s'_r$ of $r$ enables either
(i) both $\startread[\tid,\rid]$ and $\startwrite[\tid,\rid]{d}$ for all $d \in \Data$,
(ii) $\orderread[\tid,\rid]$ or $\orderwrite[\tid,\rid]$, or
(iii) $\finishwrite[\tid,\rid]$ or $\finishread[\tid,\rid]{d}$ for some $d \in \Data$. 

If state $s_{\idx{r}}$ in $R_r$ enables $c=\orderread[\tid,\rid]$ or $c=\orderwrite[\tid,\rid]$, also $s$ enables $c$, since $c$ does not require synchronisation with thread $t$. As $\pi'$ contains no actions $b$ with $c \nconc_T b$, using that $\thrmap{c}=t$, the path $\pi$ is \hyperlink{just}{not just}.\footnote{When rereading this proof as part of the proof of \autoref{pr:A-just paths} we also use that $c \notin \textit{start}(r)$.}

In case state $s_{\idx{r}}$ in $R_r$ enables an action $c=\finishread[\tid,\rid]{d}$ or $c=\finishwrite[\tid,\rid]$, by \autoref{lem:1} above also state $s$ enables $c$. As $\pi'$ contains no actions $b$ with $c \nconc_T b$, the path $\pi$ is \hyperlink{just}{not just}.$^{\rm \thefootnote}$

We may now restrict attention to case (i) above, that $s_{\idx{r}}$ enables both $\startread[\tid,\rid]$ and $\startwrite[\tid,\rid]{d}$ for all $d \in \Data$. 
In this setting we proceed with a case distinction on the action $a$, which must be of the form $\startread[\tid,\rid]$, $\startwrite[\tid,\rid]{d}$, $\finishread[\tid,\rid]{d}$ or $\finishwrite[\tid,\rid]$, since it is an action in LTS $T_{\tid}$ and $r \neq \undefsymb$. By \autoref{lem:2}, the case that $a = \finishwrite[\tid,\rid]$ or $a= \finishread[\tid,\rid]{d}$ for some $d \in \Data$ is already subsumed by the cases considered above. Thus we may restrict attention to the case that $a = \startread[\tid,\rid]$ or $a = \startwrite[\tid,\rid]{d}$ for some $d \in \Data$. In this case also  $s$ enables $a$. As $\pi'$ does not contain actions $b$ with $a \nconc_T b$, the path $\pi$ is \hyperlink{just}{not just}.\footnote{When reusing this in the proof of \autoref{pr:A-just paths}, we also use that $\pi'$ contains no actions $b \in \textit{start}(r)$.}

For the other direction, suppose that $\pi$ is $\block$-$\conc_T$-just. We have to establish that $\pi$ thread-enables an action $a \in \nonblock$. Let $\pi'$ be a suffix of $\pi$ such that an action $a \in \nonblock$ is enabled in the initial state $s$ of $\pi'$, but $\pi'$ contains no action $b$ such that $a \nconc_T b$. Let $t=\thrmap{a}$. Then $\pi'$ contains no action $b$ with $\thrmap{b}=t$. In case $a=\orderwrite[\tid,\rid]$ or $a=\orderread[\tid,\rid]$, then by \autoref{lem:1} above $\finish_t(\pi)=s_{\idx{t}}$ enables an action $c = \finishwrite[\tid,\rid]$ or $c = \finishread[\tid,\rid]{d}$, and thus $\pi$ thread-enables the action $c \in \nonblock$. Otherwise, $\finish_t(\pi)=s_{\idx{t}}$ enables $a$ and thus $\pi$ thread-enables $a \in \nonblock$.\pagebreak[3]
\end{proof}
Interestingly, this explicit characterisation of the \hyperlink{just}{$\block$-$\conc_T$-just} paths in our thread-register models is independent on whether we employ safe, regular or atomic registers.

Next, we provide similar characterisations of the $\block$-$\conc_C$-\hyperlink{just}{just} paths, for $C \in \{A, I, S\}$.
For a given register $r\in \RID$ we use the abbreviation $\textit{start}(r)$ for the set of all actions $\startread[\tid,\rid]$ and $\startwrite[\tid,\rid]{d}$ for some $t\in\TID$ and $d\in\Data$.
Recall from Definitions~\ref{def:concT}--\ref{def:concA} that $a \nconc_{\!\!A} b$ iff either $\thrmap{a}=\thrmap{b}$ or $a,b\in\textit{start}(r)$ for some $r\in\RID$. Hence a path $\pi$ is $\block$-$\conc_{\!\!A}$-unjust if, and only if, it has a suffix $\pi'$ such that an action $a \in\nonblock$ is enabled in the initial state of $\pi'$, but $\pi$ does not contain any action $b$ such that $\thrmap{a}=\thrmap{b}$ or $a,b\in\textit{start}(r)$ for some $r\in\RID$.

\begin{proposition}\rm\label{pr:A-just paths}
  A path $\pi$ starting in the initial state of $M$, is $\block$-$\conc_{\!\!A}$-just if, and only if,
  \begin{enumerate}[a)]
  \item $\pi$ thread-enables no actions $a \in \nonblock$ other than 
actions from $\textit{start}(r)$ for some $r\in\RID$, and 
  \item if, for some $r\in\RID$, an action $a \in \textit{start}(r)$ is thread-enabled by $\pi$, then $\pi$ contains infinitely many occurrences of actions $b\in\textit{start}(r)$.
  \end{enumerate}
\end{proposition}
\begin{proof}
Suppose $\pi$ thread-enables an action $a \in \nonblock$, say with $t= \thrmap{a}$ and $r= \regmap{a}$, such that if $a \in \textit{start}(r)$ then $\pi$ contains only finitely many actions $b\in\textit{start}(r)$. We have to show that $\pi$ is \hyperlink{just}{not $\block$-$\conc_{\!\!A}$-just}. Let $\pi'$ be a suffix of $\pi$ in which no actions $b$ with $\thrmap{b}=t$ occur; in case $a \in \textit{start}(r)$ we moreover choose $\pi'$ such that it contains no actions $b\in\textit{start}(r)$. From here on, the proof proceeds exactly as the one of \autoref{pr:T-just paths}, but reading $\nconc_{\!\!A}$ for $\nconc_T$.

For the other direction, suppose that $\pi$ is \hyperlink{just}{not $\block$-$\conc_{\!\!A}$-}just. We have to establish that\linebreak[3] 
(a) $\pi$ thread-enables an action $a \in \nonblock$, and (b) in case $a \in \textit{start}(r)$ then $\pi$ contains only finitely many actions $b\in\textit{start}(r)$. Let $\pi'$ be a suffix of $\pi$ such that an action $a \in \nonblock$ is enabled in the initial state $s$ of $\pi'$, but $\pi'$ contains no action $b$ such that $a \nconc_{\!\!A} b$, that is, $\pi'$ contains no action $b$ with $\thrmap{a}=\thrmap{b}$ or $a,b\in\textit{start}(r)$ for some $r\in\RID$.
The proof of (a) above proceeds as in the proof of \autoref{pr:T-just paths}, and (b) is now a trivial corollary.
\end{proof}

\noindent
Recall from Definitions~\ref{def:concT}--\ref{def:concI} that $a \nconc_{\!I} b$ iff either $\thrmap{a}=\thrmap{b}$ or $a,b\in\textit{start}(r)$ for some $r\in\RID$, with $\issw{a} \vee \issw{b}$. Hence a path $\pi$ is $\block$-$\conc_{\!I}$-unjust if, and only if, it has a suffix $\pi'$ such that an action $a \in\nonblock$ is enabled in the initial state of $\pi'$, but $\pi$ does not contain any action $b$ such that $\thrmap{a}=\thrmap{b}$ or $a,b\in\textit{start}(r)$ for some $r\in\RID$, with $\issw{a} \vee \issw{b}$.

Similarly, $a \nconc_{S} b$ iff either $\thrmap{a}=\thrmap{b}$ or $a,b\in\textit{start}(r)$ for some $r\in\RID$, with $\issw{b}$. Hence a path $\pi$ is $\block$-$\conc_{S}$-unjust if, and only if, it has a suffix $\pi'$ such that an action $a \in\nonblock$ is enabled in the initial state of $\pi'$, but $\pi$ does not contain any action $b$ such that $\thrmap{a}=\thrmap{b}$ or $a,b\in\textit{start}(r)$ for some $r\in\RID$, with $\issw{b}$.

\begin{proposition}\rm\label{pr:I-just paths}
  A path $\pi$ starting in the initial state of $M$, is $\block$-$\conc_{\!I}$-just if, and only if,
  \begin{enumerate}[a)]
  \item $\pi$ thread-enables no actions $a \in \nonblock$ other than actions from $\textit{start}(r)$ for some $r\in\RID$,
  \item if an action $\startwrite[\tid,\rid]{d}$ is thread-enabled by $\pi$, then $\pi$ contains infinitely many occurrences of actions $b\in\textit{start}(r)$, and
  \item if an action $\startread[\tid,\rid]$ is thread-enabled by $\pi$, then $\pi$ contains infinitely many occurrences of actions $b$ of the form $\startwrite[\tid',\rid]{d}$ for some $\tid'\in\TID$ and $d'\in\Data$.
  \end{enumerate}
\end{proposition}
The proof of this proposition, and the next one, proceeds just like the one of \autoref{pr:A-just paths}.

\pagebreak[2]
\begin{proposition}\rm\label{pr:S-just paths}
  A path $\pi$ starting in the initial state of $M$, is $\block$-$\conc_{S}$-just if, and only if,
  \begin{enumerate}[a)]
  \item $\pi$ thread-enables no actions $a \in \nonblock$ other than actions from $\textit{start}(r)$ for some $r\in\RID$,
  \item if an action $a \in \textit{start}(r)$ is thread-enabled by $\pi$, then $\pi$ contains infinitely many occurrences of actions $b$ of the form $\startwrite[\tid',\rid]{d}$ for some $\tid'\in\TID$ and $d'\in\Data$.
  \end{enumerate}
\end{proposition}

\noindent
An interesting consequence of these propositions is that whether or not a given path $\pi$ is
$\block$-$\conc_{C}$-just, for some $C \in \{T,S,I,A\}$, is completely determined by the set of
actions that are thread-enabled by $\pi$, and by the function $\textit{occ}_\pi:Act \rightarrow \mathbbm{N} \cup \{\infty\}$ that tells for each action how often it occurs in $\pi$.

\section{Register models that avoid overlap}\label{app:blocking registers}

To accurately capture the memory model of blocking reads and writes, we need not only the concurrency relation $\conc_{\!A}$, but also a register model that does not allow any read or write to start when another read or write to the same register is in progress. This can be achieved by changing, in \autoref{fig:procatomic}, the condition $\tid\notin(\readers{s}\cup\writers{s})$, occurring in the first two lines, into
$(\readers{s}\cup\writers{s}) = \emptyset$.

Similarly, in the blocking model with concurrent reads, where reads and writes have to await the completion of in-progress writes, but only writes have to await the completion of in-progress reads, the first two lines of \autoref{fig:procatomic} become
  \[\begin{array}{lr@{~\then~}l}
            & (\tid\mathbin{\notin}\readers{s} \wedge \writers{s}=\emptyset) &
                 \startread[\tid,\rid]\co\AReg(\rid,\usr{s,\tid}) \\
          + & ((\readers{s}\cup\writers{s}) = \emptyset) &
                 \startwrite[\tid,\rid]{\data}\co\AReg(\rid,\usw{s,\tid,\data})
  \end{array}\]

To capture the model with blocking writes and non-blocking reads, we make the same alterations as above, except that writes do not need to wait for in-progress reads. Hence, the first two lines become 
\[\begin{array}{lr@{~\then~}l}
            & (\tid\mathbin{\notin}\readers{s} \wedge \writers{s}=\emptyset) &
                 \startread[\tid,\rid]\co\AReg(\rid,\usr{s,\tid}) \\
          + & (\writers{s} = \emptyset) &
                 \startwrite[\tid,\rid]{\data}\co\AReg(\rid,\usw{s,\tid,\data})
  \end{array}\]
Note that this only models the blocking behaviour described in \autoref{sec:introduction}; we do not model that a write causes a read to be aborted and subsequently resumed.

For our verifications we have not made these modifications, but simply reused the register model of \autoref{fig:procatomic}. Below, we explain why this leads to the same verification results.

All our correctness properties for mutual exclusion protocols are linear-time properties: they hold for a process, modelled as an LTS, iff they hold for all complete paths starting in the initial state of that LTS\@. Moreover, whether they hold for a particular path depends solely on the labels of the transitions in that path, and in fact only on those transitions labelled $\crit[\tid]$ or $\noncrit[\tid]$, for some $\tid\in\TID$. This is witnessed by the modal $\mu$-calculus formulae for these properties given in \autoref{app:mucalc}.
For the (sole) purpose of deciding whether mutual exclusion, deadlock freedom or starvation freedom hold for a given LTS, once its is determined  which paths are complete, all other actions may be considered internal, or hidden.

\begin{definition}\rm
  Given a path $\pi=s_0 a_1 s_1 a_2 s_2 \dots$, let $\ell(\pi)=a_1 a_2 \dots$ be the sequence of actions occurring on that path, and let $\ell^-(\pi)$ be the result of omitting from $\ell(\pi)$ all actions other than $\crit[\tid]$ or $\noncrit[\tid]$, for some $\tid\in\TID$.

  Given a completeness criterion $C$, the set of \emph{weak $C$-complete traces} $\WCT_C(P)$ of an LTS $P$ consists of those finite and infinite strings $\ell^-(\pi)$ for $\pi$ a $C$-complete path starting in the initial state of $P$.
  Two LTSs $P$ and $Q$ are \emph{weak completed trace equivalent} w.r.t.\ $C$, notation $P =^C_\WCT Q$ if, and only if, $\WCT_C(P) = \WCT_C(Q)$.
\end{definition}

\noindent
It now follows that, given a completeness criterion $C$, two weak completed trace equivalent LTSs satisfy the same correctness properties for mutual exclusion protocols.

For a given mutual exclusion algorithm, let $M$ be the LTS of its thread-register model, using the atomic register model from \autoref{fig:procatomic}. Moreover, let $M_A$ be the variant of $M$ that employs the above register model for blocking reads and writes; let $M_I$ be the variant for the blocking model with concurrent reads, and let $M_S$ be the one for the model of blocking writes and non-blocking reads.

Now, using $A$, $I$ and $S$ as abbreviations for the completeness criteria $\justact{\conc_{\!A}}{\block}$, $\justact{\conc_I}{\block}$ and $\justact{\conc_S}{\block}$, we will show that $M_A =^A_\WCT M$, $M_I =^I_\WCT M$ and $M_S =^S_\WCT M$. This implies that it makes no difference whether, for the verifications based on the (partly) blocking memory models, we use the register models fine-tuned for that memory model as described above, or the model for atomic registers from \autoref{fig:procatomic}.

In the subsequent argument, let $C \in \{A, I, S\}$.

First, we must establish that $\conc_C$ is a valid concurrency relation for $M_C$. That the relations are irreflexive solely depends on the actions of an LTS, and hence still clearly holds within the context of $M_A, M_I$ and $M_S$.
It is less immediately obvious that the \hyperlink{second}{second property of concurrency relations} still holds, particularly because the fine-tuned models are not thread-consistent: it is for example possible for $\startread[\tid,\rid]$ to be disabled by the occurrence of $\startwrite[\tidtwo,\rid]{\data}$ in all three models, even when $\tid \neq \tidtwo$.
However, in all three cases, if an action $a$ is disabled by the occurrence of an action $b$, then $a \nconc_C b$:
\begin{itemize}
    \item In $M_A$, $\startread[\tid,\rid]$ and $\startwrite[\tid, \rid]{\data}$ can be disabled by the occurrence of $\startread[\tidtwo, \rid]$ or $\startwrite[\tidtwo, \rid]{\data'}$. Accordingly, in $\conc_{\!A}$ start read actions and start write actions both interfere with both start read and start write actions on the same register.
    \item In $M_I$, $\startread[\tid,\rid]$ can be disabled by the occurrence of $\startwrite[\tidtwo,\rid]{\data}$. Accordingly, $\conc_I$ allows start write actions to interfere with start read actions on the same register. Additionally, $\startwrite[\tid,\rid]{\data}$ can be disabled by the occurrence of $\startread[\tidtwo,\rid]$ or $\startwrite[\tidtwo,\rid]{\data'}$. Indeed, $\conc_I$ allows start reads and start writes to interfere with start writes on the same register.
    \item In $M_S$, $\startread[\tid,\rid]$ and $\startwrite[\tid, \rid]{\data}$ can both be disabled by an occurrence of $\startwrite[\tidtwo, \rid]{\data'}$. Accordingly, $\conc_S$ allows start writes to interfere with both start reads and start writes on the same register.
\end{itemize}
Thus, in all three register models an action can only be disabled by the occurrence of an action that interferes with it according to the appropriate concurrency relation. Consequently, if there is a transition $s \xrightarrow{b} s'$ such that an action $a$ is enabled in $s$ and $a \conc_C b$, then $a$ is enabled in $s'$. By induction, similarly to the proof of \autoref{lem:concT-val}, \hyperlink{second}{the second property of concurrency relations} is satisfied, and we can apply these concurrency relations to their respective models.

Next, we will show that the most relevant results of \autoref{app:complete paths} apply also to the register models introduced above.
First observe that Lemmas~\ref{lem:1} and~\ref{lem:2} apply equally well, with the same proofs, when reading $M_A$, $M_I$ or $M_S$ for $M$, since the only difference between these models is when start actions are enabled, and there lemmas refer only to order and finish actions.

\begin{lemma}\rm\label{lem:register enablings}
Let $T_r$ be the LTS of any register $r$, either defined as in one of the
Figures \ref{fig:procsafe}--\ref{fig:procatomic}, or using one of three register models defined at the beginning of this appendix.
Then any state $s_r$ of $r$ enables either
(i) $\startread[\tid,\rid]$ and $\startwrite[\tid,\rid]{d}$ for all $t\in\TID$ and all $d \in \Data$, or
(ii) $\orderread[\tid,\rid]$ or $\orderwrite[\tid,\rid]$ for some $t\in\TID$, or
(iii) $\finishwrite[\tid,\rid]$ or $\finishread[\tid,\rid]{d}$ for some $t\in\TID$ and $d \in \Data$.
\end{lemma}
\begin{proof}
In case  $\readers{s_r} \cup \writers{s_r} = \emptyset$, option (i) applies.
This can be seen from Figures \ref{fig:procsafe}--\ref{fig:procatomic}, and from the description of three register models defined at the beginning of this appendix.  Likewise, in case $\readers{s_r} \cup \writers{s_r} \neq \emptyset$, one of (ii) or (iii) must apply.  (In the case of regular registers, we use the easily checked invariant that if $t \in \readers{s}$ then $\posval{s,t}\neq\emptyset$.)
\end{proof}

\begin{lemma}\rm\label{lem:register enablings S}
Let $T_r$ be the LTS of any register $r$, as defined at the beginning of this appendix for the register model with blocking writes and non-blocking reads, i.e.\ $M_S$.
Then any state $s_r$ of $r$ enables either
(i) $\startwrite[\tid,\rid]{d}$ for all $t\in\TID$ and all $d \in \Data$, or
(ii) $\orderwrite[\tid,\rid]$ for some $t\in\TID$, or
(iii) $\finishwrite[\tid,\rid]$ for some $t\in\TID$.
If case (i) applies then, for each $t'\in\TID$, $s_r$ enables either
(a) $\startread[\tid',\rid]$, or
(b) $\orderread[\tid',\rid]$, or
(c) $\finishread[\tid',\rid]{d}$ for all $d \in \Data$.
\end{lemma}
\begin{proof}
In case $\writers{s_r} = \emptyset$, option (i) applies.
This can be seen from the description of this register model.  Likewise, in case $\writers{s_r} \neq \emptyset$, one of (ii) or (iii) must apply.

Assuming $\writers{s_r} \mathbin= \emptyset$, if $t\mathbin{\notin}\readers{s_r}$, option (a) applies and otherwise (b) or (c) apply.
\end{proof}

\begin{proposition}\rm\label{pr:just paths}
The characterisations of the $\block$-$\conc_C$-\hyperlink{just}{just} paths of $M$ from Propositions~\ref{pr:A-just paths}--\ref{pr:S-just paths} apply equally well to $M_C$.
\end{proposition}

\begin{proof}
Below we give a proof for the case $C=A$, most of which applies to all three cases $C \in \{A,I,S\}$. The parts that are specific to the case $C=A$ are coloured {\color{highlightColour}{\colourname}}, and afterwards we will give the case $C=S$ with the differences between it and $C=A$ similarly highlighted. The case $C=I$ is treated at the end.

We must prove that a path $\pi$ starting in the initial state of $M_A$, is $\block$-$\conc_A$-just if, and only if, a) $\pi$ thread-enables no actions $a \in \nonblock$ other than actions from $\mathit{start}(r)$ for some $r \in \RID$, and b) if, for some $r \in \RID$, an action $a \in \mathit{start}(r)$ is thread-enabled by $\pi$, then $\pi$ contains infinitely many occurrences of actions $b \in \mathit{start}(r)$.

Suppose $\pi$ thread-enables an action $a \in \nonblock$, say with $t= \thrmap{a}$ and $r= \regmap{a}$, such that {\color{highlightColour}if $a \in \textit{start}(r)$ then $\pi$ contains only finitely many actions $b\in\textit{start}(r)$. In the latter case $\pi$ contains only finitely many actions $b$ with $\regmap{b}=r$},
because actions $c$ with $\thrmap{c}=t'$ and $\regmap{c}=r$ must occur strictly in the order $\startwrite[\tid',\rid]{d}$ -- $\orderwrite[\tid',\rid]$ -- $\finishwrite[\tid',\rid]$.

We have to show that $\pi$ is \hyperlink{just}{not $\block$-$\conc_{C}$-just}. Let $\pi'$ be a suffix of $\pi$ in which no actions $b$ with $\thrmap{b}=t$ occur; {\color{highlightColour}in case $a \in \textit{start}(r)$ we moreover choose $\pi'$ such that it contains no actions $b$ with $\regmap{b}=r$.} Let $s$ be the initial state of $\pi'$. So $s_{\idx{t}} = \finish_t(\pi)$ enables $a$.

In case $r \mathbin=\undefsymb$, as $s_{\idx{t}}$ enables $a$ and $a$ does not require synchronisation with any register, also $s$ enables $a$. As $\pi'$ does not contain actions $b$ with $a \nconc_{C} b$, the path $\pi$ is \hyperlink{just}{not just}.

So assume that $r \in \RID$. We proceed with a case distinction on the action $a$, which must be of the form $\startread[\tid,\rid]$, $\startwrite[\tid,\rid]{d}$, $\finishread[\tid,\rid]{d}$ or $\finishwrite[\tid,\rid]$, since it is an action in LTS $T_{\tid}$ and $r \neq \undefsymb$.

First assume that $a = \finishwrite[\tid,\rid]$ or $a= \finishread[\tid,\rid]{d}$ for some $d \in \Data$. By \autoref{lem:2}, either $\orderwrite[\tid,\rid]$ or $\orderread[\tid,\rid]$ or $\finishwrite[\tid,\rid]$ or $\finishread[\tid,\rid]{d'}$ for some $d'\in\Data$ is enabled by state $s_{\idx{r}}$ in $R_r$.
In case $s_{\idx{r}}$ enables $c=\orderread[\tid,\rid]$ or $c=\orderwrite[\tid,\rid]$, also $s$ enables $c$, since $c$ does not require synchronisation with thread $t$. As $\pi'$ contains no actions $b$ with $c \nconc_{C} b$, using that $\thrmap{c}=t$, the path $\pi$ is \hyperlink{just}{not just}.
In case $s_{\idx{r}}$ enables an action $c=\finishread[\tid,\rid]{d}$ or $c=\finishwrite[\tid,\rid]$, by \autoref{lem:1} also state $s$ enables $c$. As $\pi'$ contains no actions $b$ with $c \nconc_C b$, the path $\pi$ is \hyperlink{just}{not just}.

Thus we may restrict attention to the case that $a = \startread[\tid,\rid]$ or $a = \startwrite[\tid,\rid]{d}$ for some $d \in \Data$, that is, $a\in \textit{start}(r)$. So henceforth we may use that $\pi'$ {\color{highlightColour}contains no actions $b$ with $\regmap{b}=r$.}

{\color{highlightColour}We proceed by making a case distinction on the three possibilities described by \autoref{lem:register enablings} for actions enabled by $s_{\idx{r}}$.
In case (i) of \autoref{lem:register enablings}, $s_{\idx{r}}$ enables $\startread[\tid',\rid]$ and $\startwrite[\tid',\rid]{d'}$ for all $t'\in\TID$ and all $d' \in \Data$, and thus in particular the action $a$.} Consequently, also  $s$ enables $a$. As $\pi'$ does not contain actions $b$ with $\thrmap{b}=t$ {\color{highlightColour}or actions $b \in \textit{start}(r)$}, it contains no action $b$ with $a \nconc_{C} b$. Hence, the path $\pi$ is \hyperlink{just}{not just}.

In case (ii) $s_{\idx{r}}$, and hence also $s$, enables {\color{highlightColour}$\orderread[\tid',\rid]$ or $\orderwrite[\tid',\rid]$} for some $t'\mathbin\in\TID$. Thus, by \autoref{lem:1}, $s_{\idx{t'}}$ enables an action {\color{highlightColour}$\finishwrite[\tid',\rid]$ or $\finishread[\tid',\rid]{d}$}. By the assumptions of \autoref{sec:thread-register}, this implies that $s_{\idx{t'}}$ only enables {\color{highlightColour}actions $b$ with $\regmap{b}=r$}. Since $\pi'$ does not contain such $b$, it cannot contain any action $c$ with $\thrmap{c}=t'$ either.
Consequently, $\pi$ is not just.

In case (iii) $s_{\idx{r}}$ enables an action {\color{highlightColour}$c=\finishwrite[\tid',\rid]$ or $c=\finishread[\tid',\rid]{d}$ for some $t'\in\TID$ and $d \in \Data$}.
Thus, by \autoref{lem:1}, $c$ is enabled by $s_{\idx{t'}}$ as well as by $s$.
The rest of the argument proceeds just as for case (ii) above.

For the other direction, note that $M_C$ has the same states as $M$ and can be obtained from $M$ by leaving out some transitions. This implies that any path $\pi$ in $M_C$ is also a path in $M$, and moreover, if $\pi$ is 
\hyperlink{just}{$\block$-$\conc_{C}$-unjust} in $M_C$ then it certainly is 
\hyperlink{just}{$\block$-$\conc_{C}$-unjust} in $M$. Using this, the reverse direction of this proof is direct consequence of Propositions~\ref{pr:A-just paths}--\ref{pr:S-just paths}.
\vspace{2ex}

\noindent
In the case $C = S$, we must prove that a path $\pi$ starting in the initial state of $M_S$, is $\block$-$\conc_S$-just if, and only if, a) $\pi$ thread-enables no actions $a \in \nonblock$ other than actions from $\mathit{start}(r)$ for some $r \in \RID$, and b) if an action $a \in \mathit{start}(r)$ is thread-enabled by $\pi$, then $\pi$ contains infinitely many occurrences of actions $b$ of the form $\startwrite[t',r]{d}$ for some $t' \in \TID$ and $d' \in \Data$. 
The proof is identical to the $C=A$ case, save for the {\colourname} phrases.

\noindent
Suppose $\pi$ thread-enables an action $a \in \nonblock$, say with $t= \thrmap{a}$ and $r= \regmap{a}$, such that {\color{highlightColour}if $a \in \textit{start}(r)$ then $\pi$ contains only finitely many actions $b$ of the form $\startwrite[\tid',\rid]{d}$ for some $\tid'\in\TID$ and $d'\in\Data$. In the latter case $\pi$ also contains only finitely many actions $b$ of the form $\orderwrite[\tid',\rid]$ or $\finishwrite[\tid',\rid]$},
because actions $c$ with $\thrmap{c}=t'$ and $\regmap{c}=r$ must occur strictly in the order $\startwrite[\tid',\rid]{d}$ -- $\orderwrite[\tid',\rid]$ -- $\finishwrite[\tid',\rid]$.

We have to show that $\pi$ is \hyperlink{just}{not $\block$-$\conc_{C}$-just}. Let $\pi'$ be a suffix of $\pi$ in which no actions $b$ with $\thrmap{b}=t$ occur; {\color{highlightColour}in case $a \in \textit{start}(r)$ we moreover choose $\pi'$ such that it contains no actions $b$ of the form $\startwrite[\tid',\rid]{d}$, $\orderwrite[\tid',\rid]$ or $\finishwrite[\tid',\rid]$ for some $\tid'\in\TID$ and $d'\in\Data$.} Let $s$ be the initial state of $\pi'$. So $s_{\idx{t}} = \finish_t(\pi)$ enables $a$.

In case $r \mathbin=\undefsymb$, as $s_{\idx{t}}$ enables $a$ and $a$ does not require synchronisation with any register, also $s$ enables $a$. As $\pi'$ does not contain actions $b$ with $a \nconc_{C} b$, the path $\pi$ is \hyperlink{just}{not just}.

So assume that $r \in \RID$. We proceed with a case distinction on the action $a$, which must be of the form $\startread[\tid,\rid]$, $\startwrite[\tid,\rid]{d}$, $\finishread[\tid,\rid]{d}$ or $\finishwrite[\tid,\rid]$, since it is an action in LTS $T_{\tid}$ and $r \neq \undefsymb$.

First assume that $a = \finishwrite[\tid,\rid]$ or $a= \finishread[\tid,\rid]{d}$ for some $d \in \Data$. By \autoref{lem:2}, either $\orderwrite[\tid,\rid]$ or $\orderread[\tid,\rid]$ or $\finishwrite[\tid,\rid]$ or $\finishread[\tid,\rid]{d'}$ for some $d'\in\Data$ is enabled by state $s_{\idx{r}}$ in $R_r$.
In case $s_{\idx{r}}$ enables $c=\orderread[\tid,\rid]$ or $c=\orderwrite[\tid,\rid]$, also $s$ enables $c$, since $c$ does not require synchronisation with thread $t$. As $\pi'$ contains no actions $b$ with $c \nconc_{C} b$, using that $\thrmap{c}=t$, the path $\pi$ is \hyperlink{just}{not just}.
In case $s_{\idx{r}}$ enables an action $c=\finishread[\tid,\rid]{d}$ or $c=\finishwrite[\tid,\rid]$, by \autoref{lem:1} also state $s$ enables $c$. As $\pi'$ contains no actions $b$ with $c \nconc_C b$, the path $\pi$ is \hyperlink{just}{not just}.

Thus we may restrict attention to the case that $a = \startread[\tid,\rid]$ or $a = \startwrite[\tid,\rid]{d}$ for some $d \in \Data$, that is, $a\in \textit{start}(r)$. So henceforth we may use that $\pi'$ {\color{highlightColour}contains no actions $b$ of the form $\startwrite[\tid',\rid]{d}$, $\orderwrite[\tid',\rid]$ or $\finishwrite[\tid',\rid]$ for some $\tid'\in\TID$ and $d'\in\Data$.}

{\color{highlightColour}We proceed by making a case distinction on the three possibilities described by \autoref{lem:register enablings S} for actions enabled by $s_{\idx{r}}$.
In case (i) of \autoref{lem:register enablings S}, we make a further case distinction, depending on whether $a = \startread[\tid,\rid]$ or $a = \startwrite[\tid,\rid]{d}$. In the latter case,
$s_{\idx{r}}$ enables $\startwrite[\tid',\rid]{d'}$ for all $t'\in\TID$ and all $d' \in \Data$, and thus in particular the action $a$. In the former case, options (b) and (c) for $t'=t$ of \autoref{lem:register enablings S} are ruled out, because in those cases \autoref{lem:1} would imply that an action $\finishread[\tid,\rid]{d}$ would be enabled by $s_\idx{t}$, but by assumption (see \autoref{sec:thread-register}) this cannot happen in a state enabling $\startread[\tid,\rid]$. Thus (a) applies, and also $s_{\idx{r}}$ enables action $a$.} Consequently, also  $s$ enables $a$. As $\pi'$ does not contain actions $b$ with $\thrmap{b}=t$ {\color{highlightColour}or actions $b\in\textit{start}(r)$ with $\issw{b}$}, it contains no action $b$ with $a \nconc_{C} b$. Hence, the path $\pi$ is \hyperlink{just}{not just}.

In case (ii) $s_{\idx{r}}$, and hence also $s$, enables {\color{highlightColour}$\orderwrite[\tid',\rid]$} for some $t'\mathbin\in\TID$. Thus, by \autoref{lem:1}, $s_{\idx{t'}}$ enables an action {\color{highlightColour}$\finishwrite[\tid',\rid]$}. By the assumptions of \autoref{sec:thread-register}, this implies that $s_{\idx{t'}}$ only enables {\color{highlightColour}the action $b=\finishwrite[\tid',\rid]$}. Since $\pi'$ does not contain such $b$, it cannot contain any action $c$ with $\thrmap{c}=t'$ either.
Consequently, $\pi$ is not just.

In case (iii) $s_{\idx{r}}$ enables an action {\color{highlightColour}$c=\finishwrite[\tid',\rid]$ for some $t'\in\TID$}.
Thus, by \autoref{lem:1}, $c$ is enabled by $s_{\idx{t'}}$ as well as by $s$.
The rest of the argument proceeds just as for case (ii) above.

For the other direction, note that $M_C$ has the same states as $M$ and can be obtained from $M$ by leaving out some transitions. This implies that any path $\pi$ in $M_C$ is also a path in $M$, and moreover, if $\pi$ is 
\hyperlink{just}{$\block$-$\conc_{C}$-unjust} in $M_C$ then it certainly is 
\hyperlink{just}{$\block$-$\conc_{C}$-unjust} in $M$. Using this, the reverse direction of this proof is direct consequence of Propositions~\ref{pr:A-just paths}--\ref{pr:S-just paths}.
\vspace{2ex}

\noindent
Finally, for the case $C=I$ we must prove that a path $\pi$ starting in the initial state of $M_I$ is $\block$-$\conc_I$-just if, and only if, a) $\pi$ thread-enables no actions $a \in \nonblock$ other than actions from $\mathit{start}(r)$ for some $r \in \RID$, and b) if an action $\startwrite[\tid,\rid]{\data}$ is thread-enabled by $\pi$, then $\pi$ contains infinitely many occurrences of actions $b \in \mathit{start}(r)$, and c) if an action $\startread[\tid,\rid]$ is thread-enabled by $\pi$, then $\pi$ contains infinitely many occurrences of actions $b$ of the form $\startwrite[\tidtwo,\rid]{\data}$.
Let $\pi$ be a path in $M_I$ that thread-enables an action $ a\in \nonblock$ such that one of these three conditions is violated. In the case that the violated condition is a) or b), the proof proceeds just as in the case $C=A$; in the case it is condition c) the proof proceeds just as in the case $C=S$. The reverse direction goes exactly as in the case $C=A$.
\end{proof}

\begin{lemma}\rm\label{lem:MCtoM}
    Let $C \in \{A, I, S\}$. If $\pi$ is a $C$-complete path in $M_C$ starting in the initial state of $M_C$, then it is also a $C$-complete path in $M$ starting in the initial state of $M$.
\end{lemma}
\begin{proof}
    Let $\pi$ be a $C$-complete path in $M_C$ starting in its initial state. 
    Since for all three variants, $M_C$ has stricter conditions for actions being enabled than $M$ does, $\pi$ is guaranteed to exist in $M$. It remains to show that $\pi$ is $C$-complete in $M$. This, however, is an immediate consequence of \autoref{pr:just paths}, from which it follows that a path starting in the initial state of $M_C$ is $\block$-$\conc_C$-just exactly when that same path starting in the initial state of $M$ is.
\end{proof}

We also wish to prove the other direction. However, it is not necessarily the case that any $C$-complete path in $M$ is also a $C$-complete path in $M_C$, if only because not every path in $M$ exists in $M_C$. Thus, we need a way to convert a path in $M$ into a path that is guaranteed to exist in $M_C$.
We define a rewrite relation on paths, where $\pi \leadsto \pi'$ denotes that $\pi$ is rewritten into $\pi'$. Namely $\pi \leadsto \pi'$ if, and only if, there exists paths $\sigma$ and $\rho$, states $s$ and $s'$, thread $t\in\TID$, and an action $b$ with $\thrmap{b}\neq \tid$, such that
\begin{itemize}
\item either $\pi = \sigma ~\startread[\tid,\rid]~s~b~\rho$ and $\pi' = \sigma~b~s'~\startread[\tid,\rid]~\rho$,
\item or $\pi = \sigma ~\startwrite[\tid,\rid]{d}~s~b~\rho$ and $\pi' = \sigma~b~s'~\startwrite[\tid,\rid]{d}~\rho$,
\item or $\pi = \sigma ~b~s~\finishread[\tid,\rid]{d}~\rho$ and $\pi' = \sigma~\finishread[\tid,\rid]{d}~s'~b~\rho$,
\item or $\pi = \sigma ~b~s~\finishwrite[\tid,\rid]~\rho$ and $\pi' = \sigma~\finishwrite[\tid,\rid]~s'~b~\rho$.
\end{itemize}
Thus, this rewrite relation moves any start read or start write action forwards by swapping it with an action $b$ of another thread, and in the same manner moves any finish read or finish write action backwards. To establish that whenever a path $\pi$ of any of the above forms exists in $M$, the corresponding $\pi'$ also exists in $M$, it suffices to prove that for any transition $s \xrightarrow{b} s'$ that is part of a path $\pi$ in $M$: (i) if a start action $a$ is enabled in $s$ and $\thrmap{a} \neq \thrmap{b}$, then $a$ is enabled in $s'$, and (ii) if a finish action $a$ is enabled in $s'$ and $\thrmap{a} \neq \thrmap{b}$, then $a$ is enabled in $s$.
The former follows directly from \autoref{lem:thr-consist}; (ii) straightforwardly holds in the case of a finish write action, and holds in the case of a finish read action because the return value is determined by the matching order read action, which belongs to the same thread.
Since paths may be infinite, this rewrite system need not terminate, but in the limit it converts any path $\pi$ into a normal form $\norm{\pi}$ in which each action $\orderread[\tid,\rid]$ is immediately preceded by the corresponding $\startread[\tid,\rid]$ and immediately followed by the corresponding $\finishread[\tid,\rid]{d}$, and likewise for write actions. To convergence to this limit, one should give priority to rewrite steps that apply closer to the beginning of the path.
This normal form also exists in $M_A$, $M_I$ and $M_S$.

\begin{lemma}\rm\label{lem:MtoMC}
    Let $C \in \{A, I, S\}$. If $\pi$ is a $C$-complete path in $M$ starting in the initial state of $M$, then $\norm{\pi}$ is a $C$-complete path in $M_C$ starting in the initial state of $M_C$.
\end{lemma}
\begin{proof}
    Let $\pi$ be a $C$-complete path in $M$ starting in its initial state. Trivially, $\norm{\pi}$ exists in $M_C$. It remains to show that $\norm{\pi}$ is $C$-complete in $M_C$.
    Note that $M_C$ is a subgraph of $M$. Hence, if $\norm{\pi}$ is $C$-complete in $M$, then it is also $C$-complete in $M_C$. 
    It therefore suffices to prove that $\norm{\pi}$ is $C$-complete in $M$.
This is an immediate consequence of Propositions~\ref{pr:A-just paths}--\ref{pr:S-just paths}, given that $\pi$ and $\norm{\pi}$ have the same sets of thread-enabled actions and $\textit{occ}_\pi = \textit{occ}_{\norm{\pi}}$ (i.e., each action occurs as often in $\norm{\pi}$ as it occurs in $\pi$).
\end{proof}

\begin{theorem}\rm
$M_A =^A_\WCT M$, $M_I =^I_\WCT M$ and $M_S =^S_\WCT M$.
\end{theorem}
\begin{proof}
We prove that $\WCT_C(M) = \WCT_C(M_C)$ for all $C \in \{A, I, S\}$. We prove this by mutual set inclusion.

First, let $\pi$ be a $C$-complete path from the initial state of $M_C$, such that $\ell^-(\pi) \mathbin\in \WCT_C(M_C)$. By \autoref{lem:MCtoM}, $\pi$ is also a $C$-complete path from the initial state of $M$. Thus, $\ell^-(\pi) \in \WCT_C(M)$.

Second, let $\pi$ be a $C$-complete path from the initial state of $M$, such that $\ell^-(\pi) \mathbin\in \WCT_C(M)$. Then by \autoref{lem:MtoMC}, $\norm{\pi}$ is a $C$-complete path in $M_C$. Note that $\ell^-(\norm{\pi}) = \ell^-(\pi)$, hence $\ell^-(\pi) \in \WCT_C(M_C)$.

We have proven $M_C =^C_\WCT M$ for all $C \in \{A, I, S\}$.
\end{proof}

\section{mCRL2 models}\label{app:mCRL2}
In this section, we comment on some of the details of our thread and register models in mCRL2. For information on the mCRL2 language, we refer to \cite{mCRL2language}.
We made several alterations while translating the process-algebraic definitions of MWMR safe, regular and atomic registers (see \autoref{sec:registers}) to mCRL2 processes.
Some of these were necessary to obtain valid mCRL2 models, others were employed to reduce the state space of the models.
The models are available as supplementary material.
Note that the exact models differ slightly from those used in \cite{spronck2023process}; we altered them to align more closely with the process-algebraic definitions.

The most important differences between the definitions from \autoref{sec:registers} and the mCRL2 models are as follows:
\begin{itemize}
    \item Since the three models do not require exactly the same information from the status object, and tracking unused information unnecessarily increases the state space of the model, we separate the status object into three variants, one for each register type, that only track the required information for that type.
    \item We reset values to a pre-defined default whenever we know that the value is no longer relevant. For example, the safe register model only needs to know the value that a thread intends to write if this write operation does not encounter an overlapping write; if there are multiple concurrent writes, it will not matter what their intended values were. Hence, in the safe register model we only track a single value instead of a mapping from threads to values for $\valssym$, and reset this value to its default whenever we observe zero or more than one active writer.
    \item For regular registers, instead of computing on the spot what the values of all active writes are whenever a read starts, we use a multiset to keep track of those values as writes start and finish; this is more straightforward and, since it adds no new information, does not expand the state-space.
    \item To further reduce the state space, we add the value of a new writer only to the $\posvalsym$ of active readers, rather than all threads. Since $\posvalsym$ is reset for a thread whenever it starts a read, this does not affect the behaviour of the model.
    \item We have moved summations inwards whenever possible, so that $\Data$ is only summed over when the resulting value is actually relevant.
\end{itemize}

In addition to the registers, the threads must also be modelled in mCRL2. 
We already mentioned the relevant choices made here in \autoref{sec:verification}.
It remains to discuss how the parallel composition as defined in \autoref{sec:preliminaries} and employed in \autoref{sec:thread-register} is modelled in mCRL2. In order to get the right communication between threads and registers, we create ``sending'' and ``receiving'' versions of all register interface actions, and define a communication function so that the two versions together form the correct action.
For instance, to start a read of register $\rid$, thread $\tid$ does the action $\mathit{start\_read\_s}(\tid, \rid)$. The register simultaneously does the action $\mathit{start\_read\_r}(\tid,\rid)$, and this communication appears in the model as $\mathit{start\_read}(\tid,\rid)$. 
The register and thread local actions do not need such a modification, since they are only performed by a single component.

It is worth noting that the mCRL2 version of the modal $\mu$-calculus does not support quantification over actions directly, but does allow quantification over the data parameters of actions. In order to express the formulae we need (see \autoref{app:mucalc}), we therefore add the action $\mathit{label}$ to every action in the model, creating multi-actions. We give the $\mathit{label}$ action a parameter over a set of labels $\mathbb{L}$, where each label corresponds to one of the action names used in \autoref{sec:registers}. We then hide the original actions, so that only $\mathit{label}$ is visible in the model. This way, we can refer to the labels in our formulae, which circumvents the issue of not being able to quantify over actions. This approach is based on \cite{bouwman2020off}.

\section[Modal mu-calculus formulae]{Modal $\mu$-calculus formulae}\label{app:mucalc}
In \cite{spronck2024progress}, we presented template formulae that can be instantiated to capture liveness properties that fit Dwyer, Avrunin and Corbett's Property Specification Patterns (PSP) \cite{dwyer1999patterns}, incorporating a variety of completeness criteria, including justness.

In this appendix, we recap the three correctness properties we verified, and give the modal $\mu$-calculus formulae for these properties. 
To be able to use the template formulae with justness for the liveness properties, we must show how they fit into PSP.
We do not thoroughly explain the modal $\mu$-calculus here; instead, we refer to \cite{spronck2024progress} for more information on how these formulae should be understood.
We do not use the mCRL2 modal $\mu$-calculus syntax here, since it results in longer and less readable formulae. However, the files are all given in the supplementary material.

\subparagraph{Mutual exclusion}
The mutual exclusion property says that at any given time, at most one thread will be in its critical section.
In our models, a thread $\tid$ accessing its critical section is represented by the action $\crit[\tid]$.
We reformulate the mutual exclusion property for our models as: for all threads $\tid$ and $\tidtwo$ such that $\tid \neq \tidtwo$, at any time it is impossible that both $\crit[\tid]$ and $\crit[\tidtwo]$ are enabled.
This is captured by the following modal $\mu$-calculus formula:
\begin{equation}
    \bigwedge_{\tid, \tidtwo \in \TID}( (\tid \neq \tidtwo) \imps \boxm{\clos{\allact}}\neg(\diam{\crit[\tid]}\tp \land \diam{\crit[\tidtwo]}\tp))
\end{equation}

\subparagraph{Deadlock freedom}
The deadlock freedom property says that whenever at least one thread is running its entry protocol, eventually some thread will enter its critical section.
Since we want to apply the results from \cite{spronck2024progress}, we need to first show how this can be represented in PSP, using the actions in our model.
This property fits into the global response pattern of \cite{dwyer1999patterns}: whenever the trigger occurs, in this case one thread being in its entry protocol, we want the response to occur eventually, in this case some thread executing its critical section.
We cannot directly apply the results from \cite{spronck2024progress}, however, since there we require the trigger to be a set of actions and a thread ``being in the entry protocol'' cannot simply be captured by the occurrence of an action. 
Instead, a thread $\tid$ is in its entry protocol between the execution of $\noncrit[\tid]$ and the subsequent execution of $\crit[\tid]$. 
Fortunately, this requires only a minor deviation from the template given in \cite{spronck2024progress}, namely by allowing the trigger to be a regular expression over sets of actions instead.

Assume that for $C \in \{T, S, I, A\}$ and all actions $a$ in our model, we have a function $\elimf{C}{a}$ that maps an action to all actions that interfere with it according to concurrency relation $C$.
We have encoded these functions in our mCRL2 models.
For clarity, let $\crit[\mathit{all}] = \bigcup_{\tid \in \TID}\{\crit[\tid]\}$.

In the modal $\mu$-calculus formula, we capture that there \emph{does not} exist a path from the initial state of the model that violates deadlock freedom and is just. This is equivalent to checking that deadlock freedom is satisfied on all just paths.
The formula \[\nu X.(\bigwedge_{a \in \nonblock}( \diam{a}\tp \imps \diam{\clos{\comp{\crit[\mathit{all}]}}}(\diam{\elimf{C}{a} \setminus \crit[\mathit{all}]}X)))\] holds in a state $s$ iff there is a path starting in $s$ that (i) is just, in the sense that any state enabling an action $a \notin \block$ is followed by an action from $\elimf{C}{a}$, and (ii) does not contain any action from  $\crit[\mathit{all}]$.
In other words, this formula says that from a state $s$, there is a just path where the response never occurs. We then merely need to prepend the formula for an occurrence of the trigger, namely an occurrence of $\noncrit[\tid]$, for some $\tid\in\TID$, that is not (yet) followed by $\crit[\tid]$, and negate the resulting formula to capture that the initial state of the model does not admit just paths that violate deadlock freedom.

The final formula for deadlock freedom under $\justact{\conc_{C}}{\block}$ is:
\begin{equation}
    \neg\diam{\clos{\allact}}\bigvee_{\tid \in \TID}(\diam{\noncrit[\tid] \co \clos{\comp{\crit[\tid]}}} \nu X.(\bigwedge_{a \in \nonblock}( \diam{a}\tp \imps \diam{\clos{\comp{\crit[\mathit{all}]}}}(\diam{\elimf{C}{a} \setminus \crit[\mathit{all}]}X))))
\end{equation}

\subparagraph{Starvation freedom}
Starvation freedom says that whenever a thread leaves its non-critical section, it will eventually enter its critical section.
Unlike deadlock freedom, starvation freedom does fit into the global response pattern using only sets of actions: the trigger is the occurrence of $\noncrit[\tid]$ for some $\tid \in \TID$, and the response is the occurrence of $\crit[\tid]$.
Using the same functions $\elimf{C}{a}$ for $C \in \{T, S, I, A\}$ and actions $a$ as the deadlock freedom formula, we get the following modal $\mu$-calculus formula for starvation freedom under $\justact{\conc{_C}}{\block}$:
\begin{equation}
    \bigwedge_{\tid \in \TID}(\neg \diam{\clos{\allact} \co \noncrit[\tid]}\nu X. (\bigwedge_{a \in \nonblock}(\diam{a}\tp \imps \diam{\clos{\comp{\crit[\tid]}}}(\diam{\elimf{C}{a} \setminus \crit[\tid]}X))))
\end{equation}

\section{Algorithms}\label{app:algorithms}

In this appendix, we go over all the algorithms mentioned in \autoref{tab:results} that have not already been discussed in \autoref{sec:verification}.
We give the pseudocode of algorithms, and where relevant also discuss the results of the verification.

\subsection{Anderson's algorithm}
Anderson's algorithm is presented in \cite{anderson1993fine}. It only works for 2 threads, and has different code for both threads.
For clarity, we break from our established shorthand of $j = 1-i$ for just this algorithm, and present the algorithms for $i = 0$ (\cref{alg:anderson-0}) and $j = 1$ (\cref{alg:anderson-1}) separately.
There are six Booleans total: $\varidx{p}{i}, \varidx{p}{j}, \varidx{q}{i}, \varidx{q}{j}, \varidx{t}{i}$ and $\varidx{t}{j}$. All are initialised to $\true$.

\begin{table}[ht!]
\vspace{-1em}
\noindent\begin{minipage}[t]{0.47\textwidth}
\begin{algorithm}[H]
\caption{Anderson's algorithm for $i = 0$}\label{alg:anderson-0}
\begin{algorithmic}[1]
    \State{$\varidx{p}{i} \writeop \false$}
    \State{$\varidx{q}{i} \writeop \false$}
    \State{$\varname{x} \writeop \varidx{t}{j}$}
    \State{$\varidx{t}{i} \writeop \varname{x}$}
    \If{$\varname{x} = \true$}
        \State{$\varidx{p}{i} \writeop \true$}
        \State{\textbf{await} $\varidx{p}{j} = \true$}
    \Else
        \State{$\varidx{q}{i} \writeop \true$}
        \State{\textbf{await} $\varidx{q}{j} = \true$}
    \EndIf
    \State{\textbf{critical section}}
    \State{$\varidx{p}{i} \writeop \true$}
    \State{$\varidx{q}{i} \writeop \true$}
\end{algorithmic}
\end{algorithm}
\end{minipage}
\hfill
\begin{minipage}[t]{0.47\textwidth}
\begin{algorithm}[H]
\caption{Anderson's algorithm for $j = 1$}\label{alg:anderson-1}
\begin{algorithmic}[1]
    \State{$\varidx{p}{j} \writeop \false$}
    \State{$\varidx{q}{j} \writeop \false$}
    \State{$\varname{x} \writeop \neg\varidx{t}{i}$}
    \State{$\varidx{t}{j} \writeop \varname{x}$}
    \If{$\varname{x} = \true$}
        \State{$\varidx{q}{j} \writeop \true$}
        \State{\textbf{await} $\varidx{p}{i} = \true$}
    \Else
        \State{$\varidx{p}{j} \writeop \true$}
        \State{\textbf{await} $\varidx{q}{i} = \true$}
    \EndIf
    \State{\textbf{critical section}}
    \State{$\varidx{p}{j} \writeop \true$}
    \State{$\varidx{q}{j} \writeop \true$}
\end{algorithmic}
\end{algorithm}
\end{minipage}
\end{table}

In accordance with Anderson's claim, the algorithm satisfies all three properties with non-atomic as well as atomic registers.

\subsection{Attiya-Welch's algorithm}
What we call the Attiya-Welch algorithm is presented in both \cite{AttiyaWelch04} (original presentation, \cref{alg:attiya-welch}) and \cite{Shao11} (variant presentation, \cref{alg:attiya-welch-var}) as a variant of Peterson's algorithm.
The two presentations of the algorithm have different behaviour, so we present both.
This algorithm is only defined for $N = 2$.
Each thread $i$ has a Boolean $\varidx{flag}{i}$, initialised to $\false$.
There is also a global variable $\varname{turn}$ over $\TID$, initialised to $0$.

\begin{table}[ht!]
\vspace{-1em}
\noindent\begin{minipage}[t]{0.47\textwidth}
\begin{algorithm}[H]
\caption{Attiya-Welch algorithm, orig.}\label{alg:attiya-welch}
\begin{algorithmic}[1]
        \State{$\varidx{flag}{i} \writeop \false$}\label{attiya-welch-entry2}
        \State{\textbf{await} $\varidx{flag}{j} = \false \lor \varname{turn} = j$}
        \State{$\varidx{flag}{i} \writeop \true$}
        \If{$\varname{turn} = i$}\label{attiya-welch-if}
            \If{$\varidx{flag}{j} = \true$}
                \State{\textbf{goto} line \ref{attiya-welch-entry2}}\label{attiya-welch-goto}
            \EndIf
        \Else
            \State{\textbf{await} $\varidx{flag}{j} = \false$}\label{attiya-welch-wait}
        \EndIf
        \State{\textbf{critical section}}
        \State{$\varname{turn} \writeop i$}
        \State{$\varidx{flag}{i} \writeop \false$}
\end{algorithmic}
\end{algorithm}
\end{minipage}
\hfill
\begin{minipage}[t]{0.47\textwidth}
\begin{algorithm}[H]
\caption{Attiya-Welch algorithm, var.}\label{alg:attiya-welch-var}
\begin{algorithmic}[1]
    \Repeat
    \State{$\varidx{flag}{i} \writeop \false$}
    \State{\textbf{await} $\varidx{flag}{j} = \false \lor \varname{turn} = j$}
    \State{$\varidx{flag}{i} \writeop \true$}
    \Until{$\varname{turn} = j \lor \varidx{flag}{j} = \false$}\label{attiya-welch-var-until}
    \If{$\varname{turn} = j$}\label{attiya-welch-var-check}
    \State{\textbf{await} $\varidx{flag}{j} = \false$}
    \EndIf
    \State{\textbf{critical section}}\label{attiya-welch-var-crit}
    \State{$\varname{turn} \writeop i$}
    \State{$\varidx{flag}{i} \writeop \false$}
\end{algorithmic}
\end{algorithm}
\end{minipage}
\end{table}

In the original presentation of the Attiya-Welch algorithm, no claims are made about its correctness with non-atomic registers.
It is therefore perhaps surprising to note that it does satisfy many properties with non-atomic registers.
In \cite{spronck2023process}, we showed that the original presentation of the Attiya-Welch algorithm satisfies reachability of the critical section, a property weaker than deadlock freedom, with both safe and regular registers.
Here, we show that while it satisfies starvation freedom with regular registers, it only satisfies deadlock freedom with safe registers.
{\Atrace} violating starvation freedom with safe registers is the following:
\begin{itemize}
    \item Thread 0 runs through lines 1 through 9 without competition; since $\varidx{flag}{1} = \false$ it can reach line 9 without problem. On line 10, it starts writing $0$ to $\varname{turn}$, which is already $0$.
    \item Thread 1 executes line 1, setting $\varidx{flag}{1}$ to $\false$, which it already was. On line 2, it reads $\varidx{flag}{0} = \true$ and $\varname{turn} = 1$. Note that the value $1$ has never been written to $\varname{turn}$, but due to the read overlapping with thread 0's write, the value can still be read. Thread 1 can therefore not proceed through line 2.
    \item Thread 0 finishes the exit protocol, setting $\varname{turn}$ to $0$ and $\varidx{flag}{0}$ to $\false$.
\end{itemize}
At this point, we  have $\varidx{flag}{0} = \varidx{flag}{1} = \false$ and $\varname{turn}=0$, the same values the variables had at the start. Hence, thread 0 can reach line 10 again, without interference by thread 1. By having thread 1 always read $\varidx{flag}{0}$ and $\varname{turn}$ at exactly the wrong time, it will remain forever in line 2, without ever reaching the critical section.

Note that this {\trace} relies on reading a value that has never been written. We can adjust the algorithm slightly, so that $\varname{turn}$ is read before it is written, and only updated if it does not already have the intended value. Starvation freedom is then satisfied with safe registers.

We find that the Attiya-Welch algorithm does not satisfy starvation freedom with atomic registers under $\justact{\conc_S}{\block}$. 
This violation is rather trivial: as long as $\varidx{flag}{1} = \false$, thread 0 can infinitely often execute the algorithm to get to its critical section. During this execution, it will write to $\varidx{flag}{0}$ repeatedly, preventing thread 1 from reading $\varidx{flag}{0}$ on line 2. Hence, thread 1 can be prevented from ever reaching the critical section if a write can block a read. 
Note that on line 2 of the algorithm, a thread already wants access to its critical section but has in no way communicated this to the other thread.
This is not something we can fix by merely slightly altering a condition or reading a variable before writing, it would require more significant alterations to the algorithm so that a thread has to communicate its intention before attempting to read a register another thread can infinitely often write to.
We do not explore such alterations here.

Instead, we turn our attention to the variant presentation from \cite{Shao11}.
In this version, the goto statements have been eliminated.

In \cite{Shao11}, it is claimed that the algorithm satisfies all three properties for all four interpretations of MWMR regular registers proposed in that paper.
A comparison between their definitions and our regular register model is given in \cite{spronck2023process}; here it suffices to observe that their weakest definition is weaker than our regular register model.
Hence, we would expect all three properties to be satisfied by our regular model.
As can be observed in \autoref{tab:results}, this is not the case.
While the two presentations of the algorithm are seemingly equivalent, the altered pseudocode suggests the $\varname{turn}$ variable needs to be read twice in a row, where it is read only once on line 4 of the original presentation.
This allows new-old inversion to occur with the non-atomic register models, and causes the variant presentation to no longer satisfy deadlock freedom with non-atomic registers.
{\Atrace} demonstrating the violation is given in \cite{spronck2023process}.
If we alter the model so that the value of $\varname{turn}$ is only read once for the two conditions, we see the expected behaviour, given the results in \cite{Shao11}.
If we then also make the adjustment that $\varname{turn}$ is only written to when its value would actually change, we see the same behaviour as our altered version of the original presentation.
The issue of starvation freedom not being satisfied for atomic registers with $\justact{\conc_S}{\block}$ is also present in this variant.

\subsection{Burns-Lynch's algorithm and Lamport's 1-bit algorithm}
The Burns-Lynch algorithm (\cite{burns1993bounds}, \cref{alg:burns-lynch}) and Lamport's 1-bit algorithm (\cite{Lamport86Mutex2}, \cref{alg:lamport1bit}) are two very similar algorithms.
In \cite{Lamport86Mutex2}, Lamport makes the explicit claim that this algorithm satisfies mutual exclusion and deadlock freedom with safe registers.
Both algorithms are designed for an arbitrary $N$, and both use only a single shared Boolean per thread, initialised to $\false$.
For the sake of easy comparison, we name this Boolean $\varidx{flag}{i}$ for each thread $i$ for both algorithms, although Lamport uses the name $\varidx{x}{i}$ and Burns and Lynch use $\varidx{Flag}{i}$.

\begin{table}[ht!]
\vspace{-1em}
\noindent\begin{minipage}[t]{0.47\textwidth}
\begin{algorithm}[H]
\caption{The Burns-Lynch algorithm}\label{alg:burns-lynch}
\begin{algorithmic}[1]
    \Repeat
        \State{$\varidx{flag}{i} \writeop \false$}\label{bl-3}
        \State{\textbf{await} $\forall_{j < i}: \varidx{flag}{j} = \false$}
        \State{$\varidx{flag}{i} \writeop \true$}
    \Until{$\forall_{j < i}:$ $\varidx{flag}{j} = \false$}
    \State{\textbf{await} $\forall_{j > i}: \varidx{flag}{j} = \false$}
    \State{\textbf{critical section}}
    \State{$\varidx{flag}{i} \writeop \false$}
\end{algorithmic}
\end{algorithm}
\end{minipage}
\hfill
\begin{minipage}[t]{0.47\textwidth}
\begin{algorithm}[H]
\caption{Lamport's 1-bit algorithm}\label{alg:lamport1bit}
\begin{algorithmic}[1]
    \State{$\varidx{flag}{i} \writeop \true$}\label{lamport1-l}
    \For{$j$ \textbf{from} $0$ \textbf{to} $i - 1$}
        \If{$\varidx{flag}{j} = \true$}
            \State{$\varidx{flag}{i} \writeop \false$}
            \State{\textbf{await} $\varidx{flag}{j} = \false$}
            \State{\textbf{goto} line \ref{lamport1-l}}
        \EndIf
    \EndFor
    \For{$j$ \textbf{from} $i + 1$ \textbf{to} $N - 1$}
        \State{\textbf{await} $\varidx{flag}{j} = \false$}
    \EndFor
    \State{\textbf{critical section}}
    \State{$\varidx{flag}{i} \writeop \false$}
\end{algorithmic}
\end{algorithm}
\end{minipage}
\end{table}

Neither algorithm was designed to satisfy starvation freedom, only deadlock freedom.
In \cite{buhr2015high}, it is claimed that both work with non-atomic registers, albeit without satisfying starvation freedom.
With the Attiya-Welch algorithm, we saw that minor differences in presentation can impact the correctness of an algorithm.
This does not appear to be the case with these two algorithms; they show the same, and the expected, behaviour.

\subsection{From deadlock freedom to starvation freedom}
We here give the pseudocode for the algorithm for turning a deadlock-free algorithm into a starvation-free one, as discussed in \autoref{sec:dftosf}.
See \autoref{alg:GG} for the algorithm as presented in \cite{GG}. It works for arbitrary $N$.

In addition to the registers used by the deadlock-free algorithm, this algorithm uses a Boolean $\varname{flag}$ for every thread, initialised to $\false$ and a shared register over $\TID$ called $\varname{turn}$, which here is initialised to $0$.
Naturally, these registers must be distinct from those used in the deadlock-free algorithm.

\begin{algorithm}[ht!]
\caption{Algorithm for making a deadlock-free solution starvation-free}\label{alg:GG}
\begin{algorithmic}[1]
    \State{$\varidx{flag}{i} \writeop \true$}
    \Repeat
        \State{$\varname{tmp} \writeop \varname{turn}$}
    \Until{$\varname{tmp} = i \lor \varidx{flag}{tmp} = \false$}
    \State{\textbf{entry protocol of deadlock-free algorithm}}
    \State{\textbf{critical section}}
    \State{$\varidx{flag}{i} \writeop \false$}
    \State{$\varname{tmp} \writeop \varname{turn}$}
    \If{$\varidx{flag}{tmp} = \false$}
        \State{$\varname{turn} \writeop (\varname{tmp} + 1) \mod N$}
    \EndIf
    \State{\textbf{exit protocol of deadlock-free algorithm}}
\end{algorithmic}
\end{algorithm}

\subsection{Dijkstra's algorithm}
Dijkstra's algorithm \cite{dijkstra65} is given as \autoref{alg:dijkstra}.
It works for arbitrary $N$.
Every thread $i$ has two Booleans: $\varidx{b}{i}$ and $\varidx{c}{i}$, both initialised to $\true$.
There is also a global register $\varname{k}$ over $\TID$, initialised at $0$.

\begin{algorithm}[ht!]
\caption{Dijkstra's algorithm}\label{alg:dijkstra}
\begin{algorithmic}[1]
    \State{$\varidx{b}{i} \writeop \false$}\label{dijkstra-Li0}
    \If{$\varname{k} \neq i$}\label{dijkstra-Li1}
        \State{$\varidx{c}{i} \writeop \true$}\label{dijkstra-Li2}
        \If{$\varidx{b}{k} = \true$}\label{dijkstra-Li3}
            \State{$\varname{k} \writeop i$}
        \EndIf
        \State{\textbf{goto} line \ref{dijkstra-Li1}}
    \Else
        \State{$\varidx{c}{i} \writeop \false$}\label{dijkstra-Li4}
        \For{$j$ \textbf{from} $0$ \textbf{to} $N - 1$}
            \If{$j \neq i \land \varidx{c}{j} = \false$}
                \State{\textbf{goto} line \ref{dijkstra-Li1}}
            \EndIf
        \EndFor
    \EndIf
    \State{\textbf{critical section}}
    \State{$\varidx{c}{i} \writeop \true$}
    \State{$\varidx{b}{i} \writeop \true$}
\end{algorithmic}
\end{algorithm}

We find that this algorithm never satisfies starvation freedom, which is unsurprising since Dijkstra did not require that property in his original presentation of the mutual exclusion problem.

\subsection{Kessels's algorithm}
Kessels's algorithm (\autoref{alg:kessels}) is a variant on Peterson's, presented in \cite{kessels1982arbitration}.
Its basic form is designed for $N = 2$, and each thread $i$ has two variables: a Boolean $\varidx{q}{i}$, initialised at $\false$, and $\varidx{r}{i}$ over $\TID$, initialised at $0$.

\begin{algorithm}[ht!]
\caption{Kessels's algorithm}\label{alg:kessels}
\begin{algorithmic}[1]
    \State{$\varidx{q}{j} \writeop \true$}
    \State{$\varidx{r}{j} \writeop (\varidx{r}{i} + j) \mod 2$}
    \State{\textbf{await} $\varidx{q}{i} = \false \lor (\varidx{r}{j} \neq ((\varidx{r}{i} + j) \mod 2))$}
    \State{\textbf{critical section}}
    \State{$\varidx{q}{j} \writeop \false$}    
\end{algorithmic}
\end{algorithm}

Since it is a variant on Peterson's, merely with only single-writer registers, it is unsurprising that the two algorithms give the same results to our verification.
The violating {\trace} for regular registers, like many of the {\trace}s discussed in this section, relies on new-old inversion.
\begin{itemize}
    \item Thread 0 sets $\varidx{q}{1}$ to $\true$, and then reads $\varidx{r}{0} = 0$. It therefore starts writing $1$ to $\varidx{r}{1}$, which was $0$.
    \item Thread 1 sets $\varidx{q}{0}$ to $\true$, and then reads $\varidx{r}{1}$ with overlap, reading the new value $1$. It therefore writes $1$ to $\varidx{r}{0}$. It then reads $\varidx{r}{0} = 1$ and $\varidx{r}{1} = 0$, the latter being the old value. Since $1 \neq (0 + 0) \mod 2$, it can reach its critical section.
    \item Thread 0 finishes its write to $\varidx{r}{1}$, and on line 3 reads $\varidx{r}{1} = 1$ and $\varidx{r}{0} = 1$; since $1 \neq (1 + 1)\mod 2$, it can reach its critical section.
\end{itemize}

This violation does not contradict any claims made by Kessels in \cite{kessels1982arbitration}, who said nothing about the correctness of this algorithm with non-atomic registers.
Rather, we find the violation interesting because it clearly illustrates that using only single-writer Booleans does not automatically mean that an algorithm is robust to non-atomic registers.

\subsection{Knuth's algorithm}
Knuth's algorithm (\autoref{alg:knuth}) comes from \cite{knuth1966additional}.
It works for arbitrary $N$.
Every thread $i$ has a register $\varidx{control}{i}$, which is over $\{0,1,2\}$, initialised to $0$.
Additionally, there is a global register $\varname{k}$ over $\TID$, also initialised to $0$.\footnote{In Knuth's original presentation \cite{knuth1966additional}, $\varname{k}$ was initialised to $-1$, which there was called $0$, as the threads ranged from $1$ to $N$. This difference is immaterial; when thread 1 goes once through  the algorithm without contention, $k$ is set to $0$, so this could just as well be the initial value.}

\begin{algorithm}[ht!]
  \caption{Knuth's algorithm}\label{alg:knuth}
  \begin{algorithmic}[1]
    \State{$\varidx{control}{i} \writeop 1$}\label{knuth-L0}
    \For{$j$ \textbf{from} $\varname{k}$ \textbf{downto} $0$}\label{knuth-L1}
        \If{$j = i$}
            \State{\textbf{goto} line \ref{knuth-L2}}
        \EndIf
        \If{$\varidx{control}{j} \neq 0$}
            \State{\textbf{goto} line \ref{knuth-L1}}
        \EndIf
    \EndFor
    \For{$j$ \textbf{from} $N{-}1$ \textbf{downto} $0$}
        \If{$j = i$}
            \State{\textbf{goto} line \ref{knuth-L2}}
        \EndIf
        \If{$\varidx{control}{j} \neq 0$}
            \State{\textbf{goto} line \ref{knuth-L1}}
        \EndIf
    \EndFor    
    \State{$\varidx{control}{i} \writeop 2$}\label{knuth-L2}
    \For{$j$ \textbf{from} $N{-}1$ \textbf{downto} $0$}
        \If{$j \neq i \land \varidx{control}{j} = 2$}
            \State{\textbf{goto} line \ref{knuth-L0}}
        \EndIf
    \EndFor
    \State{$\varname{k} \writeop i$}\label{knuth-L3}
    \State{\textbf{critical section}}
    \If{$i = 0$}
        \State{$\varname{k} \writeop N{-}1$}
    \Else  
        \State{$\varname{k} \writeop i {-} 1$}
    \EndIf
    \State{$\varidx{control}{i} \writeop 0$}\label{knuth-L4}
    \end{algorithmic}
\end{algorithm}

The goal of Knuth's algorithm is to improve upon Dijkstra's result by having starvation freedom in addition to deadlock freedom. This goal is indeed accomplished with atomic registers, as we see in \autoref{tab:results}. Knuth makes no claims on this algorithm's behaviour with non-atomic registers. However, it is interesting that there is a difference between the behaviour with safe and regular registers.

{\Atrace} violating deadlock freedom for two threads and safe registers is as follows:
\begin{itemize}
    \item Thread 1 starts setting $\varidx{control}{1}$ to $1$
    \item Thread 0 starts the competition; since $\varname{k} = 0$, it goes to line 12 and sets $\varidx{control}{0}$ to $2$. When it reads $\varidx{control}{1}$ on line 14, it reads a $2$, which has never been written. Hence, it has to return to line 1, and hence start writing to $\varidx{control}{0}$.
    \item Thread 1 finishes line 1 and goes to line 2, where it has to read $\varidx{control}{0}$. It reads $\varidx{control}{0} = 0$, not the new or old value, and hence goes to line 7. There, since $N - 1 = 1$, it goes to line 12 and sets $\varidx{control}{1}$ to $2$. It then goes to line 13, where it finds that $\varidx{control}{0} = 2$, due to the overlap with thread 0's write. Hence, it has to return to line 1, where it starts writing to $\varidx{control}{1}$ again.
\end{itemize}
At this point, we are back where we started, with both threads at the beginning of the algorithm. We see there is a loop where both threads continually make progress, but always get sent back on line 15.
Unlike with the Attiya-Welch algorithm, we cannot fix this by merely ensuring that a register is only written to when its value would change.

\subsection{Lamport's 3-bit algorithm}
We already discussed Lamport's 1-bit algorithm, but in \cite{Lamport86Mutex2} he also presented the 3-bit algorithm (\autoref{alg:lamport3bit}), which was designed to improve on the 1-bit algorithm by making it starvation-free.
It works for arbitrary $N$.
Every thread $i$ has three Booleans: $\varidx{x}{i}, \varidx{y}{i}$ and $\varidx{z}{i}$, all initialised to $\false$.

This algorithm uses some auxiliary operations. 
The operation $\mathrm{ORD}(S)$ returns for a set $S \subseteq \TID$ a list of all elements in $S$ ordered from smallest to largest. Note that such an ordering is defined for $\TID$ since, \hyperlink{threadorder}{as we stated in} \autoref{sec:verification}, we use natural numbers to represent thread identifiers. This list is formalised as an increasing function $\gamma:\{1,\dots,M\} \rightarrow S$, where $M=|S|$. We write $\textrm{domain}(\gamma)$ for $\{1,\ldots,M\}$ and $\textrm{range}(\gamma)$ for $S$.
Additionally, there is the Boolean function $\mathrm{CG}(v, l)$, where $v: \{1,\dots,M\} \rightarrow \mathbbm{B}$ is a Boolean function mapping an index in $\{1,\dots,M\}$ (denoting an element of $S$) to either $\true$ or $\false$, and $l \in \{1,\dots,M\}$.
\begin{align*}
   \mathrm{CG}(v,l) \defeq&~  (v(l) = \mathrm{CGV}(v, l)) \\
   \mathrm{CGV}(v, l) \defeq&\begin{cases}
        \neg v(l{-}1) &\text{if $l > 1$}\\
        v(M) &\text{if $l = 1$}
        \end{cases}
\end{align*}

We write ``\textbf{for} $i$ \textbf{from} $j$\ \textbf{cyclically to} $k$'' to mean an iteration that begins with $i = j$, then increments $i$ by 1, taking the result modulo $N$. The iteration terminates when $i = k$, without executing the loop with $i = k$.
We use $\oplus$ for addition modulo $N$.

The definitions of $\mathrm{ORD}$, $\mathrm{CG}$ and $\mathrm{CGV}$ given above differ from those given by Lamport in \cite{Lamport86Mutex2}. He works with \emph{cycles}, rather than lists, and defines $\mathrm{CG}$ and $\mathrm{CGV}$ based on indexed elements of those cycles, rather than using the indices directly as we do. We believe that our presentation is equivalent to the one given by Lamport, but find it more straightforward to explain. This presentation also aligns better with how the algorithm would be implemented, and indeed aligns more closely with our mCRL2 model.

In addition to these changes in presentation, we also made one change to the algorithm itself: line~\ref{l-new} was added by us to emphasise that the mapping $\zeta$ is a snapshot of the variables $\varidx{z}{i}$ for all $i \in \textrm{range}(\gamma)$, and that the registers themselves do not get repeatedly read on line~\ref{lamport-3bit-functioncall}.
For this, we introduce the operation $\mathrm{SAVE_{\varname{z}}^{\gamma}}$ that creates a mapping from $\textrm{domain}(\gamma)$ to the Booleans, such that for all $h \in \textrm{domain}(\gamma)$, $\mathrm{SAVE^{\gamma}_{\varname{z}}}(h) = \varidx{z}{\gamma(h)}$.
Lamport does not state explicitly in \cite{Lamport86Mutex2} that the $\varname{z}$ registers must not be read repeatedly when computing the minimum, but if the registers are non-atomic and re-read during the computation it is possible that no minimum will be found.
This observation is also made in \cite{spronck2023process}.

\begin{algorithm}[ht!]
\caption{Lamport's 3-bit algorithm}
\label{alg:lamport3bit}
\begin{algorithmic}[1]
  \State{$\varidx{y}{i} \writeop \true$} 
  \State{$\varidx{x}{i} \writeop \true$} \label{l1}
  \State{$\gamma \writeop \mathrm{ORD}\{j\mid \varidx{y}{j} = \true\}$} \label{l2}
  \State{$\zeta \writeop \mathrm{SAVE}^{\gamma}_{\varname{z}}$}\label{l-new}
  \State{$f\writeop \gamma(\mathrm{minimum}\{h\in\textrm{domain}(\gamma)\mid \mathrm{CG}(\zeta,h)=\true\})$}\label{lamport-3bit-functioncall}
  \For{$j$ \textbf{from} $f$\ \textbf{cyclically to} $i$}
  \If{$\varidx{y}{j} = \true$}
    \If{$\varidx{x}{i} = \true$} {$\varidx{x}{i} \writeop \false$}\EndIf
  \State{\textbf{goto} line~\ref{l2}}
  \EndIf
  \EndFor
  \If{$\varidx{x}{i} = \false$} {\textbf{goto} line~\ref{l1}} \EndIf
  \For{$j$ \textbf{from} $i\oplus 1$ \textbf{cyclically to} $f$}
  \If{$\varidx{x}{j} = \true$} {\textbf{goto} line \ref{l2}}\EndIf
  \EndFor
  \State{\textbf{critical section}}
  \If{$\varidx{z}{i} = \true$}
    \State{$\varidx{z}{i} \writeop \false$}
  \Else 
    \State{$\varidx{z}{i} \writeop \true$}
  \EndIf
  \State{$\varidx{x}{i} \writeop \false$}
  \State{$\varidx{y}{i} \writeop \false$}
 \end{algorithmic}
\end{algorithm}

\vspace{-10pt}
\subsection{Peterson's algorithm}
Peterson's algorithm is already given in \autoref{sec:verification}. We merely note here that an execution showing that Peterson's algorithm does not satisfy mutual exclusion with non-atomic registers is given in \cite{spronck2023process}.

\subsection{Szymanski's flag algorithm}
We discussed Szymanski's 3-bit linear wait algorithm from \cite{Szy90} in \autoref{sec:Szymanski}.
Earlier, in \cite{Szy88}, Szymanski presented the flag algorithm, which we here give as \autoref{alg:szymanski-flag}.
Note that we give the pseudocode with a minor fix of an obvious typo: on line~\ref{alg:Szy-flag:exit}, we use a $\lor$ instead of a $\land$. 
In its original presentation, the flag algorithm uses a single integer variable $\varname{flag}$ per thread ranging over $\{0, 1, 2, 3,4\}$, initially $0$. Just like the 3-bit algorithm, it is designed for arbitrary $N$.

\begin{algorithm}[hb!]
\caption{Szymanski's flag algorithm}\label{alg:Szy-flag}
\begin{algorithmic}[1]
  \State{$\varidx{flag}{i} \writeop 1$}
  \State{\textbf{await} $\forall_j:\ \varidx{flag}{j} < 3$}
  \State{$\varidx{flag}{i} \writeop 3$}
  \If{$\exists_j:\ \varidx{flag}{j}=1$}
    \State{$\varidx{flag}{i} \writeop 2$}
    \State{\textbf{await} $\exists_j:\ \varidx{flag}{j}=4$}
  \EndIf
  \State{$\varidx{flag}{i} \writeop 4$}
  \State{\textbf{await} $\forall_{j < i}:\ \varidx{flag}{j} < 2$}
  \State{\textbf{critical section}}
  \State{\textbf{await} $\forall_{j>i}:\ \varidx{flag}{j} < 2 \lor \varidx{flag}{j} > 3$}\label{alg:Szy-flag:exit}
  \State{$\varidx{flag}{i} \writeop 0$}
\end{algorithmic}
\label{alg:szymanski-flag}
\end{algorithm}

It is claimed that, while the integer version of the algorithm is not correct with non-atomic registers, by converting the integers to three Booleans each, $\varname{door\_in}$, $\varname{door\_out}$ and $\varname{intent}$, the algorithm can be made correct for non-atomic registers.
According to \cite{Szy88}, this translation should be done according to \autoref{tab:szy-flag2bits}.
This results in \autoref{alg:Szy-flag-bits}. The three Booleans are all initialised as $\false$.

\begin{table}[hb!]
\caption{Translating the register $\mathit{flag}$ to three Boolean registers $\mathit{intent}$, $\mathit{door\_in}$ and $\mathit{door\_out}$.}
\label{tab:szy-flag2bits}
\centering
\begin{tabular}{|c|c|c|c|}
\hline
\textit{flag} & \textit{intent} & \textit{door\_in} & \textit{door\_out} \\ \hline
0             & $\false$              & $\false$                & $\false$                 \\ \hline
1             & $\true$               & $\false$                & $\false$                 \\ \hline
2             & $\false$              & $\true$                 & $\false$                  \\ \hline
3             & $\true$               & $\true$                 & $\false$                  \\ \hline
4             & $\true$               & $\true$                 & $\true$                  \\ \hline
\end{tabular}
\end{table}

\begin{algorithm}[hb!]
\caption{Szymanski's flag algorithm implemented with Booleans}\label{alg:Szy-flag-bits}
\begin{algorithmic}[1]
  \State{$\varidx{intent}{i} \writeop \true$}\label{szy-bits-1}
  \State{\textbf{await} $\forall_j:\ \varidx{intent}{j} = \false \lor \varidx{door\_in}{j} = \false$}\label{szy-bits-2}
  \State{$\varidx{door\_in}{i}\writeop \true$}\label{szy-bits-3}
  \If{$\exists_j:\ \varidx{intent}{j} = \true \land \varidx{door\_in}{j} = \false$}\label{szy-bits-4}
    \State{$\varidx{intent}{i} \writeop \false$}\label{szy-bits-5}
    \State{\textbf{await} $\exists_j:\ \varidx{door\_out}{j} = \true$}\label{szy-bits-6}
  \EndIf
  \If{$\varidx{intent}{i} = \false$} {$\varidx{intent}{i} \writeop \true$}\EndIf
  \State{$\varidx{door\_out}{i} \writeop \true$}\label{szy-bits-9}
  \State{\textbf{await} $\forall_{j < i}:\ \varidx{door\_in}{j} = \false$}\label{szy-bits-10}
  \State{\textbf{critical section}}\label{szy-bits-11}
  \State{\textbf{await} $\forall_{j>i}:\ \varidx{door\_in}{j} = \false \lor \varidx{door\_out}{j} = \true$}\label{szy-bits-12}
  \State{$\varidx{intent}{i} \writeop \false$}\label{szy-bits-13}
  \State{$\varidx{door\_in}{i} \writeop \false$}\label{szy-bits-14}
  \State{$\varidx{door\_out}{i}\writeop \false$}\label{szy-bits-15}
\end{algorithmic}
\label{alg:szymanski-flag-bits}
\end{algorithm}

We analysed both variants in \cite{spronck2023process}, and found that neither algorithm is correct with non-atomic registers. Executions demonstrating these violations are provided there.
Additionally, we observed in \cite{spronck2023process} that the Boolean variant of the algorithm violates mutual exclusion even with atomic registers. We observed that the violating execution reported by mCRL2, given in full in \cite{spronck2023process}, relies on the exit protocol resetting the three Booleans in the order $\varname{intent} - \varname{door\_in} - \varname{door\_out}$. This is the order that is suggested by the final pseudocode provided by Szymanski in \cite[Figure 3]{Szy88}. We posited in \cite{spronck2023process} that the violation of mutual exclusion could be fixed by changing the order to be $\varname{door\_out} -\varname{intent} - \varname{door\_in}$ instead. Now that we can verify the liveness properties as well, we can confirm that this is indeed true: if this alternate exit protocol is used, the Boolean version of the flag protocol is equivalent to the integer version with respect to the properties verified in this paper.

\end{document}